\documentclass[12pt]{article}
\usepackage{bm,amsmath,amsthm,amssymb,paralist,graphicx,paralist,setspace}
\usepackage[round,sort&compress]{natbib}
\usepackage{authblk}
\usepackage{enumitem}
\usepackage{graphicx}
\usepackage{color}

\usepackage[letterpaper,left=0.8in,right=0.8in,top=0.75in,bottom=0.75in]{geometry}

\allowdisplaybreaks[1]

\usepackage{subfigure}
\usepackage{threeparttable}
\usepackage{booktabs}
\usepackage[colorlinks]{hyperref}

\DeclareMathOperator{\lasso}{\,LASSO}
\DeclareMathOperator{\en}{\,EN}
\DeclareMathOperator{\gl}{\,GL}

\DeclareMathOperator{\pc}{PC}
\DeclareMathOperator{\kggc}{KGGC}
\DeclareMathOperator{\lars}{\,lars}

\DeclareMathOperator{\SLOG}{\,SLOG}
\DeclareMathOperator{\Lai}{\,Lai}
\DeclareMathOperator{\rSLOG}{\,rSLOG}

\DeclareMathOperator{\CD}{\,CD}
\DeclareMathOperator{\Diag}{Diag}

\DeclareMathOperator{\sth}{sth}
\DeclareMathOperator*{\argmin}{arg\,min}
\DeclareMathOperator*{\argmax}{arg\,max}
\DeclareMathOperator{\sign}{sign}
\DeclareMathOperator{\rank}{rank}

\newcommand{\reals}{\mathbb{R}}
\newcommand{\oneto}[1]{\{1,\ldots,#1\}}

\newcommand{\bbeta}{\bm{\beta}}
\newcommand{\bh}{\hat{\beta}}

\newcommand{\bhl}[1]{\hat{\beta}^{\lasso}_{\,#1}}
\newcommand{\bbhl}{\hat{\bm{\beta}}^{\lasso}}

\newcommand{\simiid}{\stackrel{\text{iid}}{\sim}}
\newcommand{\simind}{\stackrel{\text{ind}}{\sim}}

\newcommand{\normof}[1]{\left|\left|#1\right|\right|}

\newcommand{\bbhen}{\hat{\bm{\beta}}^{\en}}
\newcommand{\bbhgl}{\hat{\bm{\beta}}^{\gl}}

\usepackage{setspace}
\theoremstyle{plain}
\newtheorem{thm}{Theorem}
\newtheorem{lem}[thm]{Lemma}

\theoremstyle{definition}

\theoremstyle{remark}
\newtheorem*{rem}{Remark}

\title{%
\vspace*{-3em}%
Lasso Regression: Estimation and Shrinkage via Limit of Gibbs Sampling}

\author[1*]{Bala Rajaratnam}
\author[2]{Steven Roberts}
\author[1]{Doug Sparks}
\author[1]{Onkar Dalal}
\affil[1]{Stanford University}
\affil[2]{Australian National University}
\affil[*]{Department of Statistics, Stanford University Stanford, CA 94305

brajarat@stanford.edu}

\begin{document}

\maketitle
\thispagestyle{empty}

\begin{abstract}
The application of the lasso is espoused in high-dimensional settings where only a small number of the regression coefficients are believed to be nonzero (i.e., the solution is sparse). Moreover, statistical properties of high-dimensional lasso estimators are \textcolor{black}{often} proved under the assumption that the correlation between the predictors is bounded.
In this vein,
coordinatewise methods, the most common means of computing the lasso solution, naturally work well
in the presence of low to moderate multicollinearity.
The computational \textcolor{black}{speed of} coordinatewise algorithms, while
excellent
for sparse
and low to moderate multicollinearity settings,
degrades as sparsity decreases and multicollinearity increases. Though lack of sparsity and high multicollinearity can be quite common in contemporary applications, model selection is still a necessity in such settings. Motivated by the limitations of coordinatewise algorithms in \textcolor{black}{such} ``non-sparse'' and ``high-multicollinearity'' settings, we propose the novel \textcolor{black}{``Deterministic Bayesian Lasso"} algorithm for computing the lasso solution. This algorithm is developed by considering a limiting version of the Bayesian lasso. In contrast to coordinatewise algorithms, the performance of the
Deterministic Bayesian Lasso
improves as sparsity decreases and
multicollinearity
increases. Importantly, in non-sparse and \textcolor{black}{high-multicollinearity} settings the proposed algorithm can offer substantial increases in computational speed over coordinatewise algorithms. A rigorous theoretical analysis demonstrates that
\color{black}
(1)~%
the Deterministic Bayesian Lasso algorithm converges to the lasso solution,
\color{black}
and (2)~it leads to a representation of the lasso estimator
which shows how it
achieves both $\ell_1$ and $\ell_2$ types of shrinkage simultaneously. Connections between the
Deterministic Bayesian Lasso
and \textcolor{black}{other algorithms are also provided.
The benefits of the Deterministic Bayesian Lasso algorithm are then illustrated on simulated and real data.}

\end{abstract}


\setcounter{page}{1}
\section{Introduction}
The process of estimating regression parameters subject to a penalty on the $\ell_1$-norm of the parameter estimates, known as the lasso \citep{tibshirani1996}, has become ubiquitous in modern statistical applications. In particular, in settings of low to moderate multicollinearity where the solution is believed to be sparse, the application of the lasso is almost \emph{de rigueur}. Outside of the sparse and low to moderate multicollinearity setting the performance of the lasso is suboptimal \citep{HZTH}. In this vein, many of the theoretical and algorithmic developments for the lasso assume and/or cater to a sparse estimator in the presence of low to moderate multicollinearity. A prime example of this phenomenon is coordinatewise algorithms, which have become the most common means of computing the lasso solution. The performance of coordinatewise algorithms, while ideal for sparse and low to moderate correlation settings, degrades as sparsity decreases and multicollinearity increases.
However, the model selection capabilities of the lasso can still be essential even in the presence of high multicollinearity or in the absence of sparsity.  The limitations of coordinatewise algorithms in such settings motivate us to propose
in this paper
the novel Deterministic Bayesian Lasso algorithm for computing the lasso solution.
The performance of this proposed algorithm improves as sparsity \textcolor{black}{decreases} and multicollinearity increases, and hence our approach offers substantial advantages over coordinatewise techniques in such settings.

The popularity of the lasso comes despite the inability to express the lasso estimator in any convenient closed form.
Hence, there is keen interest
in algorithms capable of efficiently computing the lasso solution \citep{efron2004, friedman2007,osborne2000lasso}. Arguably, the two most well known algorithms for computing the lasso solution are \emph{least angle regression} \citep{efron2004} and
the even faster
\emph{pathwise coordinate optimization} \citep{friedman2007}. Least angle regression (LARS) can be viewed as a form of stagewise regression. By exploiting the geometry of the lasso problem, LARS is able to efficiently compute the entire sequence of lasso solutions. Pathwise coordinate optimization is based on the idea of cycling through the coefficients and minimizing the objective function ``one coefficient at a time", while holding the other coefficients fixed. \textcolor{black}{Since} it has been shown to be considerably faster than competing methods, including LARS \citep{JFTHRT}, pathwise coordinate optimization is today the most commonly utilized algorithm for computing lasso solutions. While pathwise coordinate optimization is generally a fast and efficient algorithm for computing the lasso solution, the algorithm is not without limitations. In particular, the computational speed of pathwise coordinate optimization degrades as sparsity decreases and multicollinearity increases \citep{JFTHRT}. 

In addition to the efficient computation of the lasso solution, the development of methods for quantifying the uncertainty associated with lasso coefficient estimates has proved difficult \citep{park2008}. \textcolor{black}{The difficulty primarily relates to assigning measures of uncertainty to (exact) zero lasso coefficient estimates.} The recently developed Bayesian lasso \citep{park2008} addresses this issue
by
natural and economical uncertainty quantification, in the form of posterior credible intervals. \textcolor{black}{The Bayesian lasso is based on the observation of \citet{tibshirani1996} that the lasso can be interpreted as a maximum a~posteriori Bayesian procedure under a double-exponential prior.} In their development of the Bayesian lasso, \citet{park2008} expressed the double exponential prior as a mixture of normals and 
derived a Gibbs sampler for generating from the posterior.

In this paper, we exploit the structure of the Bayesian lasso and its corresponding Gibbs sampler, not for uncertainty quantification, but rather for the computation of the lasso point estimate itself.  Our approach is predicated upon the role played by the sampling variance in the lasso problem, commonly denoted
by~$\sigma^2$.
Importantly, the lasso objective function does not depend on~$\sigma^2$, and hence neither does the lasso solution. The sampling variance~$\sigma^2$ does, however, play a role in the Bayesian lasso posterior. The value of~$\sigma^2$  essentially controls the spread of the posterior around its mode. Hence, if $\sigma^2$ is small, the posterior will be tightly concentrated around its mode, \textcolor{black}{and thus} is close to the lasso solution. This
implies
that the (Bayesian lasso) Gibbs sampler with a small, fixed value of~$\sigma^2$ will yield a sequence that is tightly concentrated around the lasso solution. We also note that: (1) the lasso solution is \emph{exactly} the mode of the marginal posterior of the regression coefficients, and (2) the mode of the joint posterior of the
\textcolor{black}{regression coefficients and hyperparameters}
used by the Gibbs sampler differs from the lasso solution by a distance proportional to~$\sigma^2$.
For computation of the lasso point estimate, the relevance of the discussion in the \textcolor{black}{immediately} preceding paragraph is realized by the fact that in the limit as $\sigma^2\to0$ the Gibbs sampler reduces to a deterministic sequence. Moreover, the limit of this deterministic sequence \textcolor{black}{can be shown to be} the lasso solution. This realization motivates
our
Deterministic Bayesian Lasso
algorithm for computing the lasso point estimate.

A rigorous theoretical analysis demonstrates that
\color{black}
(1)~the
Deterministic Bayesian Lasso
converges to the lasso solution with probability 1
\color{black}
and, (2)~it leads to a representation of the lasso estimator that demonstrates how it achieves both $\ell_1$ and $\ell_2$ types of shrinkage simultaneously. Connections between the
Deterministic Bayesian Lasso
and the EM algorithm, and modifications to the
Deterministic Bayesian Lasso
for the purposes of computing other lasso-like estimators are also provided.
%
We also study the connections between our proposed algorithm and Iteratively Re-weighted Least Squares, an approach motivated by optimization.
\color{black}
The probabilistic underpinning of our proposed methodology provides, (1)~a theoretical backing for our proposed procedure and, (2)~a means of avoiding certain technical difficulties that optimization methods in the literature have to contend with.
\color{black}

\color{black}
Further, it will be demonstrated, via simulation and real data analysis, that in non-sparse and/or high-multicollinearity settings the Deterministic Bayesian Lasso has computational
advantages over coordinatewise algorithms. Such non-sparse and high-multicollinearity settings are highly prevalent in high dimensional and big data applications. \color{black} In the remainder of the paper, in reference to its motivation, and for brevity, we shall
interchangeably refer to the Deterministic Bayesian Lasso framework \textcolor{black}{by its acronym
SLOG: \emph{Shrinkage via Limit of Gibbs Sampling}.}

\color{black}
We note that one of the goals of the paper is to obtain a faster means of calculating the lasso solution in high multicollinearity and/or low sparsity settings. There are of course other computationally fast methods for high dimensional regression including ``2-step methods" such as thresholding and then regressing (``marginal regression"), or Bayesian variants such as SSVS. Besides the lasso, these  ``2-step methods" methods are also useful, and have their respective strengths. One of the primary advantages of the lasso is that the chance of bringing in many predictors, which have low predictive power in the presence of other covariates, is relatively less.\color{black}

\section{%
Methodology
}
\label{sec:theory}


\subsection{%
The Lasso and the
Bayesian Lasso Posterior}

Our developments assume that we are in the standard regression model setting with a length-$n$ response vector $\bm{y}$ that is centered ($\bm{1}_n^T\bm{y}=0$) and has distribution
\begin{align*}
\bm{y}\sim N_n(\bm{X}\bbeta,\sigma^2\bm{I}_n),
\end{align*}
where
$\bm{X}$ is the $n\times p$ design matrix and
without loss of generality
$\sigma^2>0$ is assumed known.
We assume that the columns of~$\bm{X}$ have been standardized such that $\sum_{i=1}^n X_{ij}^2=n$ for each $j\in\oneto{p}$.  The frequentist lasso estimator \citep{tibshirani1996} of the coefficient vector~$\bbeta$ is
\begin{align}
\bbhl=\argmin_{\bbeta\in\reals^p}\left(||\bm{y}-\bm{X}\bbeta||_2^2+2\lambda||\bbeta||_1\right),
\label{lasso-def}
\end{align}
where $||\cdot||_r$ denotes the usual $L_r$ vector norm and $\lambda>0$ is the regularization parameter.
Here it should be noted that if $\rank(\bm{X})=p$, then the minimization in~(\ref{lasso-def}) is strictly convex, and hence $\bbhl$ is unique.  However, if $\rank(\bm{X})<p$, which for example is necessarily true when $p>n$, then there may be uncountably many solutions which achieve the minimization in~(\ref{lasso-def}), i.e., $\bbhl$ may not be uniquely defined.  Nevertheless, uniqueness of~$\bbhl$ can still be obtained when $\rank(\bm{X})<p$ under quite mild conditions.
For instance, \citet{tibshirani2013} showed that $\bbhl$ is unique if the columns of~$\bm X$ are in a state called
``general position.''
In turn, a simple sufficient condition for the columns of~$\bm X$ to be in general position
is that the entries of~$\bm{X}$
are
drawn from a distribution
that is
absolutely continuous with respect to Lebesgue measure
on~$\reals^{n\times p}$~~\citep{tibshirani2013}.
We
will henceforth make the
assumption that
the columns of~$\bm X$ are in general position
(henceforth referred to as Assumption~1),
which
in turn
implies that $\bbhl$ is unique.
%
%
%
%
Note that
another
consequence of
Assumption~1
is that the solution to the lasso problem for any subset of the columns of~$\bm{X}$ is also unique,
since the columns in the subset are also in general position.

The lasso estimator~$\bbhl$ may be interpreted from a Bayesian perspective as
\begin{align*}
\bbhl=\argmax_{\bbeta\in\reals^p}\;\pi(\bbeta\mid\bm{y}),
\end{align*}
where $\pi(\bbeta\mid\bm{y})$ is the posterior distribution of~$\bbeta$ under the Bayesian model
\begin{align*}
\bm{y}\mid\bbeta&\sim N_n(\bm{X}\bbeta,\sigma^2\bm{I}_n)\\
\beta_j&\simiid\text{DoubleExp}(\lambda/\sigma^2).
\end{align*}
Note that the double exponential distribution may be expressed as a scale mixture of normals \citep[e.g.,][]{andrews1974}.  Hence, the Bayesian model above may be rewritten as the hierarchical model
\begin{align}
\bm{y}\mid\bbeta&\sim N_n(\bm{X}\bbeta,\sigma^2\bm{I}_n)\notag\\
\beta_j\mid\omega_j&\simind N(0,\sigma^2\omega_j)\label{lasso-model}\\
\omega_j&\simiid\text{Exp}(\lambda^2/2\sigma^2),\notag
\end{align}
which is popularly referred to as the Bayesian lasso~\citep{park2008}.
Here it should be explicitly noted that our hierarchy appears to differ from that of \citeauthor{park2008}.  However, it can be seen that the two representations are in fact equivalent by noting that our regularization parameter~$\lambda$ and the regularization parameter $\lambda^{\pc}$ of \citeauthor{park2008} are related according to $\lambda=\lambda^{\pc}\sqrt{\sigma^2}$, and we take $\sigma^2$ as known.
Under our model~(\ref{lasso-model}), the joint posterior is then
\begin{align}
\pi(\bbeta,\bm{\omega}\mid\bm{y})&\propto
f(\bm{y}\mid\bbeta)\;\left[\prod_{j=1}^p\pi(\beta_j\mid\omega_j)\right]\;\left[\prod_{j=1}^p\pi(\omega_j)\right]\notag\\
&\propto\exp\left(-\frac{1}{2\sigma^2}||\bm{y}-\bm{X}\bbeta||_2^2\right)\;
\left[\prod_{j=1}^p\omega_j^{-1/2}\right]
\exp\left(-\frac{1}{2\sigma^2}\sum_{j=1}^p\beta_j^2\omega_j^{-1}\right)\notag\\
&\qquad\times\exp\left(-\frac{\lambda^2}{2\sigma^2}\sum_{j=1}^p\omega_j\right)\notag\\
&\propto\exp\left[-\frac{1}{2\sigma^2}\left(
\normof{\bm{y}-\bm{X}\bbeta}_2^2+\sum_{j=1}^p\beta_j^2\omega_j^{-1}+\lambda^2\sum_{j=1}^p\omega_j+\sigma^2\sum_{j=1}^p\log\omega_j\right)\right].
\label{posterior}
\end{align}
For convenience, let $g(\bbeta,\bm{\omega})$ denote the quantity in parentheses in
the last line of~(\ref{posterior}).

\citet{park2008} used the joint posterior~(\ref{posterior}) to derive a Gibbs sampler for drawing from the joint lasso
posterior. The
convergence properties of such a sequence were
subsequently
investigated by \citet{kyung2010}.
The above
Gibbs sampler cycles through
the conditionals
\begin{align}
\bbeta\mid\bm{\omega},\bm{y}&\stackrel{\text{iid}}{\sim}N_p\left[\left(\bm{X}^T\bm{X}+\bm{\Omega}^{-1}\right)^{-1}\bm{X}^T\bm{y},\;
\sigma^2\left(\bm{X}^T\bm{X}+\bm{\Omega}^{-1}\right)^{-1}\right],
\label{conditional-beta}\\
\omega_j^{-1}\mid\bbeta,\bm{y}&\stackrel{\text{iid}}{\sim}\begin{cases}
\text{InverseGaussian}(\lambda/|\beta_j|,\;\lambda/\sigma^2) & \text{ if }\beta_j\ne0,\\
\text{InverseGamma}(1/2,\;\lambda^2/2\sigma^2) & \text{ if }\beta_j=0,
\end{cases}
\label{conditional-omega}
\end{align}
where $\bm{\Omega}=\Diag(\omega_1,\ldots,\omega_p)$, and where we may replace $\left(\bm{X}^T\bm{X}+\bm{\Omega}^{-1}\right)^{-1}$ by the alternative expression $\bm{\Omega}^{1/2}\left(\bm{I}_p+\bm{\Omega}^{1/2}\bm{X}^T\bm{X}\bm{\Omega}^{1/2}\right)^{-1}\bm{\Omega}^{1/2}$ whenever an element of $\bm{\omega}$ is zero.

\subsection{The Deterministic Bayesian Lasso Algorithm}

The Gibbs sampler of \citet{park2008}
was motivated by its ability to provide credible intervals for the lasso estimates.
However,
we discovered that
the particular form of the conditionals~(\ref{conditional-beta})~and~(\ref{conditional-omega}) that comprise the Gibbs sampler suggests a novel method for calculating the lasso
point estimate
itself.  Specifically, notice that as $\sigma^2\to0$, the conditional distribution of $\bbeta\mid\bm{\omega},\bm{y}$ given in~(\ref{conditional-beta}) converges to degeneracy at its mean $(\bm{X}^T\bm{X}+\bm{\Omega}^{-1})^{-1}\bm{X}^T\bm{y}$.  Similarly, the conditional distribution of $\omega_j\mid\bbeta,\bm{y}$ given in~(\ref{conditional-omega}) converges to degeneracy at $\lambda^{-1}|\beta_j|$
\color{black}
in both the $\beta_j=0$ and $\beta_j\ne0$ cases.
For the $\beta_j=0$ case, note that if $U_m\sim\text{InverseGamma}(1/2,\,c_m)$ with $c_m\to\infty$ as $m\to\infty$, then $U_m\to\infty$ in probability as $m\to\infty$, which in turn implies that $U_m^{-1}\to0$ as $m\to\infty$.  (The $\beta_j\ne0$ case is clear from the properties of the inverse Gaussian distribution.)
\color{black}

Thus, in the limit as $\sigma^2\to0$, the Bayesian lasso Gibbs sampler reduces to a deterministic sequence~$\{(\bm{b}^{(k)},\bm{w}^{(k)}):k\ge0\}$ given by
\begin{align*}
w^{(k)}_j&=\lambda^{-1}|b^{(k)}_j|
\qquad\text{ for each }j\in\{1,\ldots,p\},
\\
\bm{b}^{(k+1)}&=(\bm{W}^{(k)})^{1/2}\left[
\bm{I}_p
+(\bm{W}^{(k)})^{1/2}\bm{X}^T\bm{X}(\bm{W}^{(k)})^{1/2}\right]^{-1}(\bm{W}^{(k)})^{1/2}\bm{X}^T\bm{y},
\end{align*}
where $\bm{W}^{(k)}=\Diag(w^{(k)}_1,\ldots,w^{(k)}_p)$, and where $\bm{b}^{(0)}$ is some specified starting point.  Substituting the form of $\bm{w}^{(k)}$ into the equation for $\bm{b}^{(k+1)}$ yields
\begin{align}
\bm{b}^{(k+1)}=(\bm{B}^{(k)})^{1/2}\left[\lambda\bm{I}_p+(\bm{B}^{(k)})^{1/2}\bm{X}^T\bm{X}(\bm{B}^{(k)})^{1/2}\right]^{-1}(\bm{B}^{(k)})^{1/2}\bm{X}^T\bm{y},\label{sequence-zeros}
\end{align}
where $\bm{B}^{(k)}=\Diag(|b_1^{(k)}|,\ldots,|b_p^{(k)}|)$.  Note that if every component of~$\bm{b}^{(k)}$ is nonzero, then we may replace~(\ref{sequence-zeros}) by the simpler representation
\begin{align}
\bm{b}^{(k+1)}
&=
\left[\bm{X}^T\bm{X}+\lambda\left(\bm{B}^{(k)}\right)^{-1}\right]^{-1}\bm{X}^T\bm{y}.\label{sequence}
\end{align}
Suppose the starting point~$\bm{b}^{(0)}$ is drawn randomly from some distribution~$P_0$ on~$\reals^p$, where $P_0$ is absolutely continuous with respect to Lebesgue measure on~$\reals^p$.  Then under mild regularity conditions,
$\bm b^{(k)}\to\bbhl$ as $k\to\infty$ with $P_0$-probability~$1$,
\color{black}
where $P_0$-probability simply denotes probability under the distribution~$P_0$ from which the starting point $\bm b^{(0)}$ is drawn.
\color{black}
This result
will be shown in Section~\ref{subsec:properties}.
Thus, the recursive sequence given by~(\ref{sequence-zeros}) or~(\ref{sequence}),
which we call the SLOG algorithm,
provides a straightforward method of calculating~$\bbhl$ that holds
\color{black}
regardless of the values of $p$ and $n$. From~\eqref{sequence} it is observed that each iteration of SLOG requires the inversion of a $p \times p$ matrix. This inversion can become unduly time consuming in high dimensions. In Section~\ref{sec:rslog} a variant of SLOG is developed that successfully overcomes this problem. This variant, termed rSLOG, is able to rapidly reduce the size of the matrix that needs inverting at each iteration of SLOG.
\color{black}

Essentially, the SLOG
algorithm
may be interpreted as providing the $\bbeta$ components of a Gibbs sampler in its degenerate limit as $\sigma^2\to0$.
Some intuition for this connection may be gained by noting that the lasso estimator does not depend on the value of~$\sigma^2$.  Thus, for the purposes of finding the lasso estimator, the value of~$\sigma^2$ may be taken as
any value that
may be convenient.  Now observe from the form of the joint posterior~(\ref{posterior}) that the smaller the value of~$\sigma^2$, the more concentrated the posterior is around its mode.  (It should be noted that the lasso estimator is the mode of the marginal posterior, and the modes of the joint and marginal posteriors do not coincide.  However, they do coincide in their limits as $\sigma^2\to0$.)  Thus, the Gibbs sampler can be made arbitrarily closely concentrated around the lasso solution by taking the value of $\sigma^2$ small enough.  The SLOG algorithm simply carries this idea to its limiting conclusion by ``sampling'' directly from the degenerate limits of the conditional distributions. \color{black} An annealing type variant to SLOG, where the cycles of the Gibbs sampler are based on a decreasing $\sigma^2$ sequence, is investigated in Supplemental Section~\ref{sec:AppE}. \color{black}

\subsection{%
Alternative Representations
and Fixed-Point Results}
\label{subsec:properties}

In this section, we provide theoretical results to justify the use of the 
SLOG
algorithm for calculation of~$\bbhl$.

\subsubsection{Alternative Representation}

The
Deterministic Bayesian Lasso
algorithm has already been written in both a general form~(\ref{sequence-zeros}) and a simpler form~(\ref{sequence}), with the simpler form only applicable in the absence of components that are exactly zero.  (Note that the relevant issue is zeros in the components~$b_j^{(k)}$ of the sequence~$\bm{b}^{(k)}$ generated by the SLOG
algorithm.
Zeros in the components of~$\bbhl$ itself are irrelevant.)
In fact, we will eventually show in Lemma~\ref{lem:no-zeros} that we may use the simpler form~(\ref{sequence}) for all $k\ge0$ with $P_0$-probability~1.  However, a more general form that can be applied for any point in $\reals^p$ will still be useful.  The following lemma introduces a somewhat more intuitive representation of~(\ref{sequence-zeros}).  We first define some additional notation.  For each $k\ge0$,
let $\gamma^{(k)}=\{j:b^{(k)}_j\ne0\}$
denote the set of indices of the nonzero components of $\bm b^{(k)}$,
and let $\bm{X}_\star$ denote the matrix formed by retaining the $j$th column of~$\bm{X}$ if and only if $j\in\gamma^{(k)}$.  Similarly, let $\bm{b}_\star^{(k)}$ be the vector formed by retaining the $j$th element of~$\bm{b}^{(k)}$ if and only if $j\in\gamma^{(k)}$, and let $\bm{B}^{(k)}_\star$ be the diagonal matrix with the absolute values of the elements of~$\bm{b}^{(k)}_\star$ on its diagonal.  Also, let
$\bm{b}_{\star\star}^{(k+1)}$
denote the vector formed by retaining the $j$th element of~$\bm{b}^{(k+1)}$ if and only if $j\in\gamma^{(k)}$.  (Note that
$\bm{b}_{\star\star}^{(k+1)}$
is selected according to~$\gamma^{(k)}$, not $\gamma^{(k+1)}$.)  Then we have the following result.

\begin{lem}\label{lem:partition}
$
\bm{b}_{\star\star}^{(k+1)}
=\left[\bm{X}_\star^T\bm{X}_\star+\lambda(\bm{B}_\star^{(k)})^{-1}\right]^{-1}\bm{X}_\star^T\bm{y}$,
and $b^{(k+1)}_j=0$ for each $j\notin\gamma^{(k)}$.
\end{lem}

\begin{proof}
For convenience, we assume without loss of generality that
$\gamma^{(k)}=\{m+1,\ldots,p\}$ for some $0\le m\le p$.
Now observe that
\begin{align*}
&\left[\lambda\bm{I}_p+\left(\bm{B}^{(k)}\right)^{1/2}\bm{X}^T\bm{X}\left(\bm{B}^{(k)}\right)^{1/2}\right]^{-1}\\
&\qquad=\begin{bmatrix}
\lambda\bm{I}_m & \bm{0}_{m\times(p-m)} \\
\bm{0}_{(p-m)\times m} & \left[\lambda\bm{I}_{p-m}+\left(\bm{B}_\star^{(k)}\right)^{1/2}\bm{X}_\star^T\bm{X}_\star\left(\bm{B}_\star^{(k)}\right)^{1/2}\right]
\end{bmatrix}^{-1}\\
&\qquad=\begin{bmatrix}
\lambda^{-1}
\bm{I}_m & \bm{0}_{m\times(p-m)} \\
\bm{0}_{(p-m)\times m} & \left[\lambda\bm{I}_{p-m}+\left(\bm{B}_\star^{(k)}\right)^{1/2}\bm{X}_\star^T\bm{X}_\star\left(\bm{B}_\star^{(k)}\right)^{1/2}\right]^{-1}
\end{bmatrix},
\end{align*}
from which it follows that the recursion relation~(\ref{sequence-zeros}) may be written as
\begin{align*}
\bm{b}^{(k+1)}&=
\begin{bmatrix}
\bm{0}_{m\times m} & \bm{0}_{m\times(p-m)} \\
\bm{0}_{(p-m)\times m} & \left(\bm{B}_\star^{(k)}\right)^{1/2}
\end{bmatrix}
\begin{bmatrix}
\color{black}
\lambda^{-1}
\color{black}
\bm{I}_m & \bm{0}_{m\times(p-m)} \\
\bm{0}_{(p-m)\times m} & \left[\lambda\bm{I}_{p-m}+\left(\bm{B}_\star^{(k)}\right)^{1/2}\bm{X}_\star^T\bm{X}_\star\left(\bm{B}_\star^{(k)}\right)^{1/2}\right]^{-1}
\end{bmatrix}\\
&\qquad\qquad\qquad\qquad\qquad\qquad\times\begin{bmatrix}
\bm{0}_{m\times m} & \bm{0}_{m\times(p-m)} \\
\bm{0}_{(p-m)\times m} & \left(\bm{B}_\star^{(k)}\right)^{1/2}
\end{bmatrix}
\bm{X}^T\bm{y}\\
&=\begin{bmatrix}
\bm{0}_m\\
\left(\bm{B}_\star^{(k)}\right)^{1/2}\left[\lambda\bm{I}_{p-m}+\left(\bm{B}_\star^{(k)}\right)^{1/2}\bm{X}_\star^T\bm{X}_\star\left(\bm{B}_\star^{(k)}\right)^{1/2}\right]^{-1}\left(\bm{B}_\star^{(k)}\right)^{1/2}\bm{X}_\star^T\bm{y}
\end{bmatrix}
\end{align*}
and the result follows from the fact that $\bm{B}_\star^{(k)}$ is invertible.
\end{proof}

\begin{rem}
Lemma~\ref{lem:partition} establishes that a modified version of the simpler form~(\ref{sequence}) can still be used even in the presence of zeros in the components of $\bm b^{(k)}$.  Any such zero components simply remain zero in the next iteration.  Meanwhile, the nonzero components are updated by applying the simpler form~(\ref{sequence}) using the subvector of these nonzero components and the submatrix of the corresponding columns of $\bm X$.
\end{rem}

\subsubsection{Fixed Points}

We now establish results on fixed points of the
SLOG
algorithm.
To this end, it will be helpful to consider the recursion relation as a function.  Specifically, let $A:\reals^p\to\reals^p$ be the function that maps $\bm{b}^{(k)}$ to $\bm{b}^{(k+1)}$ according to~(\ref{sequence-zeros}), or equivalently Lemma~\ref{lem:partition}.

Suggestions
of the relationship between the sequence $\{\bm{b}^{(k)}:k\ge0\}$ and the lasso estimator are provided by the following lemmas.  The first states that the lasso's objective function, which we define to be
$Q(\bm{b})=-\normof{\bm{y}-\bm{X}\bm{b}}_2^2-2\lambda\normof{\bm{b}}_1$,
is nondecreasing
as a function of~$k$ when evaluated at each iteration of the SLOG sequence~$\bm{b}^{(k)}$,
while the second uses this result to conclude that the lasso estimator is a fixed point of the recursion
under broad conditions.

\begin{lem}\label{lem:monotone}
$Q\!\left[A(\bm{b})\right]\ge Q(\bm{b})$ for all $\bm{b}\in\reals^p$, with strict inequality if $A(\bm{b})\ne\bm{b}$.
Moreover, $Q\!\left[\bm b^{(k)}\right]$ converges as $k\to\infty$.
\end{lem}

\begin{proof}
Let $\bm{b}\in\reals^p$, and define $\bm{B}=\Diag(|b_1|,\ldots,|b_p|)$.  Also, for convenience, define $\bm{c}=A(\bm{b})$ and $\bm{C}=\Diag(|c_1|,\ldots,|c_p|)$.
Observe that by Lemma~\ref{lem:partition}, we may write $Q(\bm{b})$ and~$Q(\bm{c})$ as
\begin{align*}
Q(\bm{b})&=-\normof{\bm{y}-\bm{X}_\star\bm{b}_\star}_2^2-2\lambda\normof{\bm{b}_\star}_1,\\
Q(\bm{c})&=-\normof{\bm{y}-\bm{X}_\star\left(\bm{X}_\star^T\bm{X}_\star+\lambda\bm{B}_\star^{-1}\right)^{-1}\bm{X}_\star^T\bm{y}}_2^2-2\lambda\normof{\left(\bm{X}_\star^T\bm{X}_\star+\lambda\bm{B}_\star^{-1}\right)^{-1}\bm{X}_\star^T\bm{y}}_1,
\end{align*}
where $\bm{b}_\star$ denotes the nonzero components of $\bm{b}$ and where $\bm{B}_\star$ is the corresponding positive definite diagonal matrix (analogous to the definition of $\bm{b}^{(k)}_\star$ and $\bm{B}^{(k)}_\star$ from $\bm{b}^{(k)}$ and $\bm{B}^{(k)}$).
Then
\begin{align*}
Q(\bm{c})-Q(\bm{b})
&=-\bm{y}^T\bm{X}_\star\left(\bm{X}_\star^T\bm{X}_\star+\lambda\bm{B}_\star^{-1}\right)^{-1}\bm{X}_\star^T\bm{X}_\star\left(\bm{X}_\star^T\bm{X}_\star+\lambda\bm{B}_\star^{-1}\right)^{-1}\bm{X}_\star^T\bm{y}\\
&\qquad+2\bm{y}^T\bm{X}_\star\left(\bm{X}_\star^T\bm{X}_\star+\lambda\bm{B}_\star^{-1}\right)^{-1}\bm{X}_\star^T\bm{y}
-2\lambda\normof{\left(\bm{X}_\star^T\bm{X}_\star+\lambda\bm{B}_\star^{-1}\right)^{-1}\bm{X}_\star^T\bm{y}}_1\\
&\qquad+\bm{b}_\star^T\bm{X}_\star^T\bm{X}_\star\bm{b}_\star-2\bm{b}_\star^T\bm{X}_\star^T\bm{y}+2\lambda\normof{\bm{b}_\star}_1\\
&=\bm{y}^T\bm{X}_\star\left(\bm{X}_\star^T\bm{X}_\star+\lambda\bm{B}_\star^{-1}\right)^{-1}\bm{X}_\star^T\bm{y}
-2\lambda\normof{\left(\bm{X}_\star^T\bm{X}_\star+\lambda\bm{B}_\star^{-1}\right)^{-1}\bm{X}_\star^T\bm{y}}_1\\
&\qquad+\bm{y}^T\bm{X}_\star\left(\bm{X}_\star^T\bm{X}_\star+\lambda\bm{B}_\star^{-1}\right)^{-1}\left(\lambda\bm{B}_\star^{-1}\right)\left(\bm{X}_\star^T\bm{X}_\star+\lambda\bm{B}_\star^{-1}\right)^{-1}\bm{X}_\star^T\bm{y}\\
&\qquad+\bm{b}_\star^T\bm{X}_\star^T\bm{X}_\star\bm{b}_\star-2\bm{b}_\star^T\bm{X}_\star^T\bm{y}+2\lambda\normof{\bm{b}_\star}_1\\
&
=
\bm{c}_{\star\star}^T
\bm{X}_\star^T\bm{y}-2\lambda\normof{
\bm{c}_{\star\star}
}_1
+\lambda
\bm{c}_{\star\star}^T
\bm{B}_\star^{-1}
\bm{c}_{\star\star}
+\bm{b}_\star^T\bm{X}_\star^T\bm{X}_\star\bm{b}_\star-2\bm{b}_\star^T\bm{X}_\star^T\bm{y}+2\lambda\normof{\bm{b}_\star}_1,
\end{align*}
where $
\bm{c}_{\star\star}
=\left(\bm{X}_\star^T\bm{X}_\star+\lambda\bm{B}_\star^{-1}\right)^{-1}\bm{X}_\star^T\bm{y}$.
Now note that we may write $\normof{\bm{b}_\star}_1=\bm{b}_\star^T\bm{B}_\star^{-1}\bm{b}_\star$ since each element of $\bm{b}_\star$ is nonzero, and hence
\begin{align*}
Q(\bm{c})-Q(\bm{b})
&=
\bm{c}_{\star\star}^T
\bm{X}_\star^T\bm{y}
+\bm{b}_\star^T\left(\bm{X}_\star^T\bm{X}_\star+\lambda\bm{B}_\star^{-1}\right)\bm{b}_\star
-2\bm{b}_\star^T\bm{X}_\star^T\bm{y}\\
&\qquad+\lambda\left(
\bm{c}_{\star\star}^T
\bm{B}_\star^{-1}
\bm{c}_{\star\star}
-2\normof{
\bm{c}_{\star\star}
}_1+\normof{\bm{b}_\star}_1\right)\\
&=
\bm{c}_{\star\star}^T
\left(\bm{X}_\star^T\bm{X}_\star+\lambda\bm{B}_\star^{-1}\right)
\bm{c}_{\star\star}
+\bm{b}_\star^T\left(\bm{X}_\star^T\bm{X}_\star+\lambda\bm{B}_\star^{-1}\right)\bm{b}_\star
-2\bm{b}_\star^T\left(\bm{X}_\star^T\bm{X}_\star+\lambda\bm{B}_\star^{-1}\right)
\bm{c}_{\star\star}
\\
&\qquad+\lambda\sum_{j\,:\,b_j\ne0}\left(c_j^2|b_j|^{-1}-2|c_j|+|b_j|\right)\\
&=(
\bm{c}_{\star\star}
-\bm{b}_\star)^T\left(\bm{X}_\star^T\bm{X}_\star+\lambda\bm{B}_\star^{-1}\right)(
\bm{c}_{\star\star}
-\bm{b}_\star)
+\lambda\sum_{j\,:\,b_j\ne0}|b_j|^{-1}\left(|c_j|-|b_j|\right)^2\ge0,
\end{align*}
which establishes the first result.  To obtain the second result, note that $A(\bm{b})\ne\bm{b}$ is equivalent to $
\bm{c}_{\star\star}
\ne\bm{b}_\star$, noting that for each~$j$ such that $b_j=0$, we necessarily have $c_j=0=b_j$ as well.  Then the strict inequality follows immediately from the fact that
$\bm{X}_\star^T\bm{X}_\star+\lambda\bm{B}_\star^{-1}$ is positive definite.
To obtain
convergence of the sequence $Q\!\left[\bm b^{(k)}\right]$,
simply combine the first result with the fact that $Q\!\left[\bm b^{(0)}\right]\le Q\!\left[\bm b^{(k)}\right]\le0$ for all~$k\ge0$.
\end{proof}

\begin{lem}\label{lem:lasso-fixed}
Let $\bm X$ be drawn from a distribution that is absolutely continuous with respect to Lebesgue measure on $\reals^{n\times p}$.  Then
$A(\bbhl)=\bbhl$.
\end{lem}

\begin{proof}
Note from Lemma~\ref{lem:monotone} that $Q\!\left[A(\bm{b})\right]\ge Q(\bm{b})$ for all~$\bm{b}\in\reals^p$, and recall that $\bbhl=\argmax_{\bm{b}\in\reals^p}Q(\bm{b})$ by definition.  Then $Q\!\left[A(\bbhl)\right]=Q(\bbhl)$.
Observe that $\bbhl$ is the unique maximizer of~$Q$ by
the condition on $\bm X$.
It follows that $A(\bbhl)=\bbhl$.
\end{proof}

Lemma~\ref{lem:lasso-fixed}
above establishes
that the lasso estimator
is
a fixed point of the recursion~$A$ that maps $\bm{b}^{(k)}$ to $\bm{b}^{(k+1)}$.  It is natural to ask whether there exist other fixed points for this recursion.  The following lemma answers this question in the affirmative.  For the sake of clarity, we temporarily introduce somewhat more cumbersome notation for the lasso estimator.  Namely, we will explicitly indicate the dependence of $\bbhl$ on $\bm{y},\bm{X}$, and $\lambda$ by writing $\bbhl(\bm{y},\bm{X},\lambda)$ to mean precisely~(\ref{lasso-def}).

\begin{lem}\label{lem:other-fixed}
$A(\bm{b})=\bm{b}$ if and only if the vector~$\bm{b}_\star$ of the nonzero components of~$\bm{b}$ satisfies $\bm{b}_\star=\bbhl(\bm{y},\bm{X}_\star,\lambda)$, where $\bm{X}_\star$ is the matrix formed by retaining the columns of~$\bm{X}$ corresponding to the elements of~$\bm{b}$ retained in~$\bm{b}_\star$.
\end{lem}

\begin{proof}
By Lemma~\ref{lem:partition}, $A(\bm{b})=\bm{b}$ is equivalent to $\bm{b}_\star=\left(\bm{X}_\star^T\bm{X}_\star+\lambda\bm{B}_\star^{-1}\right)^{-1}\bm{X}_\star^T\bm{y}$, where $\bm{B}_\star$ is the diagonal matrix with the absolute values of the elements of~$\bm{b}_\star$ as its diagonal entries.  Then simply rewrite this as $\bm{X}_\star^T\left(\bm{y}-\bm{X}_\star\bm{b}_\star\right)=\lambda\sign\bm{b}_\star$, which may be recognized as the Karush--Kuhn--Tucker condition for the lasso problem using only the covariates in~$\bm{X}_\star$ (more precisely, as the case of this condition when all components of the possible solution are nonzero).  Thus, $A(\bm{b})=\bm{b}$ if and only if $\bm{b}_\star=\bbhl(\bm{y},\bm{X}_\star,\lambda)$.
\end{proof}

Lemma~\ref{lem:other-fixed} has several consequences.  First, it may be seen that $A(\bm{0}_p)=\bm{0}_p$.  Second, since the lasso solution for each subset of the columns of~$\bm{X}$ is unique by
Assumption~1,
there are at most $2^p$ fixed points of~$A$.  (In fact, there are fewer than $2^p$ fixed points whenever some components of $\bbhl$ are already zero.)  Third, every fixed point of~$A$ has at least one zero component, except for possibly $\bbhl$ itself (if each of its components is nonzero).

\section{%
Convergence Analysis%
}

In this section, we establish under mild regularity conditions that, with $P_0$-probability~$1$,
the sequence generated by the
Deterministic Bayesian Lasso
algorithm converges to~$\bbhl$.
We begin by stating and proving two lemmas that
motivate a
simplifying assumption.

\begin{lem}\label{lem:lasso-zero}
If $\bm{X}^T\bm{y}=\bm{0}_p$, then $\bbhl=\bm{0}_p$.
\end{lem}

\begin{proof}
If $\bm{X}^T\bm{y}=\bm{0}_p$, then
$\normof{\bm{y}-\bm{X}\bbeta}_2^2+2\lambda\normof{\bbeta}_1
=\normof{\bm{y}}_2^2+\normof{\bm{X}\bbeta}_2^2+2\lambda\normof{\bbeta}_1$,
which is clearly minimized by $\bbeta=\bm{0}_p$.
\end{proof}

\begin{lem}\label{lem:separable}
Suppose the columns of $\bm{X}$ may be permuted and partitioned as $\bm{X}=\left[\bm{X}_{(1)}\;\;\bm{X}_{(2)}\right]$, where $\bm{X}_{(1)}^T\bm{X}_{(2)}^{}$ equals the zero matrix of the appropriate size.  Then\\
$\bbhl(\bm{y},\bm{X},\lambda)=\left[\bbhl(\bm{y},\bm{X}_{(1)},\lambda),\;\bbhl(\bm{y},\bm{X}_{(2)},\lambda)\right]$.
\end{lem}

\begin{proof}
The proof is given in the Supplemental Section.
\end{proof}

The
point
of Lemma~\ref{lem:separable}
is that when the covariates may be permuted and partitioned into sets $\bm{X}_{(1)}$ and $\bm{X}_{(2)}$ that are uncorrelated with each other, then solving the lasso problem for $\bm{X}$ is equivalent to solving the lasso problem for $\bm{X}_{(1)}$ and $\bm{X}_{(2)}$ separately and combining the solutions.  With this result in mind, we now
assume that
for any permutation and partition $\left[\bm{X}_{(1)}\;\;\bm{X}_{(2)}\right]$ of the columns of~$\bm{X}$ such that $\bm{X}_{(1)}^T\bm{X}_{(2)}^{}$ is the zero matrix, both $\bm{X}_{(1)}^T\bm{y}$ and $\bm{X}_{(2)}^T\bm{y}$ are nonzero
(henceforth referred to as Assumption~2).
This assumption is
not restrictive,
as
can be seen from the preceding lemmas.  If
Assumption~2
did not hold, then the problem could be split into finding $\bbeta_{(1)}$ and $\bbeta_{(2)}$ separately by Lemma~\ref{lem:separable}, and one of these solutions would be exactly zero by Lemma~\ref{lem:lasso-zero}.  Thus, the effect of
Assumption~2
is merely to ensure that we are not attempting to solve a problem that can be trivially reduced to a simpler one, though it will also be needed to avoid a technical difficulty in proving the following useful result.

\begin{lem}\label{lem:no-zeros}
Under
Assumption~2,
$P_0(b_j^{(k)}=0\text{ for some }k\ge0\text{ and }j\in\oneto{p})=0.$
\end{lem}

\begin{proof}
The proof is given in the Supplemental Section.
\end{proof}

We now state and prove the following result,
which states that our SLOG algorithm converges to the lasso estimator.

\begin{thm}\label{thm:convergence}
Under
Assumptions~1~and~2,
$\bm b^{(k)}\to\bbhl$
as $k\to\infty$ with $P_0$-probability~$1$.
\end{thm}

\begin{proof}
The proof is
long and technical and is therefore
given in the Supplemental Section.
\end{proof}

It should be remarked that the only purpose of the random starting point is to ensure that with $P_0$-probability~1, our sequence avoids ``accidentally'' landing exactly on a fixed point other than $\bbhl$.  If a rule could be obtained by which the starting point could be chosen to avoid such a possibility, then we could choose the starting point by this rule
instead.

%

\section[Properties of the Deterministic Bayesian Lasso Algorithm]{%
Properties of the Deterministic Bayesian Lasso\\Algorithm}

\subsection{%
Connections to Other Methods%
}
\subsubsection{EM Algorithm}

An alternative interpretation of the
Deterministic Bayesian Lasso
algorithm may be obtained by comparing it to
the EM algorithm \citep{dempster1977}.
Recall that the
Deterministic Bayesian Lasso is based on a Gibbs sampler that includes both the parameter of interest~$\bm\beta$ and a latent variable~$\bm\omega$.  An EM algorithm for the same parameter~$\bm\beta$ and latent variable~$\bm\omega$ can be considered in which the log-likelihood is
\begin{align*}
\ell(\bbeta,\bm{\omega})=-\frac{1}{2\sigma^2}g(\bbeta,\bm{\omega}),
\end{align*}
where $g(\bm\beta,\bm\omega)$ again denotes the quantity in parentheses in the last line of~(\ref{posterior}).  The iterates of the resulting EM algorithm coincide with those of the SLOG algorithm, as we now demonstrate below.

First,
suppose that the value of~$\bbeta$ at the $k$th step of the EM algorithm is $\bm{b}^{(k)}$.  For the E-step of the EM algorithm, we obtain a function $h(\bbeta;\bm{b}^{(k)})$ defined by
\begin{align}
h(\bbeta;\bm{b}^{(k)})=E_\star\left[\ell(\bbeta,\bm{\omega})\right]
=-\frac{1}{2\sigma^2}E_\star\left[g(\bbeta,\bm{\omega})\right],\label{h1}
\end{align}
where $E_\star$ denotes an expectation taken with respect to the distribution~$P_\star$ where $\bbeta$ is fixed and $\bm{\omega}$ has the
distribution
\begin{align*}
\omega_j^{-1}\stackrel{\text{iid}}{\sim}\begin{cases}
\text{InverseGaussian}(\lambda/|b_j^{(k)}|,\;\lambda/\sigma^2) & \text{ if }b_j^{(k)}\ne0,\\
\text{InverseGamma}(1/2,\;\lambda^2/2\sigma^2) & \text{ if }b_j^{(k)}=0.
\end{cases}\tag{$P_\star$}
\end{align*}
Note that
the above
coincides with the distribution of $\bm\omega^{-1}\mid\bm\beta,\bm y$ with $\bm\beta=\bm b^{(k)}$ under the Bayesian model in~(\ref{conditional-omega}).
Then~(\ref{h1}) becomes
\begin{align}
h(\bbeta;\bm{b}^{(k)})
&=-\frac{1}{2\sigma^2}\left[\normof{\bm{y}-\bm{X}\bbeta}_2^2
+E_\star\left(\sum_{j=1}^p\beta_j^2\omega_j^{-1}\right)+E_\star\left(\lambda^2\sum_{j=1}^p\omega_j+\sigma^2\sum_{j=1}^p\log\omega_j\right)\right]
\notag\\
&=-\frac{1}{2\sigma^2}\left[\normof{\bm{y}-\bm{X}\bbeta}_2^2
+\sum_{j=1}^p E_\star\left(\beta_j^2\omega_j^{-1}\right)+c_0(\bm{b}^{(k)})\right]
,\label{h2}\end{align}
where $c_0(b^{(k)})$ does not depend on $\bbeta$.  The expectation~$E_\star\left(\beta_j^2\omega_j^{-1}\right)$ may be evaluated as
\begin{align*}
E_\star(\beta_j^2\omega_j^{-1})=\begin{cases}0&\text{ if }\beta_j=0,\\
\dfrac{\lambda\beta_j^2}{|b_j^{(k)}|}&\text{ if }\beta_j\ne0
\text{ and }b_j^{(k)}\ne0,\\
\infty&\text{ if }\beta_j\ne0\text{ and }b_j^{(k)}=0,
\end{cases}
\end{align*}
noting that the last case holds by the fact that $E_\star(\omega_j^{-1})=\infty$ when $|b_j^{(k)}|=0$ since the shape parameter of the inverse gamma distribution in $P_\star$ is $1/2$.  Then due to this last case, $h(\bbeta;\bm{b}^{(k)})=-\infty$ whenever $b_j^{(k)}=0$ and $\beta_j\ne0$.  Now consider the M-step of the EM algorithm, which takes
\begin{align*}
\bm{b}^{(k+1)}=\argmax_{\bbeta\in\reals^p}\;h(\bbeta;\bm{b}^{(k)}).
\end{align*}
If $b^{(k)}_j=0$ for some $j$, then
$b^{(k+1)}_j=0$ as well, since otherwise $h(\bm{b}^{(k+1)};\bm{b}^{(k)})=-\infty$ and the maximum is not obtained.  (Note that $h(\bm{0}_p;\bm{b}^{(k)})>-\infty$, so a value greater than $-\infty$ is clearly obtainable.)  Then the M-step essentially  maximizes the function $h(\bbeta;\bm{b}^{(k)})$ subject to the restriction that $\beta_j=0$ for every $j$ such that $b^{(k)}_j=0$.
For convenience, we now assume without loss of generality that $b^{(k)}_j=0$ for each $j\in\oneto{m}$ and $b^{(k)}_j\ne0$ for each $j\in\{m+1,\ldots,p\}$, where $0\le m\le p$.  Also, partition $\bm{X}$ as
\begin{align*}
\bm{X}=\begin{bmatrix}\bm{X}^{0} & \bm{X}_\star\end{bmatrix},
\end{align*}
where $\bm{X}^0$ is $n\times m$ and $\bm{X}_\star$ is $n\times(p-m)$, and let $\bm{B}_\star^{(k)}=\Diag(|b^{(k)}_{m+1}|,\ldots,|b^{(k)}_p|)$, noting that $\bm{B}_\star^{(k)}$ is invertible.
Then we have that
\begin{align*}
\bm{b}^{(k+1)}=\begin{bmatrix}\bm{0}_m\\
\displaystyle\argmax_{\bbeta_\star\in\reals^{p-m}}\;h_\star(\bbeta_\star;\bm{b}^{(k)})
\end{bmatrix},
\end{align*}
where
\begin{align*}
h_\star(\bbeta_\star;\bm{b}^{(k)})=-\frac{1}{2\sigma^2}\left[\normof{\bm{y}-\bm{X}_\star\bbeta_\star}_2^2+\lambda\bbeta_\star^T(\bm{B}_\star^{(k)})^{-1}\bbeta_\star+c_0(\bm{b}^{(k)})\right].
\end{align*}
Then
\begin{align*}
\argmax_{\bbeta_\star\in\reals^{p-m}}\;h_\star(\bbeta_\star;\bm{b}^{(k)})
&=\argmin_{\bbeta_\star\in\reals^{p-m}}\left[\bbeta_\star^T\bm{X}_\star^T\bm{X}_\star\bbeta-2\bbeta_\star^T\bm{X}_\star^T\bm{y}+\lambda\bbeta_\star^T(\bm{B}_\star^{(k)})^{-1}\bbeta_\star\right]\\
&=\argmin_{\bbeta_\star\in\reals^{p-m}}\left\{(\bbeta_\star-\tilde{\bbeta}_\star)^T\left[\bm{X}_\star^T\bm{X}_\star+\lambda(\bm{B}_\star^{(k)})^{-1}\right]
(\bbeta_\star-\tilde{\bbeta}_\star)\right\}=\tilde{\bbeta}_\star,
\end{align*}
where $\tilde{\bbeta}_\star=[\bm{X}_\star^T\bm{X}_\star+\lambda(\bm{B}_\star^{(k)})^{-1}]^{-1}\bm{X}_\star^T\bm{y}$.  Hence,
\begin{align*}
\bm{b}^{(k+1)}&=\begin{bmatrix}\bm{0}_m\\
\left[\bm{X}_\star^T\bm{X}_\star+\lambda\left(\bm{B}_\star^{(k)}\right)^{-1}\right]^{-1}\bm{X}_\star^T\bm{y}
\end{bmatrix}
\end{align*}
which is precisely the
form of the SLOG update as shown in
Lemma~\ref{lem:partition}.

\begin{rem}
\color{black}
In light of the connections between the SLOG algorithm and the EM algorithm, it may be asked why the EM algorithm has not thus far been central to the exposition of our proposed methodology.  First, note that the EM framework provides no motivation for the particular form of the augmentation that is employed by the SLOG algorithm.  Moreover,
\color{black}
the similarity of the SLOG algorithm and the EM algorithm does not necessarily mean that we can invoke the various results on
convergence of the EM algorithm
that have appeared in the literature \citep[e.g.,][]{wu1983}
to claim convergence of SLOG.
The nondifferentiable penalty term imposed by the lasso leads to problems with certain regularity conditions that are typically required to apply such EM algorithm convergence results.  In particular, condition~(10) of \citet{wu1983} fails for the Bayesian lasso.  It should also be noted that
a recent paper by
\citet{chretien2012}
considered
extensions of the EM algorithm for which convergence can be demonstrated even under a nondifferentiable penalty term.  However, the resulting sequence does not necessarily coincide with the iterates of the SLOG algorithm, and hence the convergence proofs of \citet{chretien2012} are not applicable here.
Hence, the formal proof of convergence of SLOG as given by Theorem~\ref{thm:convergence} is not redundant.
\end{rem}


\subsubsection{Iterative Re-weighted Least Squares}
\label{sec:Lai}

A closely
related problem of minimizing the $\ell_1$ norm of $\bm\beta$ under the linear constraint $\bm X\bm\beta = \bm y$ has
also
been studied
in the literature
\citep{daubechies2010,chartrand2008,candes2007}. In \citet{daubechies2010}, an iteratively re-weighted least squares (IRLS) algorithm was proposed to solve the
linearly
constrained
$\ell_1$
minimization. This algorithm updates $\bm\beta_{k+1}$ by solving a weighted least squares problem with weights $w_j=(\beta_{k,j}^2+\epsilon_k^2)^{-1/2}$ (computed elementwise), where $\epsilon_k$ is a sequence of small positive numbers introduced to avoid the possibility of division by zero. The algorithm is shown to converge under a so-called null space condition, a slightly weaker version of the more commonly imposed restricted isometry property
of \citet{candes2005}.
In addition to proving convergence of the algorithm, \citet{daubechies2010} also establish sufficient conditions for the limit of the algorithm to exhibit a specified degree of sparsity.

The results of \citet{daubechies2010} were
recently
extended by \citet{lai2013} to solve the lasso problem specifically. Similarly to \citet{daubechies2010}, their algorithm also updates $\bm\beta_{k+1}$ by solving a re-weighted least squares problems approximated with a
sequence of small constants $\epsilon_k$
to avoid infinite weights.
The
$\ell_1$
term in the lasso objective function is replaced by
the
approximation
\[
\|\bm\beta\|_{1,\epsilon} = \sum_{j=1}^{p} \left( \beta_j^2 + \epsilon^2 \right)^{1/2}.
\]
(Note that the approximation above is actually the $q=1$ case of their algorithm.)
The
iteration update for
$\bm\beta_{k+1}$
in terms of
$\bm\beta_{k}$
is given by
\[
\displaystyle\left(\bm X^{T}\bm X + \Diag\!\left[\frac{\lambda}{\left(\epsilon_k^2 + \beta_{k1}^2\right)^{1/2}},\ldots,\frac{\lambda}{\left(\epsilon_k^2 + \beta_{k1}^2\right)^{1/2}}\right]\right)\bm\beta_{k+1} = \bm X^{T}\bm y
\]
where the series of $\epsilon_k$ is chosen adaptively to boost speed of convergence.
In particular, $\epsilon_k=\min\{\epsilon_{k-1},\,\alpha\, r(\bm\beta_k)_{h+1}\}$. The value $r(\bm\beta_k)_{h+1}$ is the
$(h+1)$th--largest
parameter estimate
(in magnitude)
at iteration $k$. They further replace the subspace condition in \citet{daubechies2010} by the restricted isometry property (RIP) of certain order to prove preliminary results on convergence, error bound, and local convergence behavior.

The primary motivation for using the sequence $\epsilon_k$ in the IRLS algorithm for approximating the subproblems is to avoid encountering infinite weights due to division by 0. However, for the SLOG algorithm we prove that the $\epsilon$-approximations are not necessary and that
all coefficients at all iterations are nonzero with probability~$1$.
Thus we can avoid the additional complexity and errors introduced by the $\epsilon$-approximations of IRLS.

We also briefly compare
the Lai algorithm
to SLOG
for accuracy and time
(see Supplemental Section~\ref{sec:AppB}).
The results showed that in settings of high sparsity that the SLOG algorithm afforded increases in computational speed without a loss of accuracy compared to the Lai algorithm. In settings of low sparsity it was observed that, for the same level of accuracy, that SLOG offered similar computational speed compared to Lai. The findings described here were consistent across settings of both high and low multicollinearity.

Moreover, the SLOG algorithm differs from the work of \citet{lai2013} in several key respects. First, at a conceptual level, SLOG has a probabilistic interpretation as a limit of Gibbs samplers, which are well understood and commonly employed in Bayesian
inference.
Second, the assumptions made by SLOG are weaker, i.e., no restricted isometry property or null space property is required for
convergence. \textcolor{black}{
Recall that a sufficient condition for the restricted isometry property to hold is if the entries of the design matrix are independently and identically distributed sub-Gaussian random variables \citep{baraniuk2008simple,devore2009instance}. The condition on the design matrix for convergence of SLOG is however much weaker: it is sufficient for the entries of the design matrix to have a joint distribution that is absolutely continuous with respect to Lebesgue measure on~$\reals^{n\times p}$.}
Third,
the algorithms
fundamentally
differ in their approach to the presence of zeros in the coefficient
paths.
These zeros would lead to infinite weights on the following iteration for both algorithms.
The approach of \citet{lai2013} has to make an explicit allowance for this fact by choosing a series of
$\epsilon$'s,
whereas we instead use a theoretical approach to show that these exact zeros almost surely do not occur.  We therefore show that this problem is a non-issue.
Moreover,
for the Lai algorithm, the $\epsilon$'s
are chosen using an estimate of sparsity~$h$ which may not be consistent with the provided penalty parameter $\lambda$. The SLOG algorithm does not require any such additional information.

\subsection{%
Detailed
Analysis in One Dimension}

The behavior of the recursive sequence $\left\{\bm{b}^{(k)}:k\ge1\right\}$ generated by the
Deterministic Bayesian Lasso
algorithm may be better understood through an explicit analysis of its behavior in the special case when $p=1$.
%
%
In this case,
the lasso estimator reduces to a simple soft-thresholding estimator, i.e.,
\begin{align}
\bhl{}=\frac{1}{n}\sth(\bm{X}^T\bm{y},\lambda)=\sth(\bh,\lambda/n),
\label{lasso-sth}
\end{align}
where $\sth(a,c)=(a-c\sign a)\;I(|a|>c)$ denotes the usual soft-thresholding function and $\bh=\bm{X}^T\bm{y}/n$ is the least squares estimator.  Then the sequence~(\ref{sequence}) becomes
\begin{align}
b^{(k+1)}=|b^{(k)}|\left(\lambda+n|b^{(k)}|\right)^{-1}n\bh=\frac{|b^{(k)}|\;\bh}{\lambda/n+|b^{(k)}|}.\label{sequence1}
\end{align}
In this case, it is possible to express the sequence
of SLOG iterates
in non-recursive form according to the following lemma.

\begin{lem}\label{lem:one-d-form}
When $p=1$, the sequence~(\ref{sequence1}) takes the form
\begin{align*}
b^{(k)}=\frac{c^k\lambda|b^{(0)}|\sign\bh}{\lambda+n|b^{(0)}|\sum_{m=0}^{k-1}c^m}
\end{align*}
for all $k\ge1$, where $c=n|\bh|/\lambda$.
\end{lem}

\begin{proof}
The proof is given in the Supplemental Section.
\end{proof}

The closed-form expression for $b^{(k)}$ that is provided by
Lemma~\ref{lem:one-d-form}
allows several of its properties to be seen clearly, as described by the following theorem.

\begin{thm}\label{thm:one-d-properties}
When $p=1$, the sequence~(\ref{sequence1}) satisfies the following properties:
\begin{enumerate}[label=(\roman*)]
\item\label{one-d-all-nonzero}
If at least one of $\bh$ or $b^{(0)}$ is zero, then $b^{(k)}=0$ for all $k\ge1$.  Otherwise, $b^{(k)}\ne0$ for all $k\ge1$, and moreover $\sign b^{(k)}=\sign\bh\ne0$ for all $k\ge1$.
\item\label{one-d-all-same-side}
The value of $\sign(b^{(k)}-\bhl{})$ is the same for all $k\ge1$.
\item\label{one-d-fixed-points}
The function that maps $b^{(k)}$ to $b^{(k+1)}$ has two (not necessarily distinct) fixed points: zero and $\bhl{}$.
\end{enumerate}
\end{thm}

\begin{proof}
The proof is given in the Supplemental Section.
\end{proof}

The closed-form expression for $b^{(k)}$ that is provided by Lemma~\ref{lem:one-d-form} also facilitates a rigorous statement of the convergence rate of this sequence to the lasso estimator.  This result is stated by the following theorem.

\begin{thm}\label{thm:one-d-convergence}
Assume $b^{(0)}\ne0$.  Then the sequence~(\ref{sequence1}) satisfies
\begin{align*}
\left|b^{(k)}-\bhl{}\right|\le\begin{cases}
\left(\dfrac{n|\bh|}{\lambda}\right)^k\left|b^{(0)}-\bhl{}\right|&\text{ if }n|\bh|<\lambda,\\
\dfrac{\lambda}{nk}\phantom{\dfrac{\dfrac{1}{1}}{\dfrac{1}{1}}}&\text{ if }n|\bh|=\lambda,\\
\left(\dfrac{\lambda}{n|\bh|}\right)^k\left|\dfrac{\bh}{b^{(0)}}\right|\left|b^{(0)}-\bhl{}\right|&\text{ if }n|\bh|>\lambda
\end{cases}
\end{align*}
for all $k\ge1$.
\end{thm}

\begin{proof}
The proof is given in the Supplemental Section.
\end{proof}

The
main message
of
Theorem~\ref{thm:one-d-convergence}
is that in one dimension, the SLOG sequence converges geometrically
fast (or ``linear convergence'' in optimization terminology)
to the lasso solution as long as $|\bh|\ne\lambda/n$.  (Note that $|\bh|=\lambda/n$ corresponds to the boundary between zero and non-zero values of the lasso solution.)


%

%

\subsection{Computational Complexity}
\label{sec:comp}

We now consider the computational
complexity of
the
Deterministic Bayesian Lasso
algorithm.  We also compare this to the complexity of
the popular
coordinatewise method of~\citet{friedman2007}.

\subsubsection{The Deterministic Bayesian Lasso Algorithm}

Begin by assuming that all $p$ coefficients from the previous iteration are nonzero.
Then each iteration of SLOG requires computation of the vector $(\bm{X}^T\bm{X}+\lambda\bm{B}^{-1})^{-1}\bm{X}^T\bm{y}$. Since $\bm{B}$ is diagonal $p\times p$, computation of~$\lambda\bm{B}^{-1}$ is $O(p)$.  Multiplication of the $p\times p$ matrix $(\bm{X}^T\bm{X}+\lambda\bm{B}^{-1})^{-1}$ by the $p\times1$ vector $\bm{X}^T\bm{y}$ is $O(p^2)$.  Hence, the key step is the inversion of the $p\times p$ matrix $\bm{X}^T\bm{X}+\lambda\bm{B}^{-1}$.  This step is $O(p^3)$ via na\"{i}ve matrix inversion.  If $n\ge p$, this cannot be improved upon without resorting to more sophisticated methods of matrix inversion \citep[e.g.,][]{coppersmith1990}.

However, if $n<p$, then some improvement is possible by noting that $\bm{X}^T\bm{X}=\sum_{i=1}^n\bm{x}_i^{}\bm{x}_i^T$, where $\bm{x}_i^T$ denotes the $i$th row of $\bm{X}$.  Hence, we are essentially inverting a ``rank-$n$ correction'' of the diagonal matrix $\lambda\bm{B}^{-1}$.  The method of \citet{miller1981} defines $p\times p$ matrices $\bm{A}_0=\lambda\bm{B}^{-1}$ and $\bm{A}_i=\bm{A}_0+\sum_{j=1}^i\bm{x}_j\bm{x}_j^T$ for each $i\in\{1,\ldots,n\}$.  Note that $\bm{A}_0$ is diagonal, but $\bm{A}_1,\ldots,\bm{A}_n$ are not (in general).  Then clearly $\bm{A}_0^{-1}=\lambda^{-1}\bm{B}$, and the remaining inverses are given by the recursive formula
\[
\bm{A}_i^{-1}=\bm{A}_{i-1}^{-1}-\frac{1}{1+\bm{x}_i^T\bm{A}_{i-1}^{-1}\bm{x}_i^{}}\left(\bm{A}_{i-1}^{-1}\bm{x}_i\right)\left(\bm{A}_{i-1}^{-1}\bm{x}_i\right)^T.
\]
Computation of the quadratic form $\bm{x}_i^T\bm{A}_{i-1}^{-1}\bm{x}_i$ is $O(p^2)$.  Similarly, the multiplication $\bm{A}_{i-1}^{-1}\bm{x}_i$ is $O(p^2)$, and the outer product $\bm{A}_{i-1}^{-1}\bm{x}_i\bm{A}_{i-1}^{-1}\bm{x}_i^T$ is $O(p^2)$ as well.  Thus, a single step of the \citet{miller1981} recursion is $O(p^2)$, and since $n$ such steps are needed to compute the desired matrix $\bm{A}_n^{-1}=(\bm{X}^T\bm{X}+\lambda\bm{B}^{-1})^{-1}$, the overall result is $O(n\,p^2)$.

Alternatively, direct application of the Woodbury matrix identity in the $n<p$ case yields
\begin{align*}
(\bm{X}^T\bm{X}+\lambda\bm{B}^{-1})^{-1}
&=\lambda^{-1}\bm{B}-(\lambda^{-1}\bm{B})\bm{X}^T\left[\bm{I}_n+\bm{X}(\lambda^{-1}\bm{B})\bm{X}^T\right]^{-1}\bm{X}(\lambda^{-1}\bm{B})\\
&=\lambda^{-1}\bm{B}-\lambda^{-1}\bm{B}\bm{X}^T\left(\lambda\bm{I}_n+\bm{X}\bm{B}\bm{X}^T\right)^{-1}\bm{X}\bm{B}.
\end{align*}
Since $\bm{B}$ is diagonal, computation of the $n\times n$ matrix $\bm{X}\bm{B}\bm{X}^T$ is $O(n^2\,p)$.  Then inversion of the $n\times n$ matrix $\lambda\bm{I}_n+\bm{X}\bm{B}\bm{X}^T$ is $O(n^3)$.  Finally, computation of the $p\times p$ matrix $\bm{X}^T(\lambda\bm{I}_n+\bm{X}\bm{B}\bm{X}^T)^{-1}\bm{X}$ is $O(n\,p^2)$, and multiplication of the result by $\bm{B}$ on the left and right is only $O(p^2)$, noting again that $\bm{B}$ is diagonal.  Thus, the overall result is again $O(n\,p^2)$, the same as under the method of \citet{miller1981}.


Now suppose that only $p_\star$ of the $p$ total coefficients from the previous iteration are nonzero.
Then the analysis above still holds if the $p\times p$ diagonal matrix~$\bm{B}$ and the $n\times p$ matrix~$\bm{X}$ are replaced by the $p_\star\times p_\star$ diagonal matrix~$\bm{B}_\star$ and the $n\times p_\star$ matrix~$\bm{X}_\star$.  Hence, a single iteration of the SLOG algorithm is
$O(p_\star^2\min\{n,p_\star\})$.

\subsubsection{Coordinate Descent Algorithm}
\label{sec:cdalg}

The pathwise coordinate optimization approach of \citet{friedman2007} successively recalculates each component~$b_j$ of the ``current'' coefficient vector~$\bm{b}$ according to
\begin{equation}
b_j=\sth_{\lambda}\!\left(\bm{x}_{(j)}^T\bm{y}-\sum_{k\ne j}\bm{x}_{(j)}^T\bm{x}_{(k)}^{}b_k\right),
\label{sth}
\end{equation}
where $\bm{x}_{(j)}$ denotes the $j$th column (not row) of $\bm{X}$ and $\sth_\lambda(\cdot)$ is defined by
\[
\sth_\lambda(u)=\begin{cases}u-\lambda&\text{ if }u>\lambda,\\0&\text{ if }-\lambda\le u\le\lambda,\\u+\lambda&\text{ if }u<-\lambda,\end{cases}
\]
i.e., $\sth_\lambda(\cdot)$ soft-thresholds its argument at~$\lambda$.  Note that the initial computation of the quantities $\bm{x}_{(j)}^T\bm{y}$ and $\bm{x}_{(j)}^T\bm{x}_{(k)}^{}$ is equivalent to the initial computation of $\bm{X}^T\bm{y}$ and $\bm{X}^T\bm{X}$ in the SLOG algorithm.  After this, (\ref{sth}) involves only scalar operations.  Thus, updating a single~$b_j$ is~$O(p)$, and a full iteration in which $b_1,\ldots,b_p$ are each updated once is~$O(p^2)$.


Now suppose that only $p_\star$ of the $p$ total coefficients of the ``current'' coefficient vector are nonzero, and suppose for simplicity
that any
effects due to changes in~$p_\star$ during a single iteration are negligible.  Then the sum in (\ref{sth}) includes at most $p_\star$ nonzero terms, and hence a full update of $b_1,\ldots,b_p$ is in fact only~$O(p\,p_\star)$.

However, this can be improved upon by noting that, at any given point in the algorithm, the sum in~(\ref{sth}) has the same value for all $j$ such that the ``current'' $b_j$ is zero.  Thus, if the components are ordered in such a way that the ``currently'' zero components are updated consecutively, then the value of the sum in~(\ref{sth}) may be calculated only once for the update of all $p-p_\star$ zero components, assuming none of these components change to nonzero.  Hence, although each iteration updates all $p$ components, the sum in~(\ref{sth}) only needs to be computed $O(p_\star)$ times.  Hence, a single iteration of coordinatewise descent is in fact only $O(p_\star^2)$.

Thus,
based on computational complexity, it would appear
that coordinate descent
may
enjoy a
non-negligible
advantage over
Deterministic Bayesian Lasso
in terms of the time needed per iteration.
Then any advantage to be gained by SLOG must be realized through a decrease in the number of iterations needed for convergence that is substantial enough to counteract the increased complexity of each individual iteration.  We show in the next section that this is indeed
very much
the case.

\subsection{Similar Algorithms for Lasso-Like Procedures}

The
Deterministic Bayesian Lasso
algorithm was derived as a degenerate limit of a Gibbs sampler for the Bayesian lasso.  However, Bayesian interpretations and corresponding Gibbs samplers have been proposed for a variety of other penalized regression methods beyond the original lasso.  It is natural to consider using
approaches similar to the Deterministic Bayesian Lasso
to  obtain analogous algorithms for such estimators.  We now briefly discuss this idea for some specific lasso variants.
Proofs are not given for the sake of brevity.
In each of the following sections, the hierarchical Bayesian construction of the problem is due to \citet{kyung2010}.

\subsubsection{Elastic Net}

The elastic net \citep{zou2005} imposes both an $L_1$ penalty and an $L_2$ penalty on $\bbeta$ by defining the estimator
\begin{align*}
\bbhen=\argmin_{\bbeta\in\reals^p}\left(||\bm{y}-\bm{X}\bbeta||_2^2+2\lambda_1||\bbeta||_1+\lambda_2\normof{\bbeta}_2^2\right).
\end{align*}
This estimator may be equivalently defined as the mode of the marginal posterior of $\bbeta$ under the Bayesian hierarchical model
\begin{align*}
\bm{y}\mid\bbeta&\sim N_n(\bm{X}\bbeta,\;\sigma^2\bm{I}_n)\\
\beta_j\mid\omega_j&\simind N\left(0,\;\frac{\sigma^2\omega_j}{1+\lambda_2\omega_j}\right)\\
\omega_j&\simiid\text{Exp}(\lambda_1^2/2\sigma^2).
\end{align*}
Here it should be explicitly noted that
at first glance
our hierarchy appears to differ from that of \citet{kyung2010}.  However, it can be seen that the two representations are in fact equivalent by noting that our regularization parameter~$\lambda_1$ and the regularization parameter $\lambda_1^{\kggc}$ of \citet{kyung2010} are related according to $\lambda_1=\lambda_1^{\kggc}\sqrt{\sigma^2}$, and
without loss of generality
we take $\sigma^2$ as known.  Then the elastic net Gibbs sampler draws alternately from the conditionals
\begin{align*}
\bbeta\mid\bm{\omega},\bm{y}&\sim N_p\left[\left(\bm{X}^T\bm{X}+\lambda_2\bm{I}_p+\bm{\Omega}^{-1}\right)^{-1}\bm{X}^T\bm{y},\;
\sigma^2\left(\bm{X}^T\bm{X}+\lambda_2\bm{I}_p+\bm{\Omega}^{-1}\right)^{-1}\right],\\
\omega_j^{-1}\mid\bbeta,\bm{y}&\stackrel{\text{iid}}{\sim}\begin{cases}
\text{InverseGaussian}(\lambda_1/|\beta_j|,\;\lambda_1/\sigma^2) & \text{ if }\beta_j\ne0,\\
\text{InverseGamma}(1/2,\;\lambda_1^2/2\sigma^2) & \text{ if }\beta_j=0,
\end{cases}
\end{align*}
where (as with the lasso) we may replace $\left(\bm{X}^T\bm{X}+\lambda_2\bm{I}_p+\bm{\Omega}^{-1}\right)^{-1}$ by the alternative expression
$\bm{\Omega}^{1/2}\left[\bm{I}_p+\bm{\Omega}^{1/2}\left(\lambda_2\bm{I}_p+\bm{X}^T\bm{X}\right)\bm{\Omega}^{1/2}\right]^{-1}\bm{\Omega}^{1/2}$ whenever an element of $\bm{\omega}$ is zero.
Taking the degenerate limits as $\sigma^2\to0$ yields the recursion relation
\begin{align*}
\bm{b}^{(k+1)}=\left(\bm{B}^{(k)}\right)^{1/2}\left[\lambda_1\bm{I}_p+\left(\bm{B}^{(k)}\right)^{1/2}\left(\lambda_2\bm{I}_p+\bm{X}^T\bm{X}\right)\left(\bm{B}^{(k)}\right)^{1/2}\right]^{-1}\left(\bm{B}^{(k)}\right)^{1/2}\bm{X}^T\bm{y},
\end{align*}
where $\bm{B}^{(k)}=\Diag(|b^{(k)}_1|,\ldots,|b^{(k)}_p|)$ as before.
Numerical investigations indicate that the recursive algorithm above converges to the elastic net estimate for a variety of real and simulated data.

\subsubsection{Group Lasso}

The group lasso \citep{yuan2006} is intended for use when the covariates may be naturally classified into groups and there
is a rationale or need
for covariates in the same group to
be simultaneously either
included or excluded from the
model.
More precisely, suppose there are $M$ groups of covariates, and define the notation
$G_m=\{j:\text{covariate }j\text{ is in group }m\}$.
Then the group lasso estimator is defined as
\begin{align*}
\bbhgl=\argmin_{\bbeta\in\reals^p}\left(||\bm{y}-\bm{X}\bbeta||_2^2+2\lambda\sum_{m=1}^M\sqrt{\sum_{j\in G_m}\beta_j^2}\right).
\end{align*}
This estimator may be equivalently defined as the mode of the marginal posterior of $\bbeta$ under the Bayesian hierarchical model
\begin{align*}
\bm{y}\mid\bbeta&\sim N_n(\bm{X}\bbeta,\;\sigma^2\bm{I}_n)\\
\bbeta_{(m)}\mid\omega_m&\simind N_{p_m}\left(\bm{0}_{p_m},\;\sigma^2\omega_m\bm{I}_{p_m}\right)\\
\omega_m&\simiid\text{Gamma}\left(\frac{p_m+1}{2},\;\frac{\lambda^2}{2\sigma^2}\right),
\end{align*}
for $m\in\oneto{M}$, where the elements of $\bbeta_{(m)}$ are $\{\beta_j:j\in G_m\}$ and $p_m=|G_m|$, the number of covariates in group~$m$.  Once again, we note that the $\lambda$ of our hierarchy and the $\lambda^{\kggc}$ of \citet{kyung2010} differ but are related by $\lambda=\lambda^{\kggc}\sqrt{\sigma^2}$.  Then the group lasso Gibbs sampler draws alternately from the conditionals
\begin{align*}
\bbeta\mid\bm{\omega},\bm{y}&\sim N_p\left[\left(\bm{X}^T\bm{X}+\tilde{\bm{\Omega}}^{-1}\right)^{-1}\bm{X}^T\bm{y},\;
\sigma^2\left(\bm{X}^T\bm{X}+\tilde{\bm{\Omega}}^{-1}\right)^{-1}\right],\\
\omega_m^{-1}\mid\bbeta,\bm{y}&\stackrel{\text{iid}}{\sim}\begin{cases}
\text{InverseGaussian}(\lambda/\normof{\bbeta_{(m)}}_2,\;\lambda/\sigma^2) & \text{ if }\normof{\bbeta_{(m)}}_2\ne0,\\
\text{InverseGamma}(1/2,\;\lambda^2/2\sigma^2) & \text{ if }\normof{\bbeta_{(m)}}_2=0,
\end{cases}
\end{align*}
where $\tilde{\bm{\Omega}}$ is the $p\times p$ diagonal matrix with $j$th diagonal element equal to $\omega_{\tilde{m}_j}$, where $\tilde{m}_j$ is the value of~$m$ such that $j\in G_m$.  Note that we may once again replace $\left(\bm{X}^T\bm{X}+\tilde{\bm{\Omega}}^{-1}\right)^{-1}$ by the alternative expression $\tilde{\bm{\Omega}}^{1/2}\left(\bm{I}_p+\tilde{\bm{\Omega}}^{1/2}\bm{X}^T\bm{X}\tilde{\bm{\Omega}}^{1/2}\right)^{-1}\tilde{\bm{\Omega}}^{1/2}$ whenever an element of $\bm{\omega}$ is zero.  Taking the degenerate limits as $\sigma^2\to0$ yields the recursion relation
\begin{align*}
\bm{b}^{(k+1)}=(\tilde{\bm{B}}^{(k)})^{1/2}\left[\lambda\bm{I}_p+(\tilde{\bm{B}}^{(k)})^{1/2}\bm{X}^T\bm{X}(\tilde{\bm{B}}^{(k)})^{1/2}\right]^{-1}(\tilde{\bm{B}}^{(k)})^{1/2}\bm{X}^T\bm{y},
\end{align*}
where $\tilde{\bm{B}}^{(k)}$ is the $p\times p$ diagonal matrix with $j$th diagonal element equal to $\normof{\bbeta_{(\tilde{m}_j)}}_2$.
Once again, numerical investigations in a variety of settings indicate that the iterates of the algorithm above converge to the group lasso estimate.
\section{%
{%
\color{black}
Applications of
\color{black}
}%
the Deterministic Bayesian Lasso}
\label{sec:numerical}

In this section we investigate the
numerical performance
of the
Deterministic Bayesian Lasso
algorithm in terms of both the number of iterations and computational time. To investigate the \textcolor{black}{speed} of convergence of SLOG we \textcolor{black}{first} apply the algorithm to simulated data. The simulated data are generated using a range of values for $n$, $p$, and the
\textcolor{black}{
level of multicollinearity between covariates.
} Specifically, data are generated according to the following scheme of \citet{JFTHRT}:
\begin{equation}
\label{Esim1}
\bm{y} = \sum_{j=1}^{p} \bm{x}_{(j)}\beta_j+k\sigma,
\end{equation}
\noindent where $\beta_j=(-1)^j\exp(-[j-1]/10)$, $\sigma \sim N(0,1)$, the covariate data $\bm{X}$ of dimension $n \times p$ are multivariate normal with pairwise correlation $\rho$, and $k$ is selected to give a signal-to-noise ratio of 3. Even though the $\bbeta$ are all
technically
non-zero, in the simulations we control the number of non-zero coefficient estimates in the lasso solution (i.e., the degree of of sparsity).

\textcolor{black}{For each dataset generated using relationship \eqref{Esim1}, lasso solutions (denoted $\bm{\beta}^{\lars}$ and \textcolor{black}{corresponding penalty parameter} $\lambda^{\lars}$) corresponding to different levels of sparsity ($0 \leq s \leq 1$) are found using the \textbf{lars} algorithm available in the statistical package R \citep{lars}.}
\textcolor{black}{The level of sparsity, \textcolor{black}{$s$,} defines the proportion of non-zero elements in $\bm{\beta}^{\lars}$ via the relationship $\sum_{i=1}^{p}I(\beta_{i}^{\lars} \neq 0)=(1-s)\times
\min\{n,p\}$, larger $s$ corresponding to a sparser lasso solution. It is important to note that a normalization factor of $\min\{n,p\}$ has been used in the above definition of sparsity because lasso estimates can have at most $\min\{n,p\}$ non-zero coefficients. This normalization allows sparsity levels to take on values from 0 to 1. It is important  however to keep in mind that in typical high dimensional settings when $p \gg n$, the number of zero coefficients as a proportion of the total number of predictors is actually $(s \times n + (p-n))/p  \gg s$.} Implementing the SLOG algorithm with $\lambda=\lambda^{\lars}$ allows the convergence of $\bm{\hat{\beta}}^{\SLOG}$ to $\bm{\beta}^{\lars}$ to be gauged. For comparison, the convergence of the estimates obtained from the coordinate descent algorithm (CD) described in Section~\ref{sec:cdalg}, $\bm{\hat{\beta}}^{\CD}$ to $\bm{\beta}^{\lars}$, is also investigated. \textcolor{black}{Judging the convergence of $\bm{\hat{\beta}}^{\SLOG}$ and $\bm{\hat{\beta}}^{\text{CD}}$ relative to $\bm{\beta}^{\lars}$ ensures that our reported timings for SLOG and CD are comparable - in the sense that the reported timings for SLOG and CD correspond to being the same distance from the ``ground truth" \textbf{lars} solution. Our choice of \textbf{lars} for this purpose stems from the fact that the \textbf{lars} algorithm is an established method that does not favor either CD or SLOG and, before the advent of CD, was arguably the most common means of computing the lasso solution.} The values $K^{\SLOG}$ and $K^{\CD}$ will denote the number of iterations until convergence of $\bm{\hat{\beta}}^{\SLOG}$ and $\bm{\hat{\beta}}^{\CD}$ to $\bm{\beta}^{\lars}$.
The quantity $d(\SLOG,\lars)$ denotes the scaled distance $||\bm{\hat{\beta}}^{\SLOG}-\bm{\beta}^{\lars}||_2/||\bm{\beta}^{\lars}||_2$. \textcolor{black}{The quantity $d(\SLOG(k))=
||\bm{\hat{\beta}}^{\SLOG}(k)-\bm{\hat{\beta}}^{\SLOG}(k-1)||_2 / ||\bm{\hat{\beta}}^{\SLOG}(k-1)||_2$ denotes the scaled distance between the values of $\bm{\hat{\beta}}^{\SLOG}$ at successive iterations $k-1$ and $k$. The meaning of $d(.,.)$ and $d(.)$ with other values follows similarly. In the paper, whenever SLOG is run in a stand alone manner (i.e., not in comparison with CD) it is iterated until $d(\SLOG(k))<1e-3$.}

\textcolor{black}{ We developed code in the statistical package R \citep{statR} to implement both the SLOG and CD algorithms.}
\textcolor{black}{ Both the SLOG and CD algorithms require a choice of starting value for the coefficients. Unless otherwise specified, SLOG and CD will be run using constant, and ``uninformed", starting values of $\bm{\hat{\beta}}^{\SLOG}=\bm{\hat{\beta}}^{\CD}=\sign(\bm{X}^T\bm{y})\lambda^{\lars}/p$.} In the analysis that follows, before apply SLOG and CD the covariate data will be \textcolor{black}{centered} and scaled to have mean zero and unit variance \textcolor{black}{and the response data will be centered to have mean zero}.

\begin{figure}
\hfill
\subfigure[Uninformed Start: $n=1000, p=200$]{\includegraphics[scale=0.8]{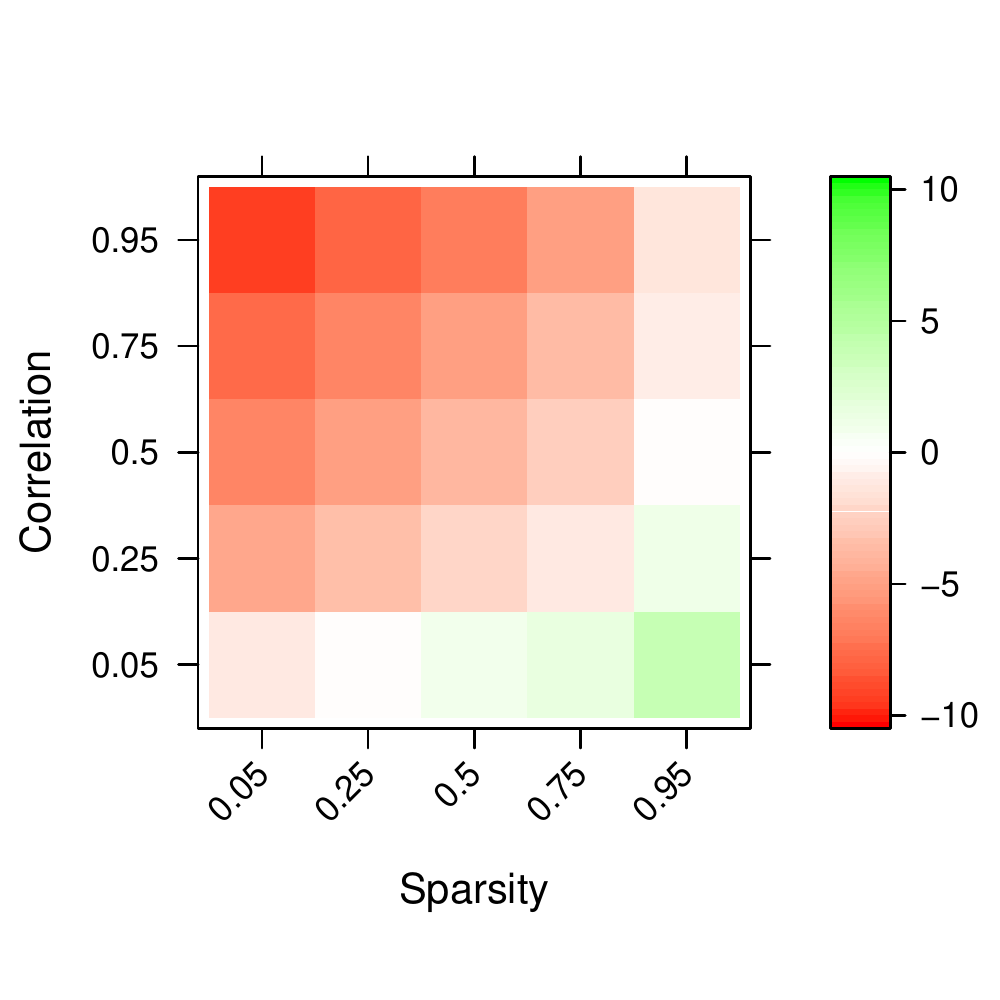} \label{F:heat1a}}
\hfill
\subfigure[Uninformed Start: $n=200, p=1000$]{\includegraphics[scale=0.8]{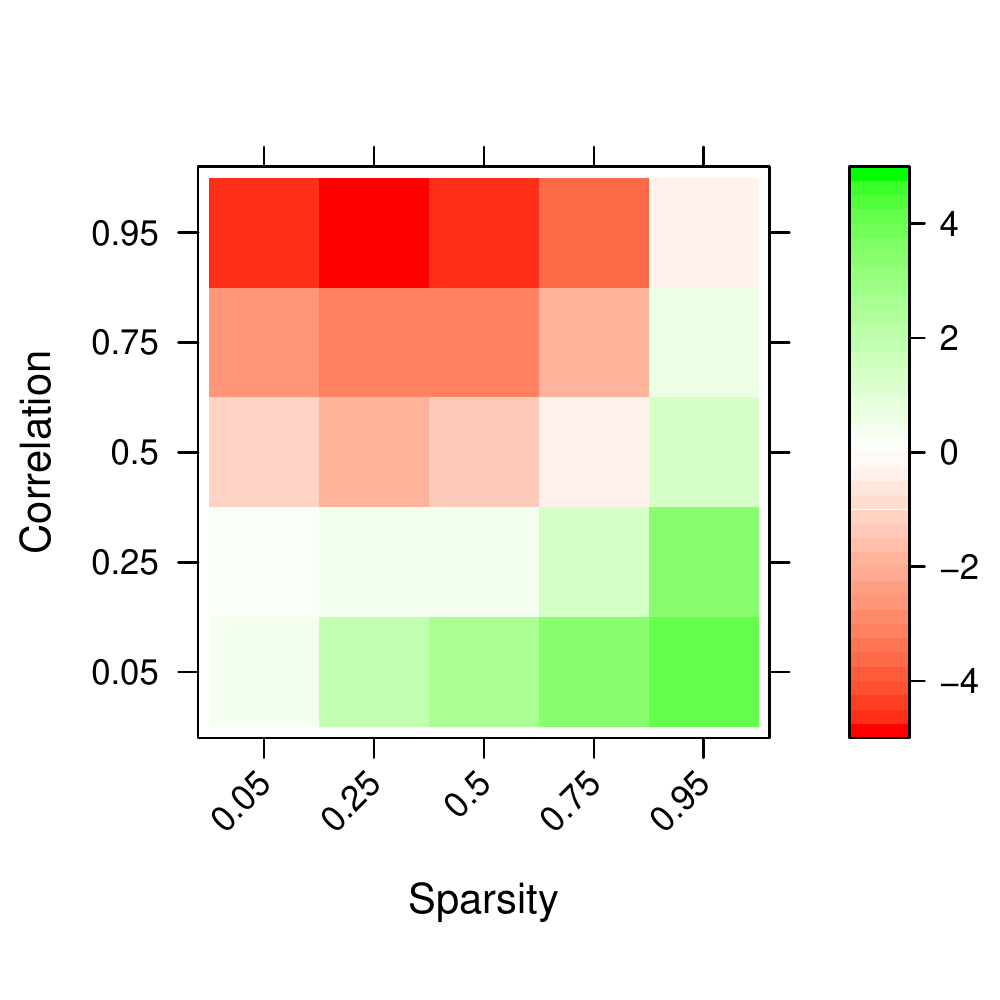} \label{F:heat1b}}
\hfill
\subfigure[Random Start: $n=1000, p=200$]{\includegraphics[scale=0.8]{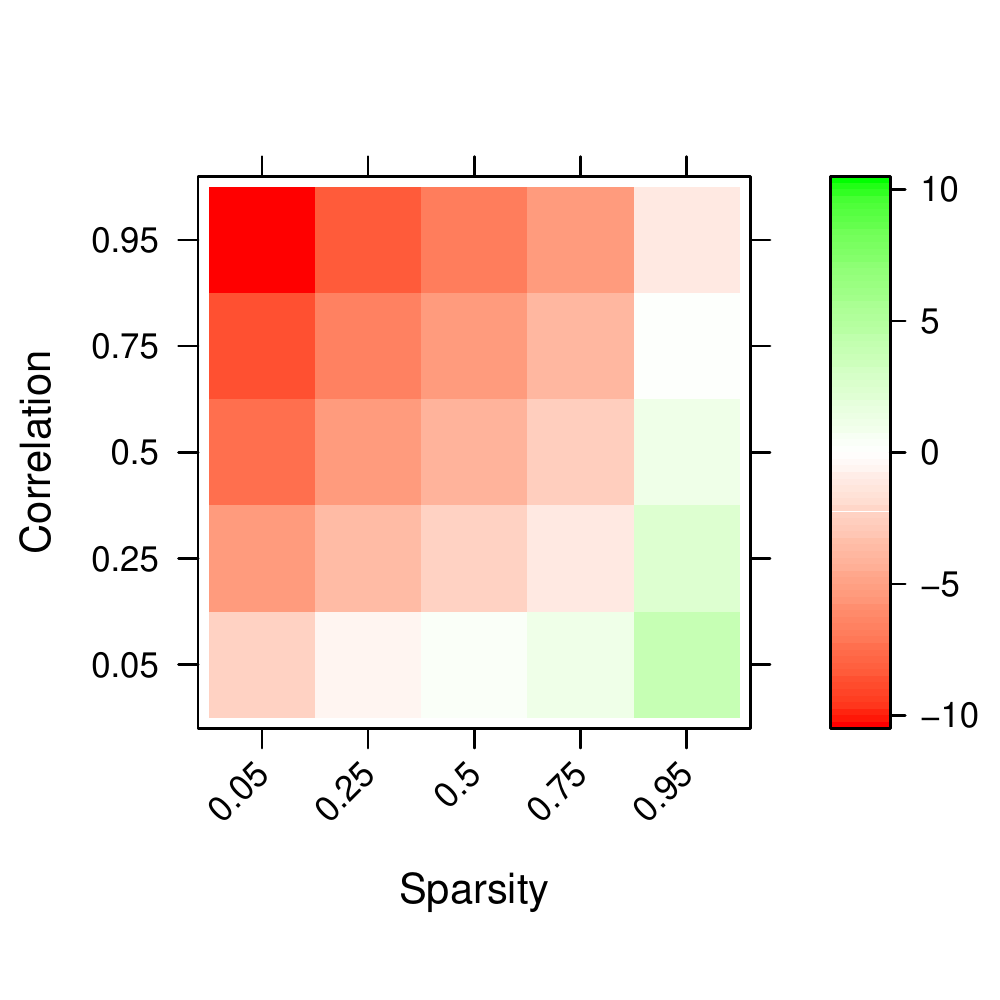} \label{F:heat1c}}
\hfill
\subfigure[Constant Start: $n=1000, p=200$]{\includegraphics[scale=0.8]{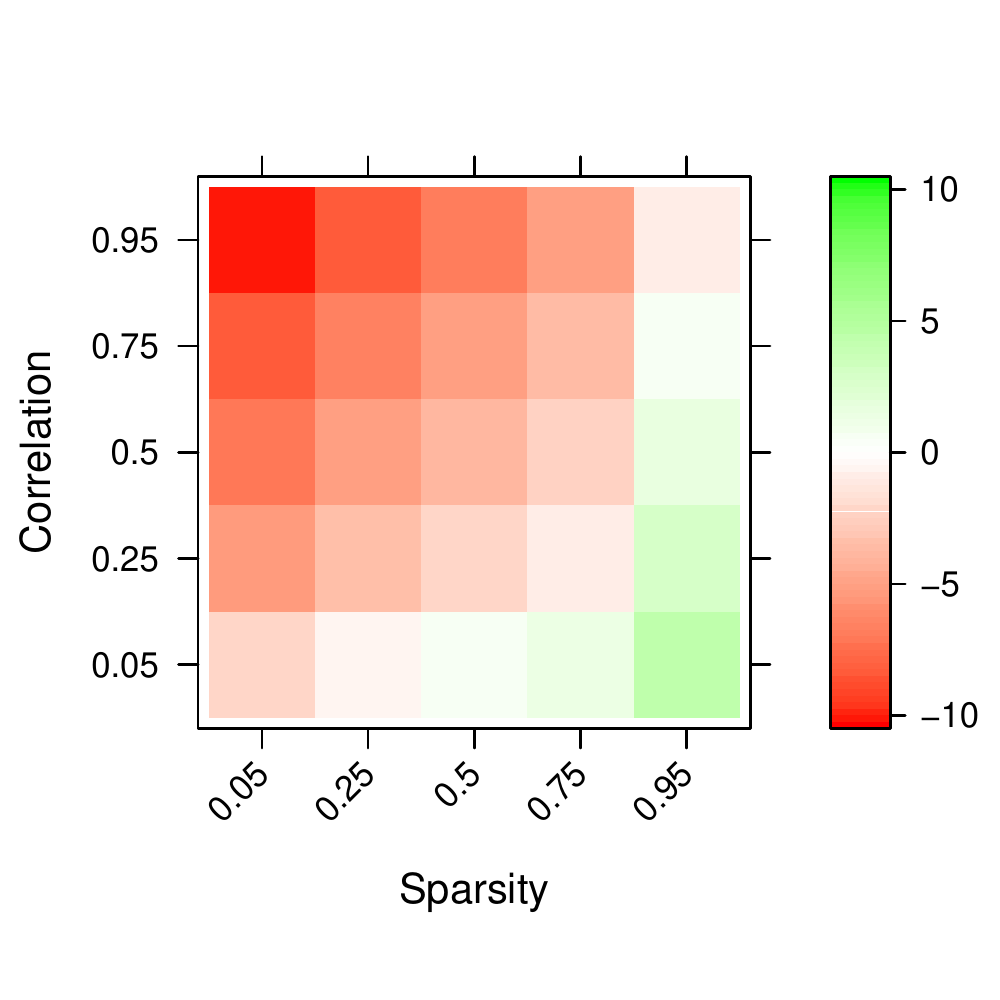} \label{F:heat1d}}
\caption{\label{F:heat1} Difference in the number of iterations until convergence
\textcolor{black}{for the SLOG algorithm relative to the CD algorithm as given by $\log(K^{\SLOG}/K^{\CD})$} . ``Red" regions correspond to SLOG taking fewer iterations to converge as compared to CD, and vice versa for ``green" regions. Data were simulated using relationship \eqref{Esim1} with values of correlation ($\rho$) in the range [0.05,0.95]. Sparsity levels ($s$) in the range [0.05,0.95] were used for $\bm{\beta}^{\lars}$. \textcolor{black} {The starting values for each coefficient in SLOG and CD are:
Uninformed Start: $\sign(\bm{X}^T\bm{y})\lambda^{\lars}/p$, Random Start: sampled from a Uniform distribution on [-5,5], and Constant Start: 0.1.} Convergence for each algorithm was achieved when $d(\CD,\lars)$ or $d(\SLOG,\lars)$ \textcolor{black}{attained a level} less than \textcolor{black}{1e-3. The values in the plots are based on 100 simulated datasets.}}
\end{figure}

\subsection{Iteration Comparison for SLOG}

Figure~\ref{F:heat1} illustrates the difference in the number of iterations until convergence to $\bm{\beta}^{\lars}$ for the SLOG algorithm relative to the CD algorithm \textcolor{black}{for various levels of sparsity, multicollinearity, and the choice of coefficient starting values}. The most striking observation is the marked improvement in the convergence of SLOG relative to CD as sparsity decreases and multicolinearity increases. For example, in scenarios of both low sparsity and high multicollinearity, CD is observed to require \textcolor{black}{over 250,000} additional iterations to converge compared to SLOG. These results, a consequence of the ``one-at-a-time'' coefficient updating employed by CD, are not surprising. It has been documented that CD converges more \textcolor{black}{slowly as multicollinearity increases} \citep{friedman2007}, yet a viable alternative that is demonstrably better has not been proposed. Additionally, as sparsity decreases there are more non-zero coefficients (i.e., a larger ``active set'' of covariates) that are impacted by each one-at-a-time update employed by CD. The issues associated with one-at-a-time updating are avoided by the simultaneous (i.e. ``all-at-once'') updating employed by SLOG.

\begin{figure}
\hfill
\subfigure[$s=0.05$, $\rho=0.10$ $(\circ)$/$\rho=0.95$ $(\triangle)$]{\includegraphics[scale=0.8]{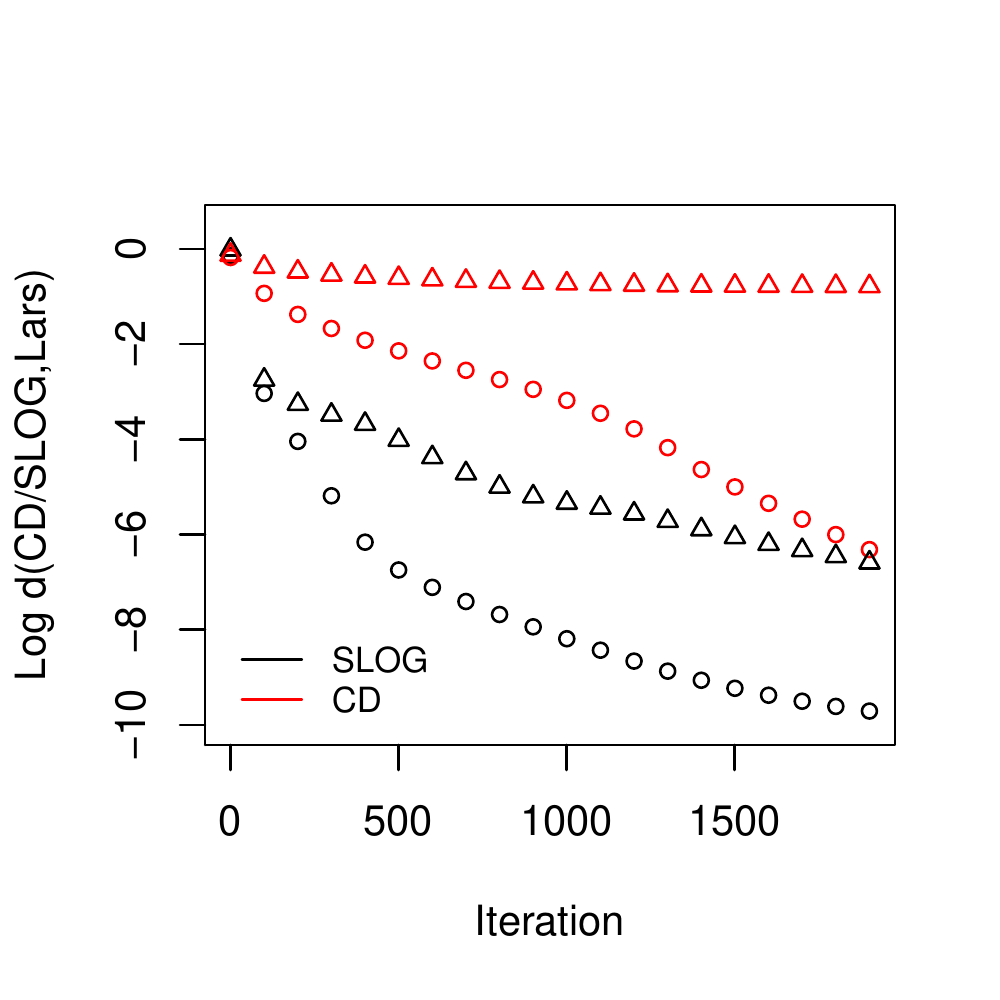} \label{F:conva}}
\hfill
\subfigure[$\rho=0.95$, $s=0.90$ $(\circ)$/$s=0.05$ $(\triangle)$]{\includegraphics[scale=0.8]{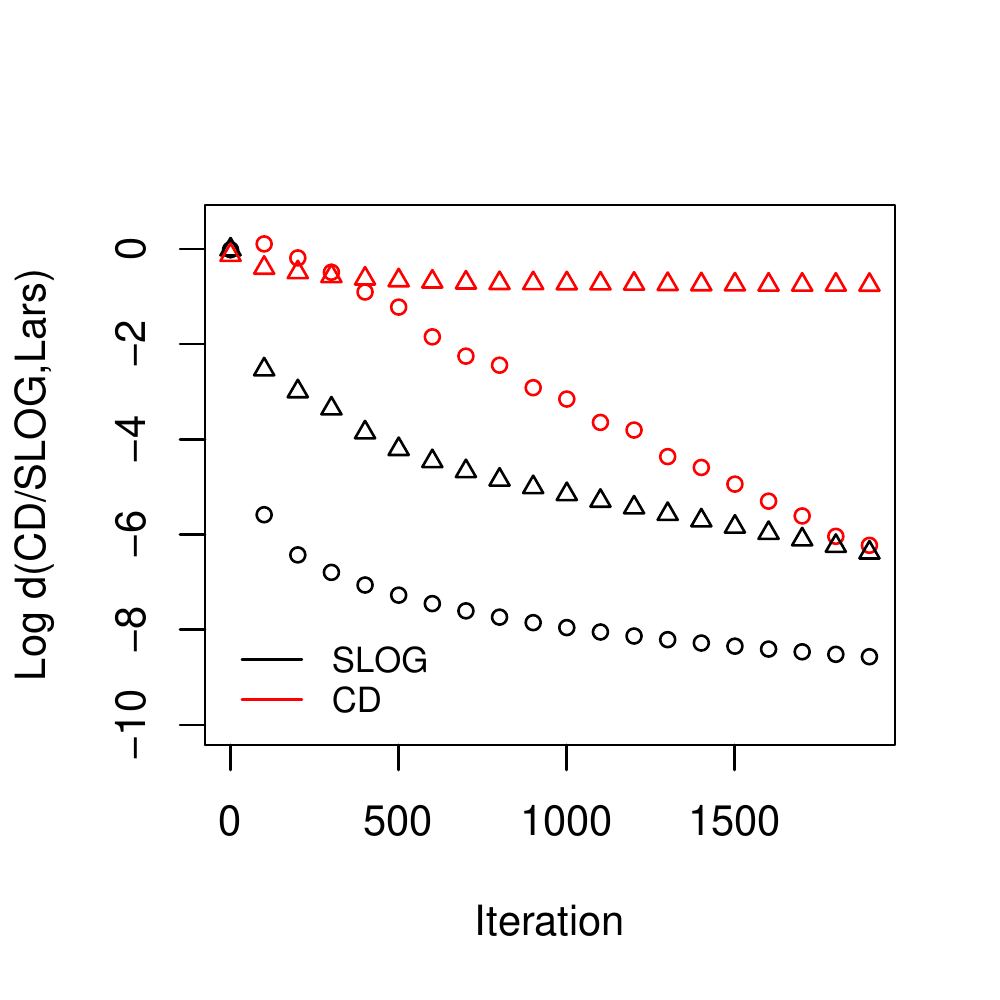} \label{F:convb}}
\hfill
\caption{\label{F:conv} Plots of the distance of $\bm{\hat{\beta}}^{\SLOG}$ and $\bm{\hat{\beta}}^{\CD}$ from $\bm{\beta}^{\lars}$ versus iteration number, for data simulated using relationship \eqref{Esim1} with $n=200$, $p=1000$ and varying levels of correlation ($\rho$). Varying values for the sparsity ($s$) of $\bm{\beta}^{\lars}$ are considered.}
\end{figure}

Examples of the ``overall" convergence paths for both the SLOG and CD algorithms are provided in Figure~\ref{F:conv}. This figure illustrates the more rapid convergence of SLOG compared to CD in situations of low sparsity (Figure~\ref{F:conva}) \textcolor{black}{and/or} high multicollinearity (Figure~\ref{F:convb}). Moreover, the very flat convergence path for CD when $s=0.05$ and $\rho=0.95$ demonstrates the extreme difficulties that CD experiences when faced with both low sparsity \textcolor{black}{and/or} high multicollinearity. The SLOG algorithm does not experience these difficulties.

As a further comparison of the convergence \textcolor{black}{patterns} of SLOG compared to CD, Figure~\ref{F:path} contains plots of the coefficient estimates from both algorithms as a function of iteration number. These plots illustrate the strikingly divergent paths that the SLOG and CD coefficient estimates take on their way to $\bm{\beta}^{\lars}$. In particular, in the setting of high multicollinearity and low sparsity (Figure~\ref{F:patha}) SLOG converges immediately to $\bm{\beta}^{\lars}$. Conversely, CD converges much more slowly to $\bm{\beta}^{\lars}$. Moreover,
\textcolor{black}{
due to the high multicollinearity, the path CD takes for one of the coefficients is initially in the opposite direction to the solution.}
In the setting of low multicollinearity and high sparsity (Figure~\ref{F:pathb}) essentially the opposite behavior to Figure~\ref{F:patha} is observed.

\begin{figure}
\hfill
\subfigure[$s=0, \rho=0.99, n=50, p=3$]{\includegraphics[scale=0.8]{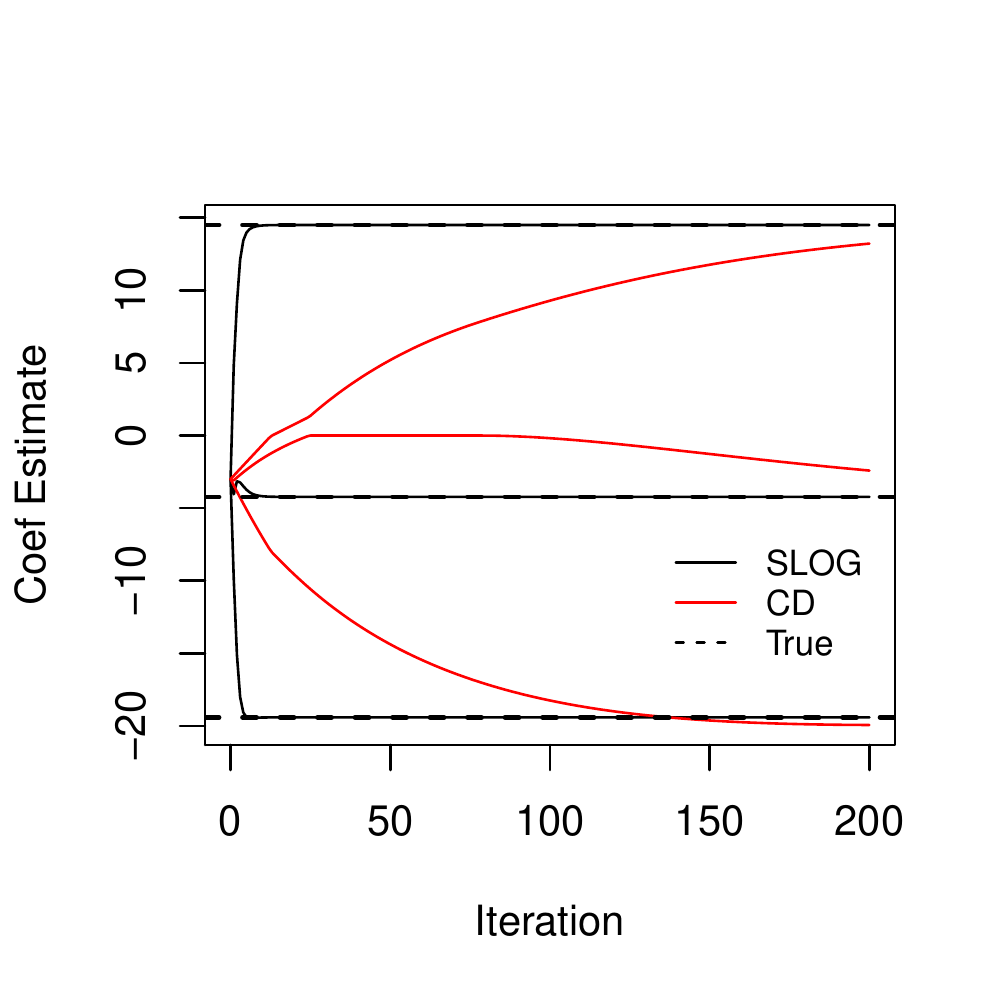} \label{F:patha}}
\hfill
\subfigure[$s=2/3, \rho=0.25, n=50, p=3$]{\includegraphics[scale=0.8]{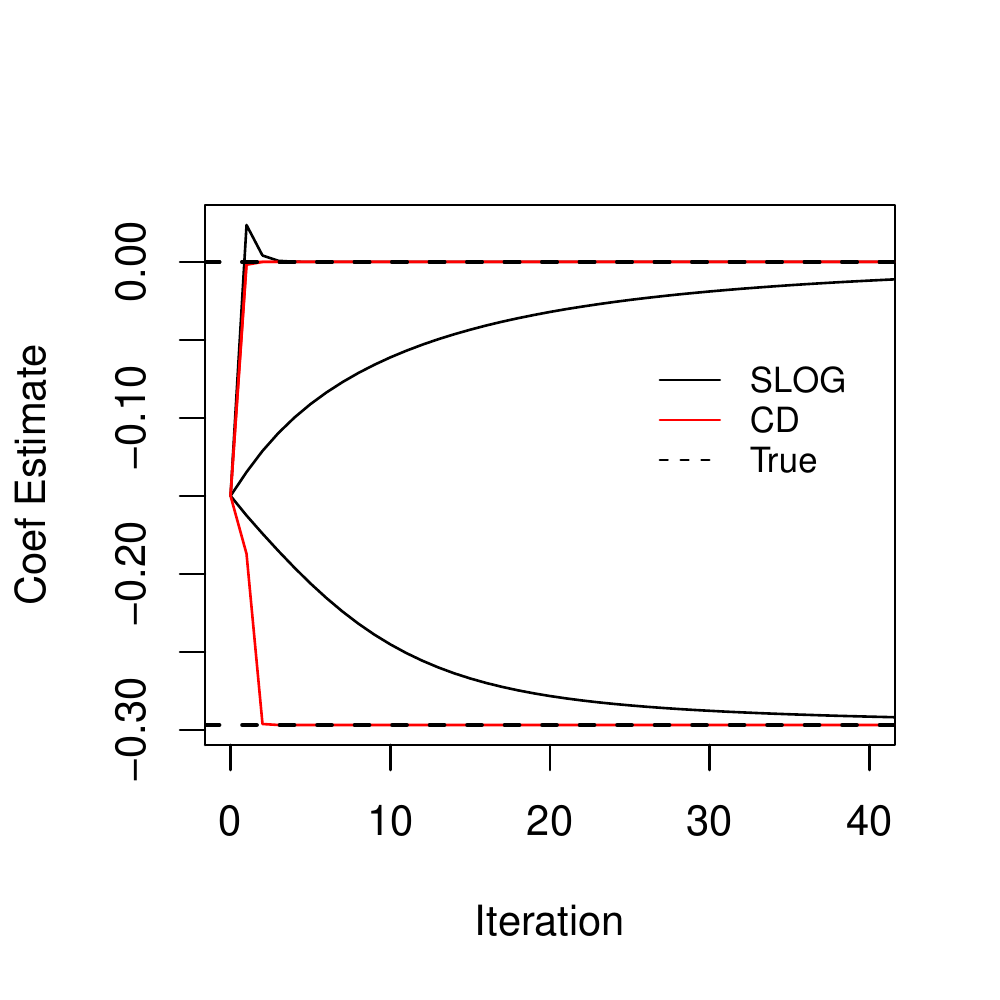} \label{F:pathb}}
\hfill
\caption{\label{F:path} Convergence of $\bm{\hat{\beta}}^{\SLOG}$ and $\bm{\hat{\beta}}^{\CD}$ to $\bm{\beta}^{\lars}$ in a ``non-sparse, high multicollinearity" setting (Figure \ref{F:patha}) and a ``sparse, low multicollinearity" setting (Figure \ref{F:pathb}). Data were simulated using relationship \eqref{Esim1}. The dashed vertical lines represent the values of $\bm{\beta}^{\lars}$. \textcolor{black}{SLOG and CD were implemented using starting values for each coefficient of -3 in (a) and -0.15 in (b).}}
\end{figure}

\textcolor{black}{
Figure~\ref{F:zero} contains plots of the rate at which the SLOG and CD algorithms set coefficient estimates to effective zeros. For the purposes of this figure, effective zeros for SLOG and CD are, respectively, defined as $|\bm{\hat{\beta}}^{\SLOG}|$ or $|\bm{\hat{\beta}}^{\CD}|$ being smaller than \textcolor{black}{1e-13.}} The rate at which the two algorithms set coefficients to exactly zero cannot be \textcolor{black}{directly} compared because by design SLOG coefficient estimates approach zero in the limit, rather than being set exactly to zero. Figure~\ref{F:zeroa} clearly illustrates that in situations of high multicollinearity and low sparsity that CD has much more difficulty locating zeros compared to SLOG. This finding, coupled with the fact that CD has more difficulty converging to non-zero coefficients, explains the larger number of iterations required by CD to converge compared to SLOG.
\textcolor{black}{In the situation of low multicollinearity and high sparsity (Figure~\ref{F:zerob}) the opposite behavior is observed with SLOG having more difficulty locating zeros compared to CD. Moreover, in the setting of Figure~\ref{F:zerob} CD is able to set the overwhelming majority of coefficients to effective zeros in a single iteration.} In Section~\ref{sec:rslog} we propose a modification to SLOG that enables the SLOG coefficient estimates to be set to exactly zero.

\begin{figure}
\hfill
\subfigure[$s=0.05, \rho=0.95$]{\includegraphics[scale=0.8]{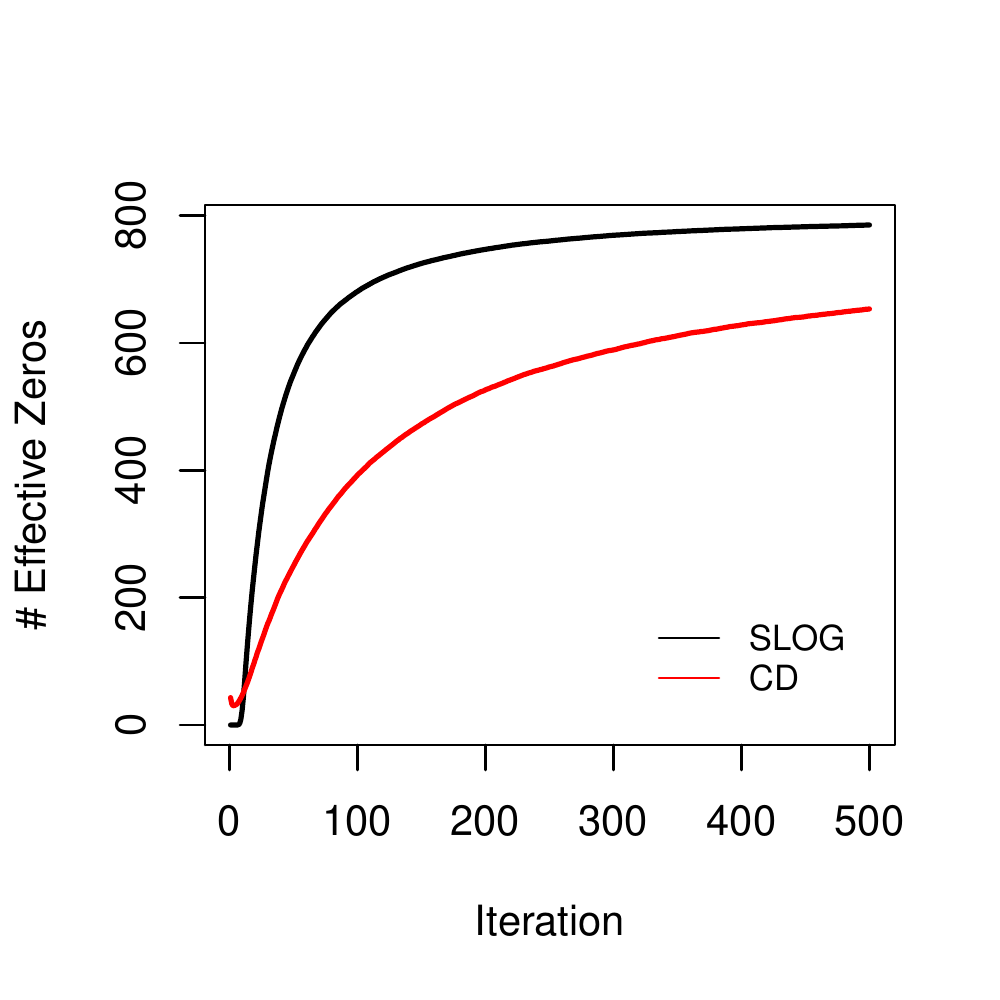} \label{F:zeroa}}
\hfill
\subfigure[$s=0.90, \rho=0.10$]{\includegraphics[scale=0.8]{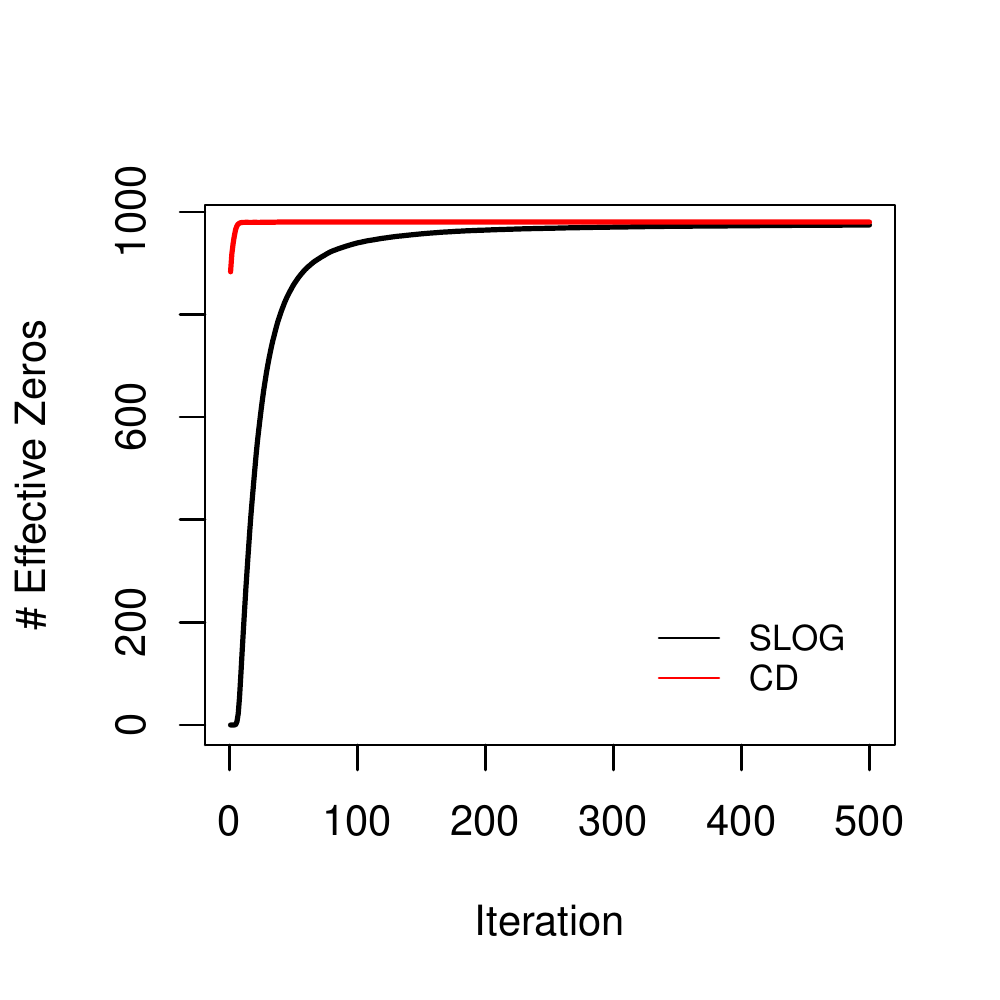} \label{F:zerob}}
\hfill
\caption{\label{F:zero} Number of elements of $\bm{\hat{\beta}}^{\SLOG}$ and $\bm{\hat{\beta}}^{\CD}$ that are less than
than \textcolor{black}{1e-13} in absolute value (effective zeros) in a ``non-sparse, high multicollinearity" setting (Figure \ref{F:zeroa}) and a ``sparse, low-multicollinearity" setting (Figure \ref{F:zerob}). Data were simulated using relationship \eqref{Esim1} with $n=200$ and $p=1000$. \textcolor{black}{The values in the plots are based on 100 simulated datasets.}}
\end{figure}

The results of this section have demonstrated that in situations of high multicollinearity \textcolor{black}{and/or} low sparsity, that the SLOG algorithm enjoys a substantial advantage over the CD algorithm in terms of the number of iterations required until convergence. However, as demonstrated in Section~\ref{sec:comp}, the computational complexity of each SLOG iteration is substantially greater than that of each CD iteration. In the next section we explore, via computational time, whether the reduced number of iterations required by SLOG, in certain settings, is enough to offset the greater complexity of each of its iterations.

\subsection{Timing Comparison for SLOG}

\subsubsection{Reduced Deterministic Bayesian Lasso Algorithm}
\label{sec:rslog}
Unlike the CD algorithm,
the SLOG algorithm is not designed to set coefficients to \textcolor{black}{exact zeros}.
Computationally, this property of SLOG requires the inversion of a $p \times p$ matrix at each iteration.
A simple modification to SLOG which overcomes \textcolor{black}{this hurdle} is to set coefficients to \textcolor{black}{exact zeros} once they fall below a pre-defined threshold $\theta$.
\textcolor{black}{Recall that } Lemma~\ref{lem:partition} can \textcolor{black}{now} be invoked to justify the inversion of a $(p-p_k)\times (p-p_k)$ matrix \textcolor{black}{at iteration $k$ of SLOG}, where $p_k$ denotes the number of coefficients that have been set to zero at the start of iteration $k$.
Figure~\ref{F:zero} provides illustrations of the rate at which $p_k$ increases with iteration number $k$ for \textcolor{black}{$\theta=$1e-13.}
We term this approach of ``thresholding" coefficients to zero the reduced SLOG (rSLOG) algorithm. \textcolor{black}{We shall demonstrate that in} high-dimensional settings the rSLOG algorithm can offer massive speed-ups in computational time compared to \textcolor{black}{traditional} SLOG. These speed-ups are \textcolor{black}{also} possible without a loss of estimation accuracy.
Two \textcolor{black}{competing} factors need to be balanced when choosing a threshold for the rSLOG algorithm: as $\theta$ decreases (1) the likelihood of rSLOG incorrectly thresholding coefficients to zero decreases; and (2) the computational speed of rSLOG generally decreases.
It is essential to note that if rSLOG correctly thresholds coefficients to zero that its use will \emph{not} result in a loss of estimation accuracy compared to SLOG.
\textcolor{black}{We shall demonstrate that there is a relatively large range of $\theta$ which is ``suitable";} Suitable $\theta$ being those for which rSLOG offers substantial speed-up over SLOG and CD, without a loss of estimation accuracy. In this paper, unless otherwise noted, we set \textcolor{black}{$\theta=$1e-13.} When interpreting the magnitude of $\theta$ it is important to remember that we standardize the covariate data to have zero mean and unit variance before applying rSLOG.

\begin{figure}
\hfill
\subfigure[$s=0.75, \rho=0.25$]{\includegraphics[scale=0.8]{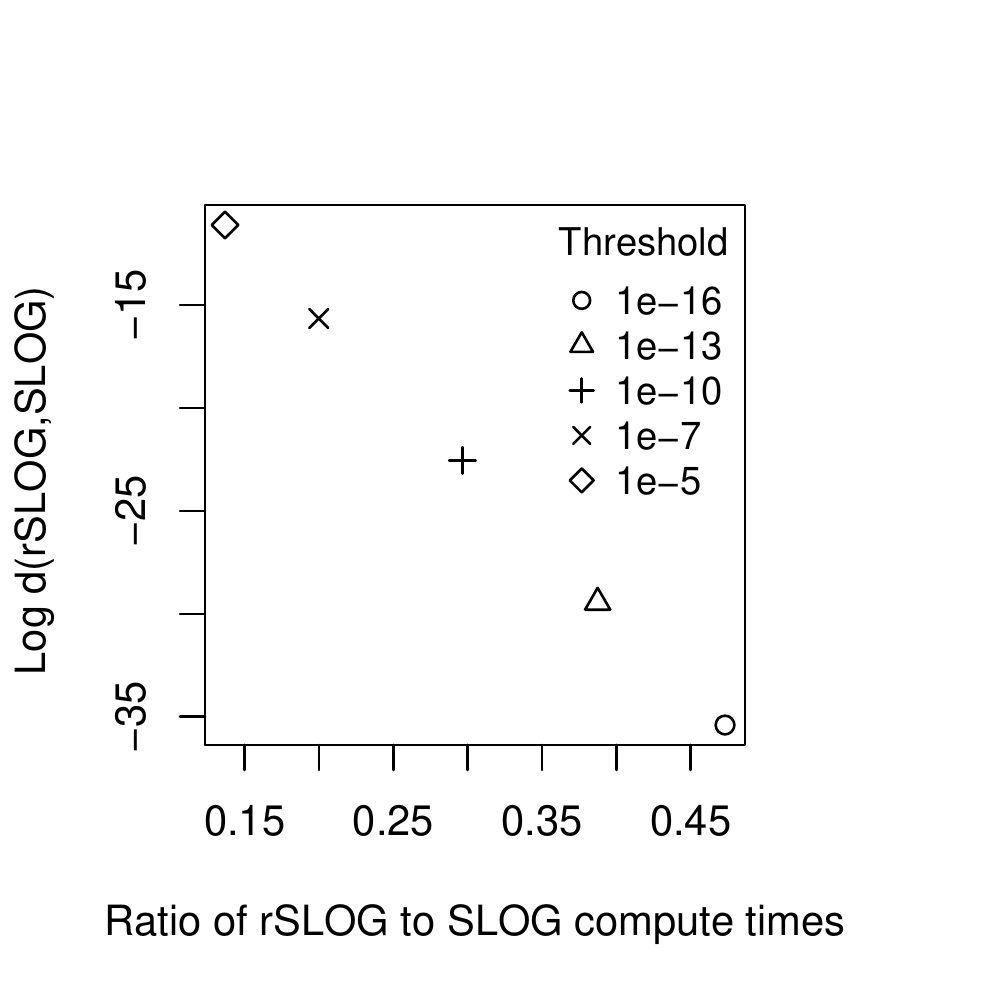} \label{F:theta1}}
\hfill
\subfigure[$s=0.05, \rho=0.95$]{\includegraphics[scale=0.8]{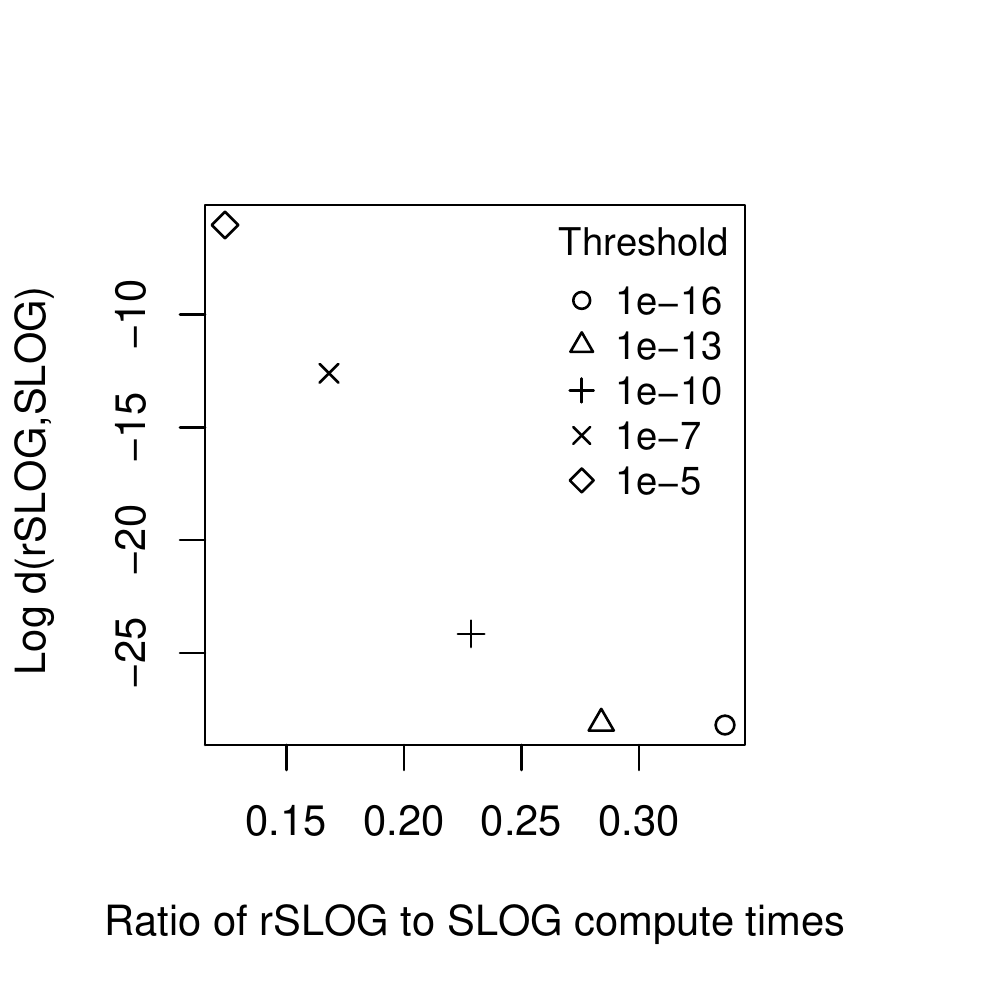} \label{F:theta2}}
\hfill
\caption{\label{F:theta}
\textcolor{black}{
Plots of the distance of $\bm{\hat{\beta}}^{\rSLOG}$ from $\bm{\hat{\beta}}^{\SLOG}$ and the ratio of the computation times for $\bm{\hat{\beta}}^{\rSLOG}$ compared to $\bm{\hat{\beta}}^{\SLOG}$ for five different threshold values ($\theta$). The values in the plots are averages over 100 datasets simulated using relationship \eqref{Esim1} with $n=200$ and $p=1000$. SLOG and rSLOG were run until $d(\SLOG(k))$ and $d(\rSLOG(k))$ were less than 1e-3, respectively.}}
\end{figure}

Figure~\ref{F:theta} provides examples of the trade-off between the closeness of $\bm{\hat{\beta}}^{\rSLOG}$ to $\bm{\hat{\beta}}^{\SLOG}$ and the computational time of rSLOG, over a range of threshold values. We note that because SLOG does not set coefficients to \textcolor{black}{exact zeros that,} the distance between $\bm{\hat{\beta}}^{\SLOG}$ and $\bm{\hat{\beta}}^{\rSLOG}$ will never be zero. \textcolor{black}{Figure~\ref{F:theta} clearly demonstrates the wide range of $\theta$ values that result in $\bm{\hat{\beta}}^{\rSLOG}$ being ``close" to $\bm{\hat{\beta}}^{\SLOG}$. Further, the relative computation times for rSLOG and SLOG indicate that rSLOG offers the potential for substantial improvements in computation time compared to SLOG. With these relative computation times in mind, it is clear that there is a large range of threshold values for which rSLOG offers substantial speed-ups, without loss of accuracy, compared to SLOG.}

\textcolor{black}{In addition to rSLOG, there are other potential means of avoiding the inversion of a $p \times p$ matrix at each iteration of SLOG. One such possibility is a ``hybrid" approach that splits the covariates into blocks and then applies SLOG (or rSLOG) to blocks with high multicollinearity and CD to blocks with low multicollinearity. This hybrid approach is developed in Supplemental Section~\ref{sec:AppD}.}


\subsubsection{Coordinate Descent Algorithm via \textbf{glmnet}}
\label{sec:cdglmnet}
\textcolor{black}{The most popular means of fitting the lasso in practice is the \textbf{glmnet} function of \citet{JFTHRT} available in~R. This function implements the CD algorithm in a pathwise fashion. The popularity of \textbf{glmnet} is attributable to its ability to efficiently compute the lasso solution. Thus for timing comparisons of CD against SLOG/rSLOG we shall use \textbf{glmnet}. The \textbf{glmnet} function implemented in~R does the majority of its numerical computations in Fortran \citep{JFTHRT}. The SLOG/rSLOG algorithm is implemented using code that is wholly written in~R.}

\textcolor{black} {The timings for the CD algorithm (implemented via \textbf{glmnet}) reported below are based on a convergence threshold of 1e-13.} In particular, the \textbf{glmnet} function is run until the maximum change in the objective function (i.e. the penalized residual sum of squares) is less than the convergence threshold multiplied by the null deviance \citep{JFTHRT}. Additionally, the \textbf{glmnet} function is fitted using a decreasing sequence of 50 $\lambda$ values that span $[\lambda^{\text{max}},\lambda^{\lars}]$. The value $\lambda^{\text{max}}$, automatically found by \textbf{glmnet}, is the smallest value of $\lambda$ that sets each coefficient estimate to zero. The use of a sequence of $\lambda$, rather than just the single value of interest $\lambda^{\lars}$, is recommended by the authors of the \textbf{glmnet} package who state: ``\emph{Do not supply a single value for lambda \ldots Supply instead a decreasing sequence of lambda values. glmnet relies on its warm starts for speed, and its often faster to fit a whole path than compute a single fit.}" The robustness of the reported results to the length of the $\lambda$ sequence used is investigated. The timings for the SLOG/rSLOG algorithm are based on the single $\lambda$ of interest, $\lambda^{\lars}$, and starting values of $\bm{\hat{\beta}}^{\SLOG}=\bm{\hat{\beta}}^{\rSLOG}=\sign(\bm{X}^T\bm{y})\lambda^{\lars}/p$. To ensure similar convergence \textcolor{black}{thresholds} for CD and SLOG/rSLOG, the latter two algorithms are iterated until the distance of $\bm{\hat{\beta}}^{\SLOG}$ (or $\bm{\hat{\beta}}^{\rSLOG}$) from $\bm{\beta}^{\lars}$ is at least as small as the distance of $\bm{\hat{\beta}}^{\CD}$ from $\bm{\beta}^{\lars}$.

\subsubsection{Simulation Study}
For data generated using relationship \eqref{Esim1}, the timings of CD (via \textbf{glmnet}) and rSLOG for computing the lasso solution corresponding to $\lambda^{\lars}$ were recorded. In the simulations we focus on low sparsity and/or high multicollinearity \textcolor{black}{settings} because these are the situations where rSLOG offers the potential for increased computational speed compared to CD. \textcolor{black} {Due to the benefits of rSLOG compared to SLOG described in Section \ref{sec:rslog} only the rSLOG algorithm is considered in the simulations that follow.}

\textcolor{black} {Figure~\ref{F:box1} contains box plots of the computation time until convergence for the CD and rSLOG algorithms applied to simulated datasets of size $n=1000$ and $p=200$. It is observed that in low sparsity and/or high multicollinearity settings rSLOG can offer substantial increases in computational speed compared to CD. The box plots also illustrate that the variability of the rSLOG computation time is substantially smaller than the corresponding CD computation time. Moreover outlying computation times are more prevalent for CD. The observed improvements in speed offered by rSLOG increase as sparsity decreases and/or multicollinearity increases. The increased computational speed of rSLOG can be attributed to the massive decrease in the number of iterations required. The decrease in the number of iterations required by rSLOG is more than sufficient to compensate for the increased complexity of each of its iterations. For example, when $s=0.05$ and $\rho=0.95$, on average, rSLOG required approximately 315 iterations to converge compared to approximately 780,000 for CD. The increased variability of the CD computation time relative to the rSLOG computation time is likely due to the ``piecemeal" nature of the one-at-a-time updating employed by CD, compared to the more ``holistic" all-at-once updating employed by rSLOG. The nature of the updating employed by CD makes its computation time more dependent on the vagaries of the simulated data.} A thorough analysis was undertaken to investigate the robustness of the results reported in Figure~\ref{F:box1} under different regimes. In particular, qualitatively similar results were obtained under the following modifications to the simulation settings used in Figure~\ref{F:box1}: (1) changing the length of the $\lambda$ sequence used in CD to 100 or 25; (2) \textcolor{black} {changing $\theta$ to 1e-10 or 1e-16;} and (3) decreasing the \textbf{glmnet} convergence threshold to 1e-11 or 1e-9. When the \textbf{glmnet} convergence threshold is decreased rSLOG is still observed to offer increased computational speed compared to CD, however, the magnitude of the increase is less.

\begin{figure}
\hfill
\subfigure[Correlation $\rho=0.85$]{\includegraphics[scale=0.8]{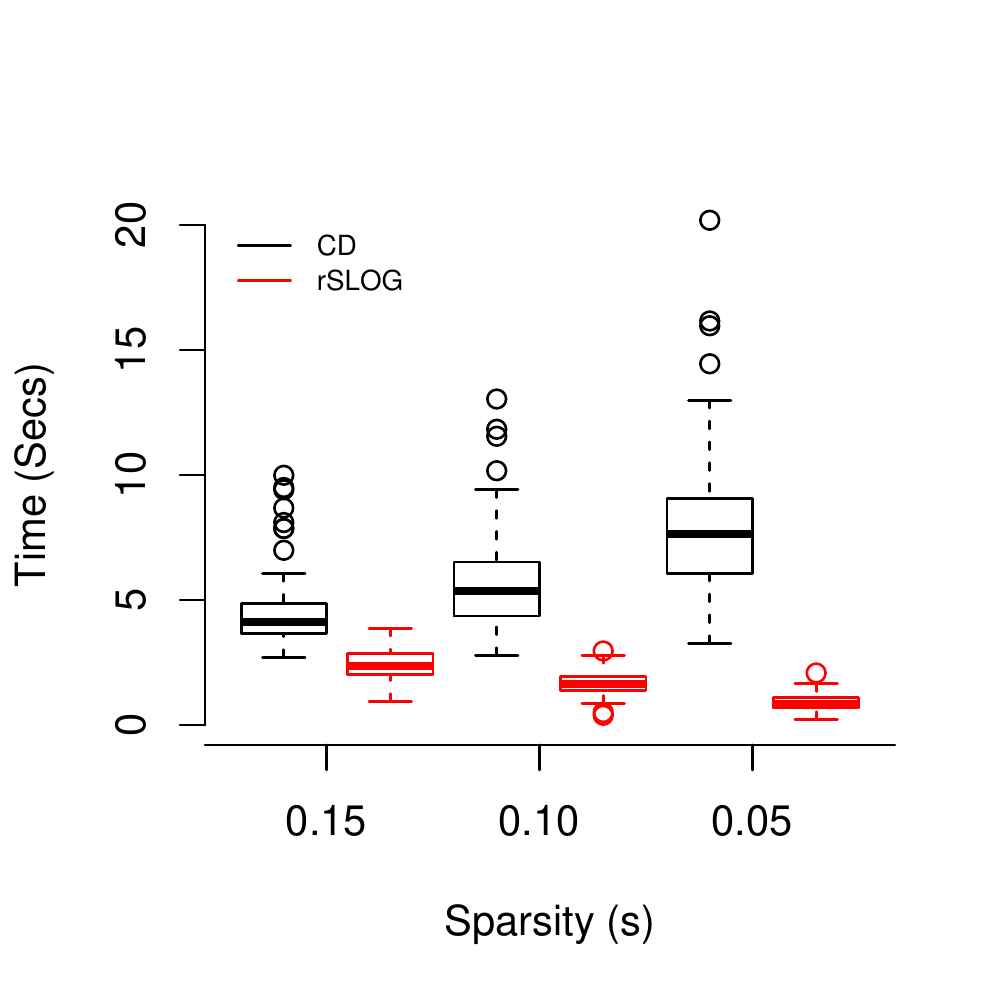} \label{F:box1a}}
\hfill
\subfigure[Correlation $\rho=0.90$]{\includegraphics[scale=0.8]{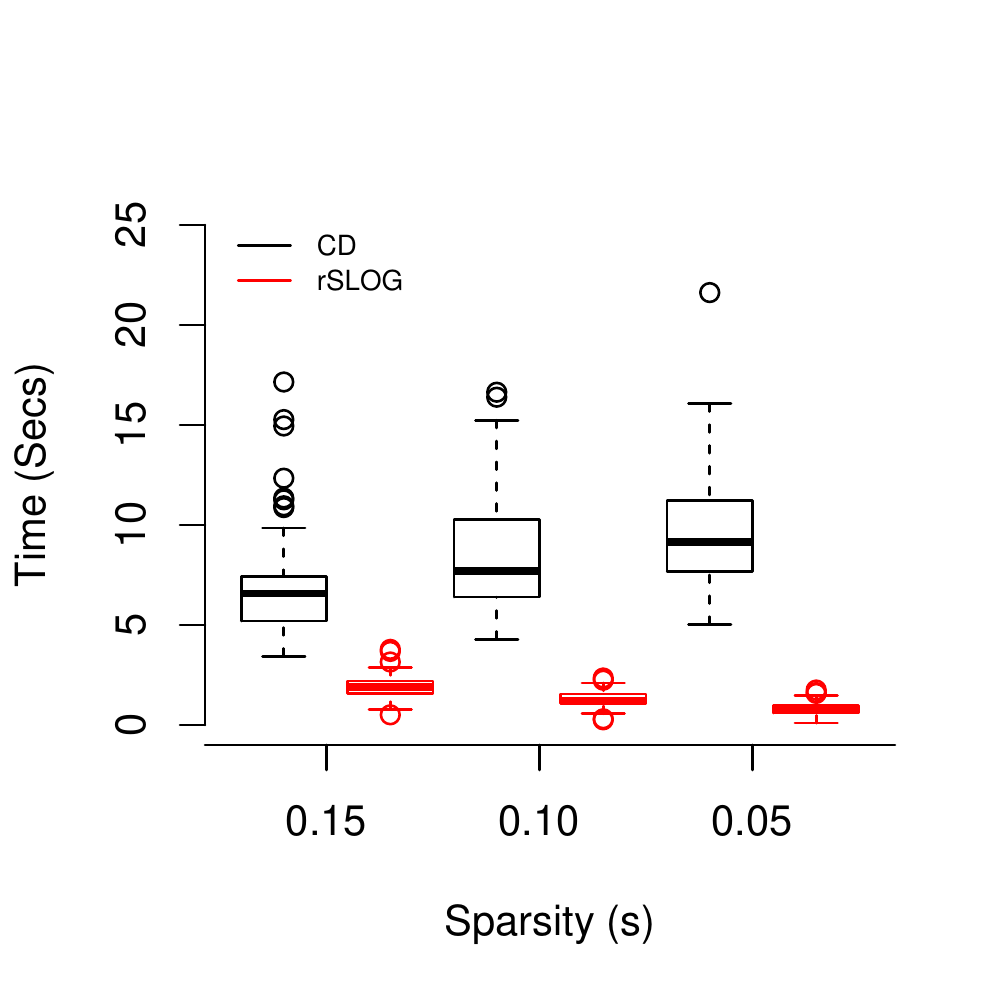} \label{F:box1b}}
\hfill
\subfigure[Correlation $\rho=0.95$]{\includegraphics[scale=0.8]{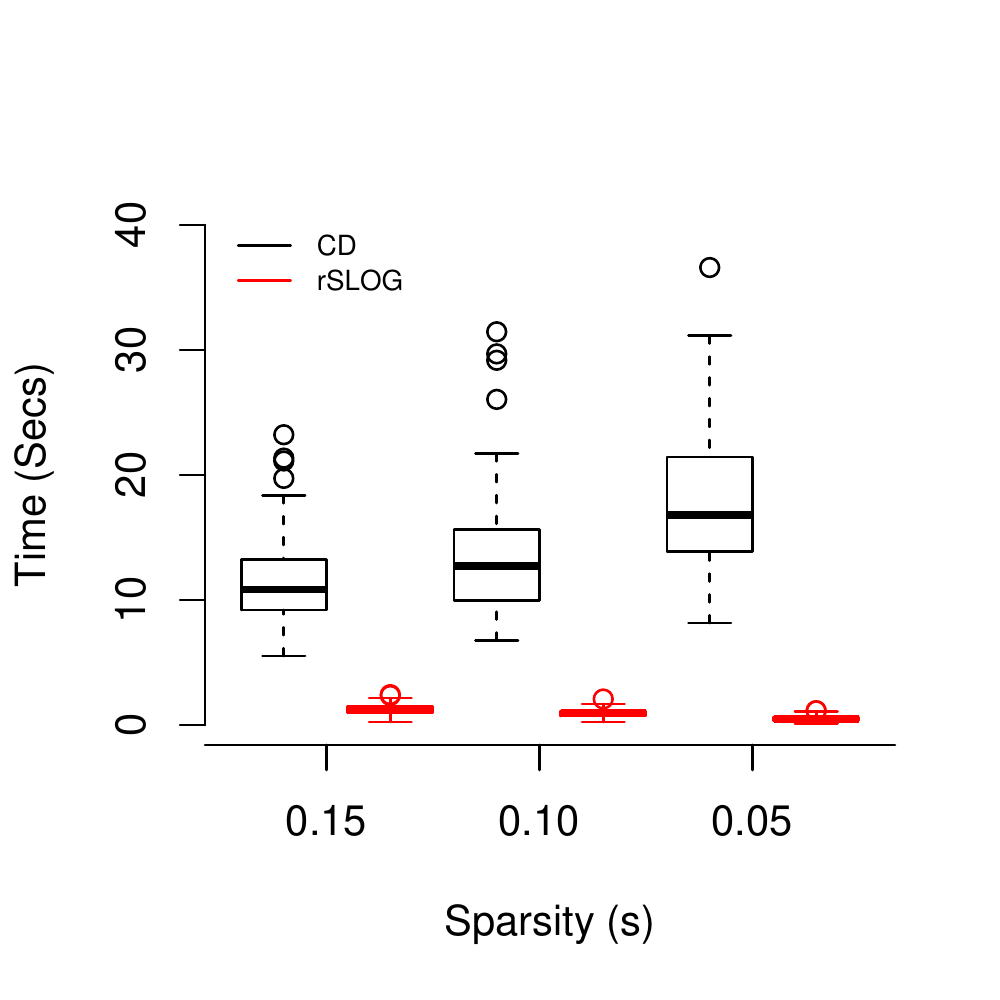} \label{F:box1c}}
\hfill
\subfigure[Correlation $\rho=0.99$]{\includegraphics[scale=0.8]{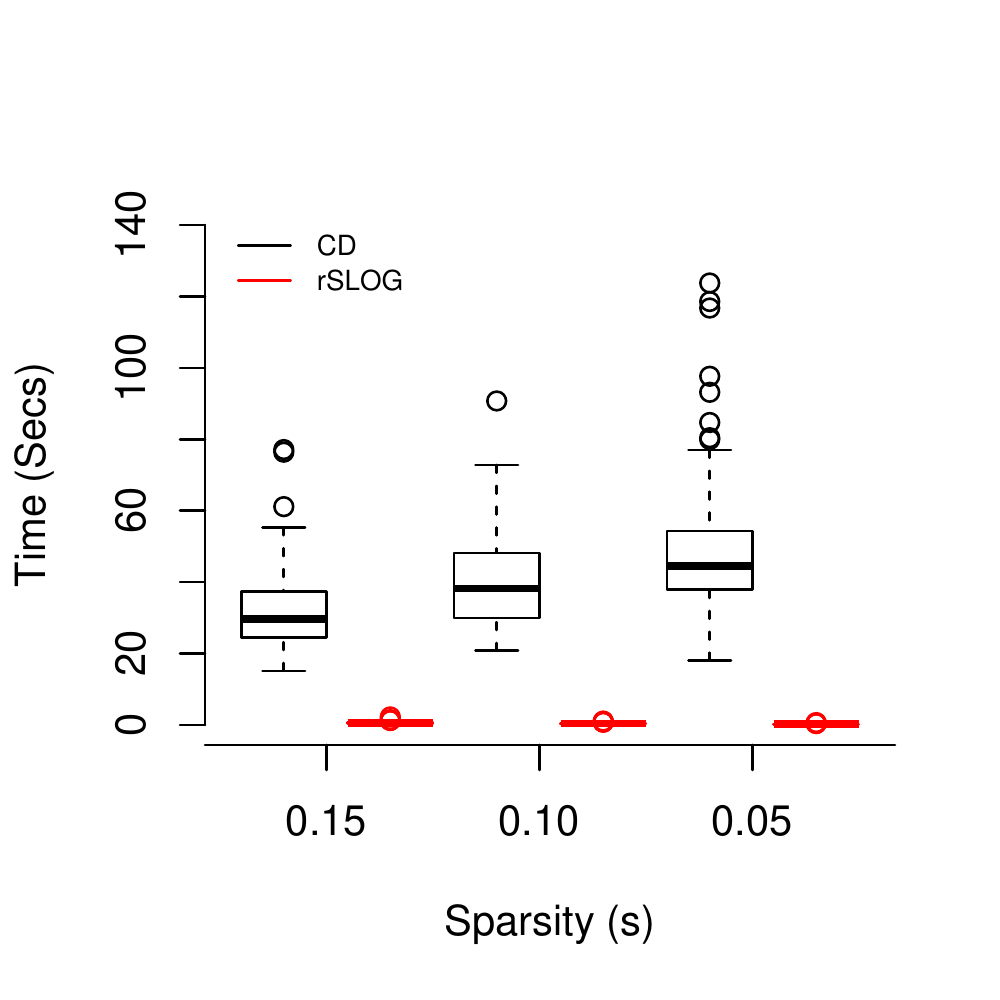} \label{F:box1d}}
\caption{\label{F:box1} Time in seconds until convergence of the CD and rSLOG algorithms applied to data generated using relationship \eqref{Esim1} with $n=1000$, $p=200$. Each box plot is based on 100 simulated datasets.}
\end{figure}

\begin{figure}
\hfill
\subfigure[Correlation $\rho=0.95$]{\includegraphics[scale=0.8]{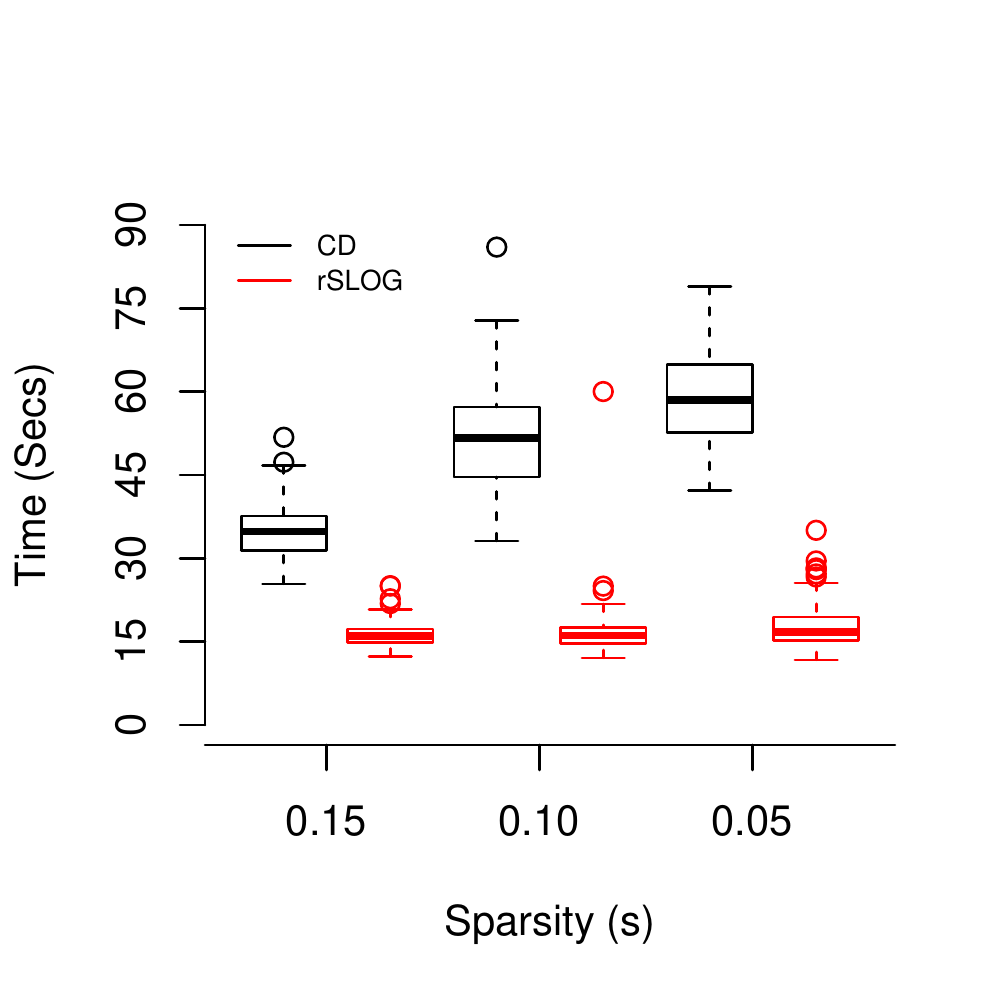} \label{F:box2a}}
\hfill
\subfigure[Correlation $\rho=0.99$]{\includegraphics[scale=0.8]{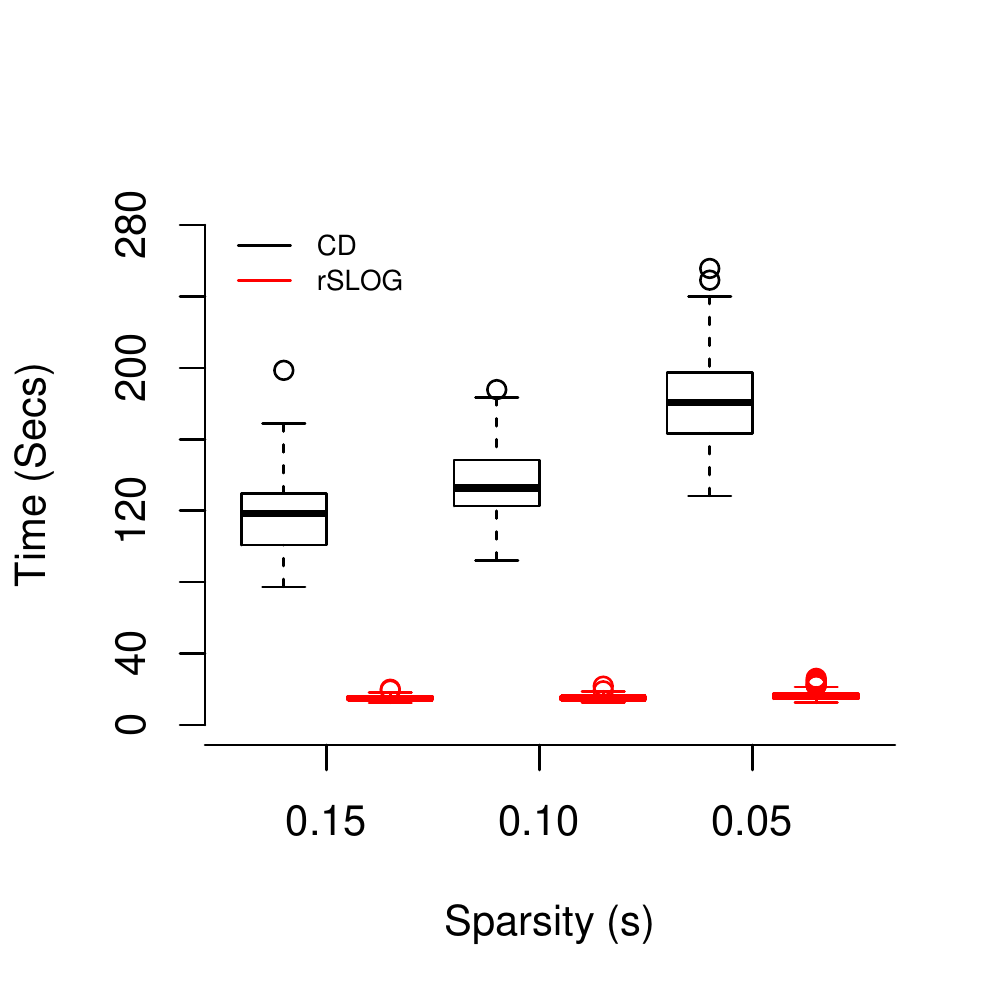} \label{F:box2b}}
\hfill
\subfigure[Correlation $\rho=0.95$]{\includegraphics[scale=0.8]{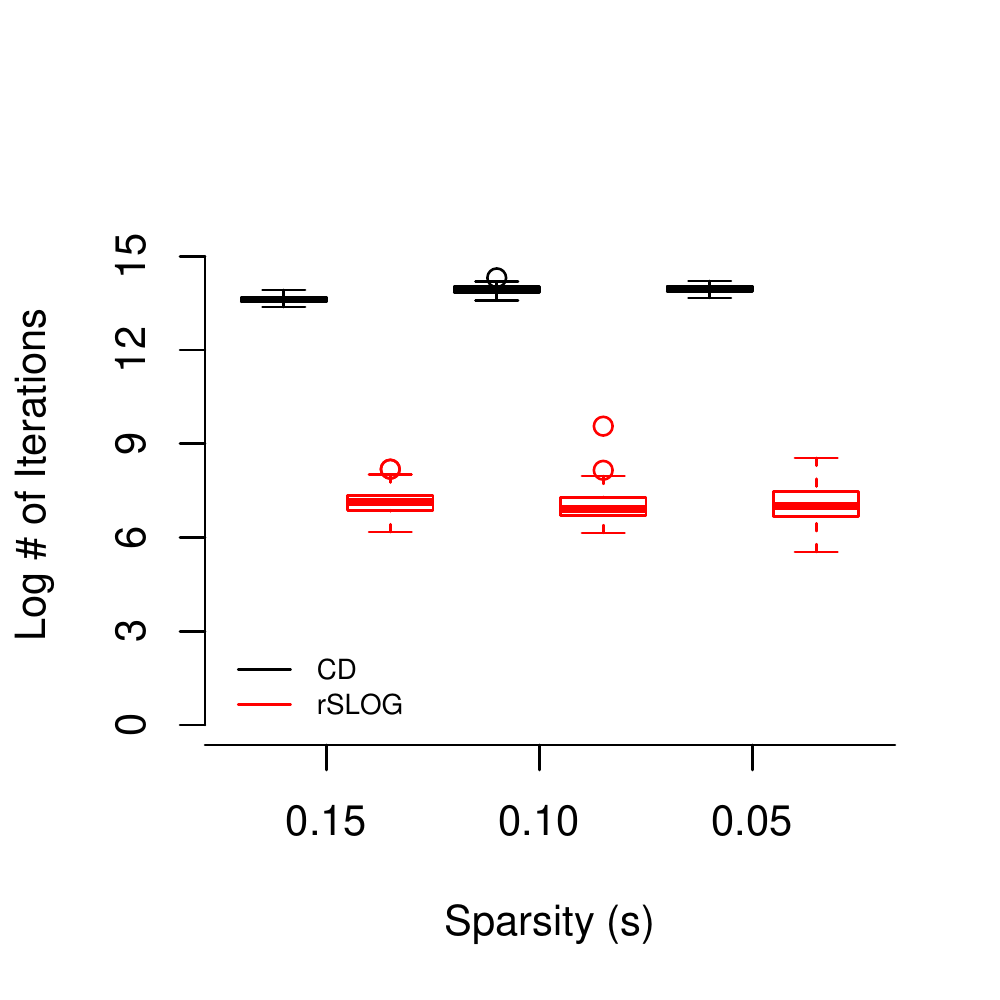} \label{F:box2c}}
\hfill
\subfigure[Correlation $\rho=0.99$]{\includegraphics[scale=0.8]{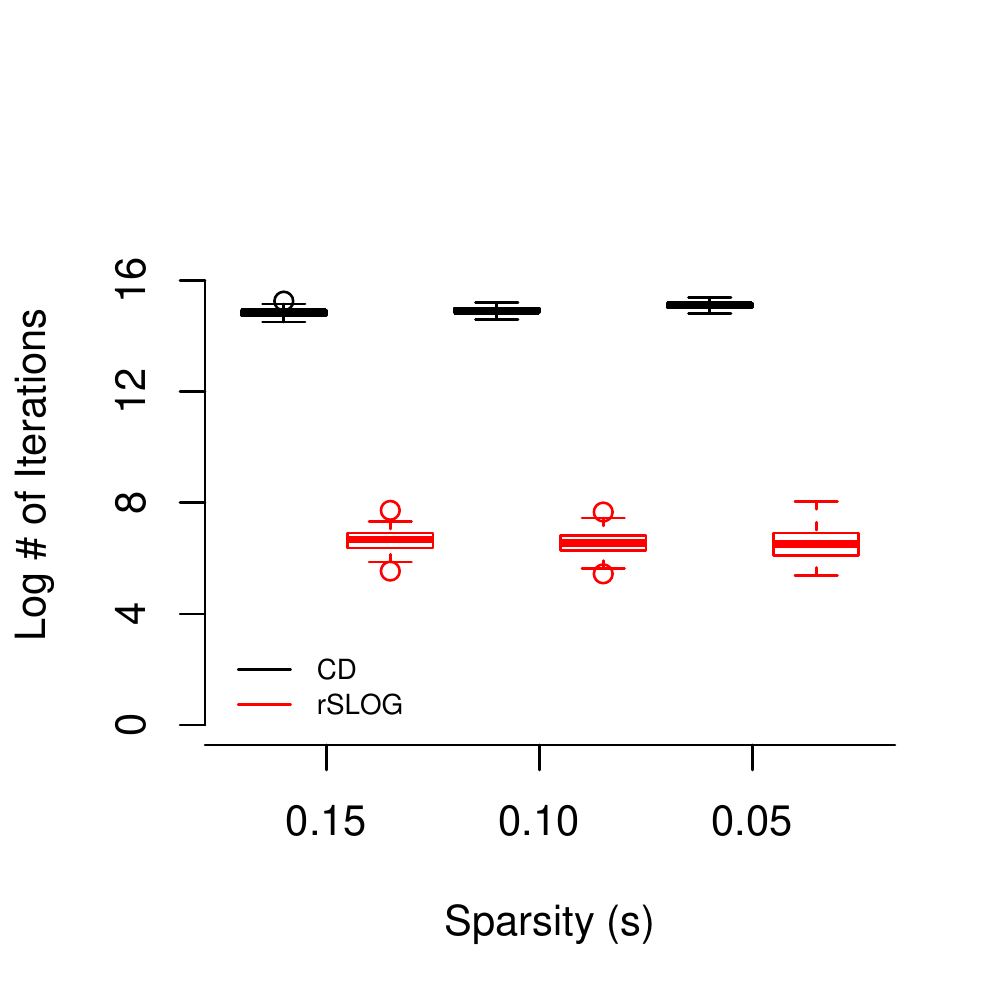} \label{F:box2d}}
\caption{\label{F:box2} Time in seconds and number of iterations until convergence of the CD and rSLOG algorithms applied to data generated using relationship \eqref{Esim1} with $n=200$, $p=1000$. Each box plot is based on 100 simulated datasets.}
\end{figure}

\textcolor{black} {A comparison of the timings and \textcolor{black} {number of iterations of} the rSLOG and CD algorithms applied to the setting where $p=1000$ and $n=200$ are given in Figure~\ref{F:box2}. In high multicollinearity and/or low sparsity settings Figure~\ref{F:box2a} and Figure~\ref{F:box2b} illustrate that rSLOG can offer substantial increases in computational speed compared to CD. Moreover, the improvements in computation  offered by rSLOG in this higher-dimensional setting are more substantial than in the lower-dimensional setting of Figure~\ref{F:box1}. The reasons for this are twofold: (1) as $p$ increases the adverse effects of multicollinearity on the performance of CD are compounded; and (2) as both $n$ and $p$ increase a given level of sparsity corresponds to a larger number of non-zero coefficients. As previously discussed, both of these factors hinder the convergence of CD more so than for rSLOG. Once more the variability of the rSLOG computation times is substantially reduced compared to the CD computation times. The massive reduction in the number of iterations required by rSLOG compared to CD is clearly illustrated in Figure~\ref{F:box2c} and Figure~\ref{F:box2d}. In fact, this reduction in the number of iterations required by rSLOG compared to CD is more than sufficient to compensate for the additional computational cost of each of its iterations. This trade-off is the reason for the faster overall compute times observed for rSLOG compared to CD.}
\\ \\
\noindent {\it Remarks:}
\begin{enumerate}
\item \textcolor{black}{A detailed and thorough investigation of the robustness of the results in Figure~\ref{F:box2a} and Figure~\ref{F:box2b} under various regimes was undertaken. In particular, qualitatively similar findings to those in Figure~\ref{F:box2a} and Figure~\ref{F:box2b} were found under the following alternative specifications for the $\beta_j$ in relationship \eqref{Esim1}: (1) all the $\beta_j$ were set to 0.1 (or 0.5); (2) one-fifth of the $\beta_j$ were set to 0.1 (or 0.5) and the rest to zero; and (3) same as (1) but each $\beta_j$ was selected from a Uniform on [$-$0.1,0.1] distribution.}



\item A \textbf{glmnet} convergence threshold of 1e-13 is used to ensure close convergence of $\bm{\hat{\beta}}^{\SLOG}$, $\bm{\hat{\beta}}^{\rSLOG}$ and $\bm{\hat{\beta}}^{\CD}$ to $\bm{\beta}^{\lars}$. \textcolor{black}{In high multicollinearity and low sparsity settings smaller values of the \textbf{glmnet} convergence threshold tend to yield values of $\bm{\hat{\beta}}^{\CD}$ that can be undesirably far from $\bm{\beta}^{\lars}$.}


\item \textcolor{black} {The findings from our simulations suggest that in situations of high multicollinearity and low sparsity one may prefer to use rSLOG over CD. The proliferation of high dimensional data means that settings of both high multicollinearity and low sparsity are now a common occurrence. In particular, Section~\ref{sec:ISdata} provides an example of the application of rSLOG to a dataset with these two attributes. Of course, when the desired level of sparsity is unknown it would be beneficial to use both SLOG and CD depending on the value of $s$.}

\item \textcolor{black} {The primary focus of this paper has been on understanding the use of SLOG in settings of high multicollinearity and low sparsity.  Supplemental Section~\ref{sec:additional} demonstrates that rSLOG can provide improved computational speed compared to CD  even in settings of high sparsity or low multicollinearity. Supplemental Section~\ref{sec:additional} also  explores further the role played by the three factors:  sample size ($n$), sparsity ($s$), and multicollinearity ($\rho$), in the relative compute times of CD and rSLOG when $p = 1000$.}

\item Alternative approaches for thresholding coefficients to zero in the SLOG algorithm are available. One such approach is to run the SLOG algorithm for a few iterations and then threshold the coefficients with the smallest magnitude to zero. The number of coefficients thresholded to zero would commensurate with the desired sparsity of the solution. The advantage of this approach compared to rSLOG is that it does not require an a priori choice of the threshold value. The disadvantage compared to rSLOG \textcolor{black}{however} is that the initial iterations require the inversion of the full $p \times p$ covariate matrix.

\end{enumerate}
\subsubsection{Infrared Spectroscopy Data}
\label{sec:ISdata}

As a final comparison of SLOG, rSLOG and CD, the three algorithms are applied to an infrared spectroscopy dataset. The infrared spectroscopy data were collected during a study to determine whether near infrared (NIR) spectroscopy could be used to predict the composition of cookie dough \citep{osborne1984}. \textcolor{black}{The data used has $n=40$ cookie dough samples, and is available from the R Package \textbf{ppls} \citep{ppls}}. The response vector of length 40 is a measure of the fat content of each dough sample. The covariate data of dimension $(n=40) \times (p=700)$ contains the NIR reflectance spectrum of each dough sample, measured at 700 points. \textcolor{black}{The covariate data exhibits a high degree of multicollinearity with over 70\% of the pairwise correlations exceeding 0.90 and a median pairwise correlation of 0.96. The multicollinearity of the infrared spectroscopy data is therefore well within the range in which rSLOG performs well compared to CD.} In the analysis that follows we investigate timing comparisons for the infrared spectroscopy data over varies levels of sparsity $s$. The level of sparsity is set by selecting a value of the lasso regularization parameter ($\lambda$) that gives the required number of exact zero coefficient estimates.

\begin{table}
\begin{small}
\begin{center}
\begin{threeparttable}
\caption{Time in seconds until convergence of the CD, SLOG and rSLOG algorithms applied to the infrared spectroscopy data over a range of sparsity values ($s$) for $\bm{\beta}^{\lars}$.}
\label{T:sim3}
\begin{footnotesize}
\begin{tabular}{rrrrrrrr}
\hline
& \multicolumn{5}{c}{Time (secs)} & &  \\
\cmidrule(lr){2-6}
 &  & \multicolumn{3}{c}{rSLOG} &  &  &  \\
\cmidrule(lr){3-5}
\multicolumn{1}{c}{$s$}& \multicolumn{1}{c}{CD} & $\theta=$1e-10 & $\theta=$1e-13 & $\theta=$1e-16& \multicolumn{1}{c}{SLOG} & $\log(K^{\CD})$ & $\log(K^{\SLOG})$\\
\hline
0.95 & 0.03 & 86.18 & 99.94 & 85.48 & 296.90 & 9.67 & 12.61 \\
  0.90 & 0.03 & 92.43 & 91.76 & 89.62 & 352.33 & 10.00 & 12.59 \\
  0.75 & 11.74 & 5.24 & 5.76 & 6.91 & 191.08 & 14.81 & 8.00 \\
  0.50 & 15.90 & 8.13 & 9.41 & 10.23 & 202.10 & 14.94 & 9.64 \\
  0.25 & 62.57 & 3.80 & 5.91 & 6.16 & 123.49 & 15.98 & 6.00 \\
  0.15 & 81.32 & 4.35 & 5.80 & 6.38 & 83.74 & 16.17 & 5.56 \\
  0.10 & 99.48 & 4.60 & 5.58 & 6.77 & 91.34 & 16.32 & 5.71 \\
  0.05 & 105.05 & 4.77 & 5.21 & 7.49 & 186.12 & 16.34 & 6.89 \\
\hline
\hline
\end{tabular}
\end{footnotesize}
\begin{scriptsize}
\begin{tablenotes}[flushleft]
\item Note: $K^{\CD}$ and $K^{\SLOG}$ are the number of passes over the data/iterations until convergence for CD/SLOG, respectively.
\end{tablenotes}
\end{scriptsize}
\end{threeparttable}
\end{center}
\end{small}
\end{table}

Table~\ref{T:sim3} contains the timings for the SLOG, rSLOG and CD algorithms applied to the infrared spectroscopy data. \textcolor{black}{We first note the substantial reduction} in computational time for rSLOG compared to SLOG. The reason for this reduction is that at each iteration SLOG is inverting a $700 \times 700$ matrix. In contrast, rSLOG is inverting successively smaller matrices that eventually decrease to an approximate dimension of $(1-s)40 \times (1-s)40$. \textcolor{black}{Second we note} the robustness of the timings of the rSLOG algorithm with alternative thresholds \textcolor{black}{of $\theta=$1e-10 or $\theta=$1e-16.} However, the most important observation is the large reduction in the computational time of rSLOG compared to CD, over most of the range of sparsity values considered. The primary reason for the increased speed of rSLOG is the relatively small number of iterations the algorithm requires to converge. For example, when $s=0.05$ CD requires approximately 12 million iterations to converge as compared to approximately 1000 for rSLOG. \textcolor{black}{The use of cross-validation to find the ``optimal" level of sparsity for the infrared spectroscopy data suggests that $s$ close to 0 is optimal. This finding illustrates that for this data the sparsity of the lasso solution is well within the range of values where rSLOG performs better compared to CD. A plot of the cross-validation error versus $s$ is given in Supplemental Section~\ref{sec:AppC}.}

In addition to the infrared spectroscopy data, there are many other examples of real data with high multicollinearity. Examples occur naturally in practice, including: (1) data that exhibit high spatial correlation including measures of water contamination or temperature over a given region; (2) data that exhibit strong temporal correlation such as the daily share prices of stocks within the same asset class; and (3) in gene expression data it is common for the pairwise correlation between expression levels to be large and positive. The application of the SLOG or rSLOG algorithms to data such as these may offer similar benefits to those observed here for the infrared spectroscopy data. Further, even for data that do not exhibit high multicollinearity the SLOG/rSLOG algorithms can provide increased computational speed if a non-sparse solution is of interest. An example of such a dataset is the well known Diabetes data of dimension $n=442$ and $p=64$ analyzed in \citet{efron2004}.

\section{Conclusion}
\label{sec:conclusion}
In this paper we have proposed a novel algorithm (the
Deterministic Bayesian Lasso
algorithm) for computing the lasso solution. The algorithm is based on exploiting the structure of the Bayesian lasso and its corresponding Gibbs sampler. Our study of the
Deterministic Bayesian Lasso
algorithm yields important new theoretical and computational insights into the efficient computation of the lasso solution. Importantly, from a practical perspective the algorithm is shown to offer substantial increases in computational speed compared to coordinatewise algorithms, in settings of low sparsity and high multicollinearity.


\newpage

\bibliographystyle{ims}
\bibliography{references} 

\clearpage
\appendix
\begin{center} \section*{Supplemental Section} \end{center}
\section{Proofs}
\setcounter{page}{1}
\begin{proof}[%
Proof of Lemma~\ref{lem:separable}%
]
Suppose that $\bm{X}_{(1)}^T\bm{X}_{(2)}^{}$ equals the zero matrix, and partition $\bbeta=(\bbeta_{(1)},\bbeta_{(2)})$ similarly to $\bm{X}$.  Now observe that minimizing $\normof{\bm{y}-\bm{X}\bbeta}_2^2+2\lambda\normof{\bbeta}_1$ is equivalent to minimizing $\normof{\bm{y}-\bm{X}\bbeta}_2^2+2\lambda\normof{\bbeta}_1+\normof{\bm{y}}_2^2$, and write
\begin{align*}
\normof{\bm{y}-\bm{X}\bbeta}_2^2+2\lambda\normof{\bbeta}_1+\normof{\bm{y}}_2^2
&=2\bm{y}^T\bm{y}-2\bbeta^T\bm{X}^T\bm{y}+\bbeta^T\bm{X}^T\bm{X}\bbeta+2\lambda\normof{\bbeta}_1\\
&=2\bm{y}^T\bm{y}-2
\bbeta_{(1)}^T\bm{X}_{(1)}^T
\bm{y}-2
\bbeta_{(2)}^T\bm{X}_{(2)}^T
\bm{y}+\bbeta_{(1)}^T\bm{X}_{(1)}^T\bm{X}_{(1)}\bbeta_{(1)}\\
&\qquad+\bbeta_{(2)}^T\bm{X}_{(2)}^T\bm{X}_{(2)}\bbeta_{(2)}+2\lambda\normof{\bbeta_{(1)}}_1+2\lambda
\normof{\bbeta_{(2)}}_1\\
&=\normof{\bm{y}-\bm{X}_{(1)}\bbeta_{(1)}}_2^2+2\lambda\normof{\bbeta_{(1)}}_1\\
&\qquad+\normof{\bm{y}-\bm{X}_{(2)}\bbeta_{(2)}}_2^2+2\lambda\normof{\bbeta_{(2)}}_1,
\end{align*}
hence this may be minimized by minimizing over $\bbeta_{(1)}$ and $\bbeta_{(2)}$ separately.
\end{proof}

\begin{proof}[Proof of Lemma~\ref{lem:no-zeros}]
For each $j\ge1$, let $(\reals^+)^j=\{\bm b\in\reals^j:b_\ell>0\text{ for all }\ell\le j\}$, and let $\mu_j$ denote Lebesgue measure on $\reals^j$.  Now define the sets $\Omega^{(0)}=\{\bm b:\reals^p:b_j=0\text{ for some }j\}$ and
$\Omega^{(1)}=\{\bm b\in(\reals^+)^p:A(\bm b)\in\Omega^{(0)}\}$.  Then let $\Omega=\Omega^{(0)}\cup\Omega^{(1)}$.
The proof consists of two parts:
\begin{enumerate}
\item[(i)] First, we show that the function $A$ cannot map a subset of $(\reals^+)^p\setminus\Omega^{(1)}$ with positive $\mu_p$-measure to a set of $\mu_p$-measure zero.
(Since $A$ depends on the components of its argument only through their absolute values, it follows that $A$ cannot map a subset of $\reals^p\setminus\Omega$ with positive $\mu_p$-measure to a set of $\mu_p$-measure zero.)
\item[(ii)] Second, we show that $\mu_p(\Omega)=0$, which establishes that $P_0(\bm b^{(0)}\in\Omega)=0$.
\end{enumerate}
By induction, (i) and (ii) together imply that $P_0(\bm b^{(k)}\in\Omega)=0$ for all $k\ge1$.  This result establishes the lemma since $\Omega^{(0)}\subseteq\Omega$.

We first show~(i).
Let $\tilde A$ denote the restriction of $A$ to the domain $(\reals^+)^p\setminus\Omega^{(1)}$.  Now let $\bm b\in(\reals^+)^p\setminus\Omega^{(1)}$, and let $\bm a=\tilde A(\bm b)$.  Then $a_j\ne0$ for each $j$, so we may write
\[
b_j=\frac{\lambda a_j}{\bm e_j^T\bm X^T
(\bm y-\bm X\bm a)
}\quad\text{ for all }j,
\]
where $\bm e_j$ denotes the $j$th unit vector of length $p$.  Now observe that
\begin{align*}
\frac{\partial b_j}{\partial a_\ell}
=\frac{\lambda a_j\bm e_j^T\bm X^T\bm X\bm e_\ell}{\left[\bm e_j^T\bm X^T(\bm y-\bm X\bm a)\right]^2}+\frac{\lambda\,I(j=\ell)}{\bm e_j^T\bm X^T(\bm y-\bm X\bm a)}
=
\frac{b_j^2}{\lambda a_j}
\left[\bm e_j^T\bm X^T\bm X\bm e_\ell+\frac{\lambda}{b_j}\,I(j=\ell)\right],
\end{align*}
and thus
\[
\frac{\partial\bm b}{\partial \bm a}=\frac{1}{\lambda}\Diag\!\left(\frac{b_1^2}{a_1},\ldots,\frac{b_p^2}{a_p}\right)\left[\bm X^T\bm X+\lambda\Diag\!\left(\frac{1}{b_1},\ldots,\frac{1}{b_p}\right)\right].
\]
Then
\[
\left|\frac{\partial\bm b}{\partial \bm a}\right|=\frac{1}{\lambda^p}\left(\prod_{j=1}^p\frac{b_j^2}{a_j}\right)\left|\bm X^T\bm X+\lambda\Diag\!\left(\frac{1}{b_1},\ldots,\frac{1}{b_p}\right)\right|\ne0,
\]
where $|\cdot|$ denotes the determinant.  It follows that $\tilde A$ cannot map any set with positive $\mu_p$-measure to a set of $\mu_p$-measure zero.  This establishes the first part of the proof since $\tilde A$ and $A$ coincide on $(\reals^+)^p\setminus\Omega^{(1)}$ by the definition of $\tilde A$.

Second, we show~(ii).
Clearly $\mu_p(\Omega^{(0)})=0$, so it suffices to show that $\mu_p(\Omega^{(1)})=0$.  For any nonempty set $M\subseteq\{1,\ldots,p\}$, define $\Omega^{(0)}_M=\{\bm b\in\reals^p:b_j=0\text{ if and only if }j\in M\}$ and $\Omega^{(1)}_M=\{\bm b\in(\reals^+)^p:A(\bm b)\in\Omega^{(0)}_M\}$.  Clearly $\Omega^{(1)}=\bigcup_{M\in\mathcal M}\Omega^{(1)}_M$, where $\mathcal M$ is the class of all nonempty subsets of $\{1,\ldots,p\}$.  Since $\mathcal M$ is finite, it suffices to show that $\mu_p(\Omega^{(1)}_M)=0$ for each $M\in\mathcal M$.  Note that if $M=\{1,\ldots,p\}$, then
$\bm b\in\Omega^{(1)}_M$ implies $A(\bm b)=\bm 0$, which implies by the form of $A$ that $\bm b=\bm0\notin(\reals^+)^p$, a contradiction.  Hence, $\Omega^{(1)}_M=\varnothing$, so $\mu_p(\Omega^{(1)}_M)=0$.  Thus, instead let $M\in\mathcal M\setminus\{1,\ldots,p\}$, and let $\bm b\in\Omega^{(1)}_M$ and $\bm a=A(\bm b)$.  For convenience of notation, we may assume without loss of generality that $M=\{1,\ldots,m\}$, where $1\le m\le p-1$.  Then partition $\bm a$, $\bm b$, and $\bm X$ as
\[
\bm a=\begin{bmatrix}\bm0\\\bm a_\star\end{bmatrix},\qquad
\bm b=\begin{bmatrix}\bm b_0\\\bm b_\star\end{bmatrix},\qquad
\bm X=\begin{bmatrix}\bm X_0&\bm X_\star\end{bmatrix},
\]
where $\bm a_\star$ and $\bm b_\star$ have length $p-m$ and $\bm X_\star$ has $p-m$ columns.  Then since $a_j=0$ for all $1\le j\le m$, we have
$\bm X_0^T(\bm y-\bm X\bm a)=\bm X_0^T(\bm y-\bm X_\star\bm a_\star)=\bm0$, which implies that
\begin{align}
\bm X_0^T\bm X_\star\bm a_\star=\bm X_0^T\bm y.\label{zero-constraint}
\end{align}
If $\bm{X}_0^T\bm{X}_\star=\bm{0}$, then $\bm{X}_0^T\bm{y}\ne0$ by
Assumption~2,
which creates a contradiction.  So instead $\bm{X}_0^T\bm{X}_\star\ne\bm{0}$, and thus (\ref{zero-constraint}) imposes at least one linear constraint on $\bm a_\star$.  Then $\bm a_\star\in H_\star$ for some hyperplane $H_\star\subset\reals^{p-m}$.  Next, since $a_j\ne0$ for all $m+1\le j\le p$, we have
\[
b_j=\frac{\lambda a_j}{\bm e_j^T\bm X^T(\bm y-\bm X\bm a)}=\frac{\lambda a_j}{\bm e_j^T\bm X^T(\bm y-\bm X_\star\bm a_\star)}\quad\text{ for all }m+1\le j\le p.
\]
Thus, the components of $\bm b_\star$ depend only on $\bm a_\star$, while the components of $\bm b_0$ are unrestricted.  Then we have
\begin{align}
\Omega^{(1)}_M=(\reals^+)^m\times\{\bm b_\star\in(\reals^+)^{p-m}:A_\star(b_\star)\in H_\star\setminus\Omega^{(0)}_\star\},
\label{unconstrained-constrained}
\end{align}
where $\Omega^{(0)}_\star=\{\bm b\in\reals^{p-m}:b_j=0\text{ for some }j\}$ and the function $A_\star:\reals^{p-m}\to\reals^{p-m}$ is defined analogously to $A$, with $\bm X$ replaced by $\bm X_\star$.  Now define
$\Omega^{(1)}_\star=\{\bm b_\star\in(\reals^+)^{p-m}:A_\star(\bm b_\star)\in\Omega^{(0)}_\star\}$ and $\Omega_\star=\Omega^{(0)}_\star\cup\Omega^{(1)}_\star$.
Then in (\ref{unconstrained-constrained}), we have
\[
\{\bm b_\star\in(\reals^+)^{p-m}:A_\star(b_\star)\in H_\star\setminus\Omega^{(0)}_\star\}\subset(\reals^+)^{p-m}\setminus\Omega_\star.
\]
Then we may use the same argument as in the first part of the proof to show that $A_\star$ cannot map a subset of $(\reals^+)^{p-m}\setminus\Omega_\star$ with positive $\mu_{p-m}$-measure to a set of $\mu_{p-m}$-measure zero.  Then since $\mu_{p-m}(H_\star)=0$, we have
\[
\mu_{p-m}\!\left(\{\bm b_\star\in(\reals^+)^{p-m}:A_\star(b_\star)\in H_\star\setminus\Omega^{(0)}_\star\}\right)=0,
\]
and hence $\mu_p(\Omega^{(1)}_M)=0$.
\end{proof}

\begin{proof}[%
Proof of Theorem~\ref{thm:convergence}%
]
Let $\Psi=\left\{\bm b\in\reals^p:
A(\bm b)=\bm b
\right\}$ denote the set of all fixed points of~$A$,
and recall from Lemma~\ref{lem:other-fixed} that $|\Psi|\le2^p$.
Then let
$\delta=\min_{\bm b_1,\bm b_2\in\Psi}\|\bm b_1-\bm b_2\|_2>0$ denote the smallest distance between fixed points.
Also, for each $k\ge1$, let $\bm b^{(k)}_\Psi=\argmin_{\bm b\in\Psi}\|\bm b^{(k)}-\bm b\|_2$ denote the fixed point closest to $\bm b^{(k)}$.  (If such a point is not unique, the choice may be made arbitrarily among all such closest fixed points.)
The proof now proceeds in several steps.
First, we establish that $\|\bm b^{(k)}-\bm b^{(k)}_\Psi\|_2\to0$ as $k\to\infty$, i.e., the distance between $\bm b^{(k)}$ and the nearest fixed point~$\bm b^{(k)}_\Psi$ tends to zero.  Second, we demonstrate that
$\bm b^{(k)}_\Psi=\bm b_\Psi$ for all sufficiently large~$k$, and thus that $\bm b^{(k)}\to\bm b_\Psi$.  Third, we show that with $P_0$-probability~$1$, $\bm b_\Psi\notin\Psi\setminus\{\bbhl\}$, which implies that $\bm b_\Psi=\bbhl$.

We begin by establishing that $\|\bm b^{(k)}-\bm b^{(k)}_\Psi\|_2\to0$.
Let $q:\reals^p\to\reals$ be the function such that
$q(\bm b)=Q[A(\bm b)]-Q(\bm b)$ for all $\bm b\in\reals^p$.
By Lemma~\ref{lem:monotone}, $q(\bm b)=0$ if and only if $\bm b\in\Psi$.
Also, since $Q\!\left[\bm b^{(k)}\right]$ converges as $k\to\infty$ by Lemma~\ref{lem:monotone}, it follows that
$q\!\left[\bm b^{(k)}\right]\to0$
as $k\to\infty$.
Now observe that $q$ is continuous since it is a composition of matrix multiplication, vector norm--taking, and addition.
Also observe that $q(\bm b)\to\infty$ as $\|\bm b\|_2\to\infty$.
Then we have
$\|\bm b^{(k)}-\bm b^{(k)}_\Psi\|_2\to0$
as $k\to\infty$.

We now demonstrate that $\bm b^{(k)}_\Psi=\bm b_\Psi$ for all sufficiently large~$k$.
Note that $A:\reals^p\to\reals^p$ is continuous, which follows immediately from the continuity of matrix addition, multiplication, inversion, and square root--taking, of which the function $A$ is a composition.
Now let $\varepsilon>0$.  Then there exists $\delta>0$ such that for each $\bm b\in\Psi$,
$\|A(\bm\beta)-A(\bm b)\|_2\le\varepsilon/4$
for all $\bm\beta\in\reals^p$ such that $\|\bm\beta-\bm b\|_2\le\delta$.
Thus, if $\bm b\in\Psi$ and $\|\bm\beta-\bm b\|_2\le\min\{\delta,\,\varepsilon/4\}$, then
\[
\|A(\bm\beta)-\bm\beta\|_2
\le
\|A(\bm\beta)-\bm b\|_2+\|\bm b-\bm\beta\|_2
=
\|A(\bm\beta)-A(\bm b)\|_2+\|\bm b-\bm\beta\|_2\le\frac{\varepsilon}{4}+\frac{\varepsilon}{4}=\frac{\varepsilon}{2}.
\]
Next, there exists $K\ge1$ such that
$\|\bm b^{(k)}-\bm b^{(k)}_\Psi\|_2\le\min\{\delta,\,\varepsilon/4\}$
for all $k\ge K$.
Then
\begin{align*}
\left\|\bm b^{(k+1)}_\Psi-\bm b^{(k)}_\Psi\right\|_2
&\le
\left\|\bm b^{(k+1)}_\Psi-\bm b^{(k+1)}_{}\right\|_2
+\left\|\bm b^{(k+1)}_{\phantom{\$}}-\bm b^{(k)}_{\phantom{\$}}\right\|_2
+\left\|\bm b^{(k)}_{}-\bm b^{(k)}_\Psi\right\|_2\\
&\le
\frac{\varepsilon}{4}+\left\|A\!\left[\bm b^{(k)}\right]-\bm b^{(k)}\right\|_2+\frac{\varepsilon}{4}\le
\varepsilon.
\end{align*}
Hence $\|\bm b^{(k+1)}_\Psi-\bm b^{(k)}_\Psi\|\to0$, which implies that $\bm b^{(k+1)}_\Psi=\bm b^{(k)}_\Psi=\bm b_\Psi$ for all sufficiently large~$k$ since $\Psi$ contains only finitely many points.  Thus, $\bm b^{(k)}\to\bm b_\Psi\in\Psi$ as $k\to\infty$.  For notational convenience, we will henceforth simply write $\bm b^{(k)}\to\bm b\in\Psi$.

Third and finally, we now show that with $P_0$-probability~1,
$\bm b\notin\Psi\setminus\{\bbhl\}$, noting that $\bbhl\in\Psi$ by Lemma~\ref{lem:lasso-fixed}.  This step constitutes the remainder of the proof.
First,
Lemmas~\ref{lem:other-fixed}~and~\ref{lem:no-zeros} imply
that
\[
P_0\!\left(\bm{b}^{(k)}\in\Psi\setminus\{\bbhl\}\text{ for some }k\ge1\right)=0,
\]
noting that $\bm{b}\in\Psi\setminus\{\bbhl\}$ only if at least one component of $\bm{b}$ is zero.
Now suppose $\bm{b}^{(k)}\notin\Psi$ for every $k\ge1$, and also suppose $\bm{b}^{(k)}\to\bm{b}\in\Psi\setminus\{\bbhl\}$ as $k\to\infty$.  We now demonstrate that this leads to a contradiction by establishing that there exists $J\in\oneto{m}$ such that
$|b^{(k+1)}_J-b_J^{}|\ge|b^{(k)}_J-b_J^{}|$ for all sufficiently large~$k$.
As before, we may assume without loss of generality that $b_j=0$ for each $j\in\oneto{m}$ and $b_j\ne0$ for each $j\in\{m+1,\ldots,p\}$, where $0\le m\le p$.  Then partition the vector~$\bm{b}$ as $\bm{b}=(\bm{0}_m,\bm{b}_\star)$.
Now note that $A(\bm{b})=\bm{b}$ since $\bm{b}\in\Psi$, and hence
\begin{align}
\bm{b}_\star
=\left(\bm{X}_\star^T\bm{X}_\star+\lambda\bm{B}_\star^{-1}\right)^{-1}\bm{X}_\star^T\bm{y},
\label{accidental-fixed}
\end{align}
by Lemma~\ref{lem:partition}, where we have partitioned $\bm{X}$ into $\bm{X}_0$ and $\bm{X}_\star$ as before and defined the diagonal matrix
$
\bm{B}_\star
=\Diag(|b_{m+1}|,\ldots,|b_p|)$.  Now note that since $\bm{b}\ne\bbhl$ and 
Assumption~1
holds, there exists $\bbeta\in\reals^p$ such that $Q(\bbeta)>Q(\bm{b})$.  Now partition $\bbeta=(\bbeta_0,\bbeta_\star)$, where $\bbeta_0$ has length~$m$, and note that
\begin{align*}
Q(\bbeta)-Q(\bm{b})
&=
-\bbeta^T\bm{X}^T\bm{X}\bbeta+2\bbeta^T\bm{X}^T\bm{y}-2\lambda\normof{\bbeta}_1
+\bm{b}^T\bm{X}^T\bm{X}\bm{b}-2\bm{b}^T\bm{X}^T\bm{y}+2\lambda\normof{\bm{b}}_1\\
&=
-\bbeta_0^T\bm{X}_0^T\bm{X}_0\bbeta_0-\bbeta_\star^T\bm{X}_\star^T\bm{X}_\star\bbeta_\star-2\bbeta_0^T\bm{X}_0^T\bm{X}_\star\bbeta_\star+2\bbeta_0^T\bm{X}_0^T\bm{y}+2\bbeta_\star^T\bm{X}_\star^T\bm{y}\\
&\qquad-2\lambda\bbeta_0^T\sign\bbeta_0-2\lambda\bbeta_\star^T\sign\bbeta_\star+\bm{b}_\star^T\bm{X}_\star^T\bm{X}_\star\bm{b}_\star-2\bm{b}_\star^T\bm{X}_\star^T\bm{y}+2\lambda\bm{b}_\star^T\sign\bm{b}_\star\\
&=
-\bbeta_0^T\bm{X}_0^T\bm{X}_0\bbeta_0
-2\bbeta_0^T\left[\lambda\sign\bbeta_0-\bm{X}_0^T(\bm{y}-\bm{X}_\star\bbeta_\star)\right]
-\bbeta_\star^T\bm{X}_\star^T\bm{X}_\star\bbeta_\star\\
&\qquad-2\bbeta_\star^T(\lambda\sign\bbeta_\star-\bm{X}_\star^T\bm{y})
+\bm{b}_\star^T\bm{X}_\star^T\bm{X}_\star\bm{b}_\star+2\bm{b}_\star^T(\lambda\sign\bm{b}_\star-\bm{X}_\star^T\bm{y}),
\end{align*}
where for any vector~$\bm{a}$, we write $\sign\bm{a}$ to denote the vector with components equal to the signs of the components of~$\bm{a}$.
Adding and subtracting $2\bbeta_0^T\bm{X}_0^T\bm{X}_\star\bm{b}_\star$ yields
\begin{align*}
Q(\bbeta)-Q(\bm{b})
&=
-\bbeta_0^T\bm{X}_0^T\bm{X}_0\bbeta_0
-2\bbeta_0^T\left[\lambda\sign\bbeta_0-\bm{X}_0^T(\bm{y}-\bm{X}_\star\bm{b}_\star)\right]
-\bbeta_\star^T\bm{X}_\star^T\bm{X}_\star\bbeta_\star\\
&\qquad-2\bbeta_\star^T(\lambda\sign\bbeta_\star-\bm{X}_\star^T\bm{y})
+\bm{b}_\star^T\bm{X}_\star^T\bm{X}_\star\bm{b}_\star+2\bm{b}_\star^T(\lambda\sign\bm{b}_\star-\bm{X}_\star^T\bm{y})\\
&\qquad+2\bbeta_0^T\bm{X}_0^T\bm{X}_\star(\bm{b}_\star-\bbeta_\star).
\end{align*}
Now observe that
\begin{align*}
\left|2\bbeta_0^T\bm{X}_0^T\bm{X}_\star(\bm{b}_\star-\bbeta_\star)\right|
&\le
2\normof{\bm{X}_0\bbeta_0}_2\normof{\bm{X}_\star(\bm{b}_\star-\bbeta_\star)}_2\\
&\le
\normof{\bm{X}_0\bbeta_0}_2^2+\normof{\bm{X}_\star(\bm{b}_\star-\bbeta_\star)}_2^2\\
&=\bbeta_0^T\bm{X}_0^T\bm{X}_0\bbeta_0+\bm{b}_\star^T\bm{X}_\star^T\bm{X}_\star\bm{b}_\star+\bbeta_\star^T\bm{X}_\star^T\bm{X}_\star\bbeta_\star-2\bbeta_\star^T\bm{X}_\star^T\bm{X}_\star\bm{b}_\star
\end{align*}
by the Cauchy-Schwarz inequality and the fact that the geometric mean of two nonnegative real numbers is bounded above by the arithmetic mean.  Then
\begin{align*}
Q(\bbeta)-Q(\bm{b})
&\le
-\bbeta_0^T\bm{X}_0^T\bm{X}_0\bbeta_0
-2\bbeta_0^T\left[\lambda\sign\bbeta_0-\bm{X}_0^T(\bm{y}-\bm{X}_\star\bm{b}_\star)\right]
-\bbeta_\star^T\bm{X}_\star^T\bm{X}_\star\bbeta_\star\\
&\qquad-2\bbeta_\star^T(\lambda\sign\bbeta_\star-\bm{X}_\star^T\bm{y})
+\bm{b}_\star^T\bm{X}_\star^T\bm{X}_\star\bm{b}_\star+2\bm{b}_\star^T(\lambda\sign\bm{b}_\star-\bm{X}_\star^T\bm{y})\\
&\qquad+\bbeta_0^T\bm{X}_0^T\bm{X}_0\bbeta_0+\bm{b}_\star^T\bm{X}_\star^T\bm{X}_\star\bm{b}_\star+\bbeta_\star^T\bm{X}_\star^T\bm{X}_\star\bbeta_\star-2\bbeta_\star^T\bm{X}_\star^T\bm{X}_\star\bm{b}_\star\\
&=
-2\bbeta_0^T\left[\lambda\sign\bbeta_0-\bm{X}_0^T(\bm{y}-\bm{X}_\star\bm{b}_\star)\right]-2\bbeta_\star^T(\lambda\sign\bbeta_\star+\bm{X}_\star^T\bm{X}_\star\bm{b}_\star-\bm{X}_\star^T\bm{y})\\
&\qquad+2\bm{b}_\star^T(\lambda\sign\bm{b}_\star+\bm{X}_\star^T\bm{X}_\star\bm{b}_\star-\bm{X}_\star^T\bm{y}).
\end{align*}
It follows from~(\ref{accidental-fixed}) that
$
\bm{X}_\star^T\bm{X}_\star\bm{b}_\star-\bm{X}_\star^T\bm{y}=-\lambda\sign\bm{b}_\star
$,
and thus
\begin{align*}
Q(\bbeta)-Q(\bm{b})
&\le
-2\bbeta_0^T\left[\lambda\sign\bbeta_0-\bm{X}_0^T(\bm{y}-\bm{X}_\star\bm{b}_\star)\right]-2\lambda\bbeta_\star^T(\sign\bbeta_\star-\sign\bm{b}_\star)\\
&\le
-2\bbeta_0^T\left[\lambda\sign\bbeta_0-\bm{X}_0^T(\bm{y}-\bm{X}_\star\bm{b}_\star)\right]\\
&=-2\sum_{j=1}^m\left\{\lambda\left|\left[\beta_0\right]_j\right|-\left[\beta_0\right]_j\left[\bm{X}_0^T(\bm{y}-\bm{X}_\star\bm{b}_\star)\right]_j\right\}\\
&\le-2\sum_{j=1}^m\left\{\lambda\left|\left[\beta_0\right]_j\right|-\left|\left[\beta_0\right]_j\right|\left|\left[\bm{X}_0^T(\bm{y}-\bm{X}_\star\bm{b}_\star)\right]_j\right|\right\}\\
&=-2\sum_{j=1}^m\left|\left[\beta_0\right]_j\right|\left\{\lambda-\left|\left[\bm{X}_0^T(\bm{y}-\bm{X}_\star\bm{b}_\star)\right]_j\right|\right\}.
\end{align*}
However, since $Q(\bbeta)>Q(\bm{b})$, this last expression must be strictly positive, and hence we must have
\begin{align}
\left|\left[\bm{X}_0^T(\bm{y}-\bm{X}_\star\bm{b}_\star)\right]_{J}\right|>\lambda
\label{J-contradiction}
\end{align}
for some $J\in\oneto{m}$.

Now let $\varepsilon>0$ be small, and note that the suppositions imply that for any sufficiently large~$k$, we may write $\bm{b}^{(k)}=\bm{b}+\varepsilon
\bm{r}
$, where $|
r_j
|\le1$ for each $j\in\oneto{p}$ with at least one $
r_j
\ne0$.
(Note that $\bm{r}$ depends on~$k$, as do the objects defined below that depend on~$\bm{r}$, but we suppress the $k$ for brevity of notation.)
Now partition the vector~$\bm{r}$ as $\bm{r}=(\bm{u},\bm{v})$, where the vector~$\bm{u}$ has length~$m$, so that we may write $\bm{b}^{(k)}=(\varepsilon\bm{u},\;\bm{b}_\star+\varepsilon\bm{v})$.
Since
Assumption~2
holds, we may apply Lemma~\ref{lem:no-zeros}.  Then each component of~$\bm{b}^{(k)}$ is nonzero with $P_0$-probability~1, so from this point we will assume that this is the case, with the understanding that all subsequent statements are to be taken to hold with $P_0$-probability~1.
Then by~(\ref{sequence}) we may write
\begin{align*}
\bm{b}^{(k+1)}
&=
\left[\bm{X}^T\bm{X}+\lambda\left(\bm{B}^{(k)}\right)^{-1}\right]^{-1}\bm{X}^T\bm{y}
\end{align*}
where the diagonal matrix~$\bm{B}^{(k)}$ is defined as before.  Now partition $\bm{X}$ into $\bm{X}_0$ and $\bm{X}_\star$ as before and define diagonal matrices $\bm{U}$, $\bm{V}$, and $\bm{B}_\star$ whose diagonal elements are the absolute values of the components of the vectors $\bm{u}$, $\bm{v}$, and $\bm{b}_\star$, respectively.  Then we have
\begin{align*}
\bm{b}^{(k+1)}
&=
\begin{bmatrix}
\bm{X}_0^T\bm{X}_0+\lambda\varepsilon^{-1}\bm{U}^{-1} & \bm{X}_0^T\bm{X}_\star\\
\bm{X}_\star^T\bm{X}_0 & \bm{X}_\star^T\bm{X}_\star+\lambda\left(\bm{B}_\star+\varepsilon\bm{V}\right)^{-1}
\end{bmatrix}^{-1}
\begin{bmatrix}
\bm{X}_0^T\bm{y}\\ \bm{X}_\star^T\bm{y}
\end{bmatrix}\\
&=
\begin{bmatrix}
\bm{W}_{00}&\bm{W}_{0\star}\\ \bm{W}_{0\star}^T&\bm{W}_{\star\star}
\end{bmatrix}
\begin{bmatrix}
\bm{X}_0^T\bm{y}\\ \bm{X}_\star^T\bm{y}
\end{bmatrix}
=\begin{bmatrix}\bm{W}_{00}\bm{X}_0^T\bm{y}+\bm{W}_{0\star}\bm{X}_\star^T\bm{y}\\
\bm{W}_{0\star}^T\bm{X}_\star^T\bm{y}+\bm{W}_{\star\star}\bm{X}_\star^T\bm{y}\end{bmatrix},
\end{align*}
where
\begin{align*}
\bm{W}_{00}&=\left(\bm{X}_0^T\bm{X}_0+\lambda\varepsilon^{-1}\bm{U}^{-1}-\bm{X}_0^T\bm{X}_\star\bm{Z}\bm{X}_\star^T\bm{X}_0\right)^{-1},\\
\bm{W}_{0\star}&=-\bm{W}_{00}\bm{X}_0^T\bm{X}_\star\bm{Z},\\
\bm{W}_{\star\star}&=\bm{Z}+\bm{Z}\bm{X}_\star^T\bm{X}_0\bm{W}_{00}\bm{X}_0^T\bm{X}_\star\bm{Z},\\
\bm{Z}&=\left[\bm{X}_\star^T\bm{X}_\star+\lambda\left(\bm{B}_\star+\varepsilon\bm{V}\right)^{-1}\right]^{-1}.
\end{align*}
Now observe that
\begin{align*}
\bm{Z}&=\left[\bm{X}_\star^T\bm{X}_\star+\lambda\bm{B}_\star^{-1}-\lambda\varepsilon\bm{B}_\star^{-1}\bm{V}\bm{B}_\star^{-1}\right]^{-1}+O(\varepsilon^2)\\
&=\left(\bm{X}_\star^T\bm{X}_\star+\lambda\bm{B}_\star^{-1}\right)^{-1}
+\lambda\varepsilon\left(\bm{X}_\star^T\bm{X}_\star+\lambda\bm{B}_\star^{-1}\right)^{-1}\bm{B}_\star^{-1}\bm{V}\bm{B}_\star^{-1}
\left(\bm{X}_\star^T\bm{X}_\star+\lambda\bm{B}_\star^{-1}\right)^{-1}
+O(\varepsilon^2),\\
&=\tilde{\bm{Z}}+\lambda\varepsilon\tilde{\bm{Z}}\bm{B}_\star^{-1}\bm{V}\bm{B}_\star^{-1}\tilde{\bm{Z}}+O(\varepsilon^2),
\end{align*}
where $\tilde{\bm{Z}}=\left(\bm{X}_\star^T\bm{X}_\star+\lambda\bm{B}_\star^{-1}\right)^{-1}$.  Thus,
\begin{align*}
\bm{W}_{00}&=\varepsilon\left(\lambda\bm{U}^{-1}+\varepsilon\bm{X}_0^T\bm{X}_0-\varepsilon\bm{X}_0^T\bm{X}_\star\bm{Z}\bm{X}_\star^T\bm{X}_0\right)^{-1}\\
&=\varepsilon\left[\lambda\bm{U}^{-1}+\varepsilon\bm{X}_0^T\bm{X}_0-\varepsilon\bm{X}_0^T\bm{X}_\star\tilde{\bm{Z}}\bm{X}_\star^T\bm{X}_0\right]^{-1}+
O(\varepsilon^2)
\\
&=\lambda^{-1}\varepsilon\bm{U}+O(\varepsilon^2),\\
\bm{W}_{0\star}&=-\lambda^{-1}\varepsilon\bm{U}\bm{X}_0^T\bm{X}_\star
\tilde{\bm{Z}}+O(\varepsilon^2),\\
\bm{W}_{\star\star}&=\tilde{\bm{Z}}
+\lambda\varepsilon\tilde{\bm{Z}}\bm{B}_\star^{-1}\bm{V}\bm{B}_\star^{-1}
\tilde{\bm{Z}}+\lambda^{-1}\varepsilon\tilde{\bm{Z}}\bm{X}_\star^T\bm{X}_0\bm{U}\bm{X}_0^T\bm{X}_\star\tilde{\bm{Z}}+O(\varepsilon^2)\\
&=\tilde{\bm{Z}}+\varepsilon\tilde{\bm{Z}}
\left(\lambda\bm{B}_\star^{-1}\bm{V}\bm{B}_\star^{-1}+\lambda^{-1}\bm{X}_\star^T\bm{X}_0\bm{U}\bm{X}_0^T\bm{X}_\star\right)
\tilde{\bm{Z}}+O(\varepsilon^2),
\end{align*}
and therefore the first $m$ components of $\bm{b}^{(k+1)}$ are
\begin{align*}
(b^{(k+1)}_1,\ldots,b^{(k+1)}_m)=
\bm{W}_{00}\bm{X}_0^T\bm{y}+\bm{W}_{0\star}\bm{X}_\star^T\bm{y}
&=\lambda^{-1}\varepsilon\bm{U}\bm{X}_0^T\bm{y}-\lambda^{-1}\varepsilon\bm{U}\bm{X}_0^T\bm{X}_\star
\tilde{\bm{Z}}\bm{X}_\star^T\bm{y}+O(\varepsilon^2)\\
&=\lambda^{-1}\varepsilon\bm{U}\bm{X}_0^T(\bm{y}-\bm{X}_\star\bm{b}_\star)+O(\varepsilon^2),
\end{align*}
noting from~(\ref{accidental-fixed}) that $\tilde{\bm{Z}}\bm{X}_\star^T\bm{y}=\bm{b}_\star$.
Then for $j\in\oneto{m}$,
\begin{align*}
b_j^{(k+1)}=\lambda^{-1}\varepsilon|u_j|\left[\bm{X}_0^T(\bm{y}-\bm{X}_\star\bm{b}_\star)\right]_j+O(\varepsilon^2),
\end{align*}
where we write $\left[\bm{X}_0^T(\bm{y}-\bm{X}_\star\bm{b}_\star)\right]_j$ to denote the $j$th component of the vector in square brackets.
Then we may combine this with~(\ref{J-contradiction}) to yield that for all sufficiently small~$\varepsilon$,
\begin{align*}
\left|b_J^{(k+1)}-b_J^{}\right|=\left|b_J^{(k+1)}\right|=\lambda^{-1}\varepsilon|u_J|\left|\left[\bm{X}_0^T(\bm{y}-\bm{X}_\star\bm{b}_\star)\right]_J\right|>\varepsilon|u_J|=\left|b_J^{(k)}\right|=\left|b_J^{(k)}-b_J^{}\right|.
\end{align*}
This establishes the contradiction and completes the proof.
\end{proof}

\begin{proof}[%
Proof of Lemma~\ref{lem:one-d-form}%
]
Begin by noting from~(\ref{sequence1}) that if $b^{(0)}=0$, then $b^{(k)}=0$ for all $k\ge0$, and it can be seen that this agrees with the equation in the lemma.  Hence, we may assume that $b^{(0)}\ne0$ and rewrite equation in the lemma as
\begin{align}
b^{(k)}=\frac{c^k\lambda\sign\bh}{\lambda|b^{(0)}|^{-1}+n\sum_{m=0}^{k-1}c^m}.
\label{sequence1-simpler}
\end{align}
We now proceed by induction.
To establish the $k=1$ base case, note that (\ref{sequence1-simpler}) reduces to
\begin{align*}
b^{(1)}=\frac{n\bh}{\lambda|b^{(0)}|^{-1}+n}=\frac{|b^{(0)}|\;\bh}{\lambda/n+|b^{(0)}|},
\end{align*}
which is precisely~(\ref{sequence1}).
Now assume that~(\ref{sequence1-simpler}) holds for some $k\ge1$, and combine this with~(\ref{sequence1}) to obtain
\begin{align*}
b^{(k+1)}=\frac{|b^{(k)}|\;\bh}{\lambda/n+|b^{(k)}|}
&=\left(\frac{\lambda}{n}+\frac{c^k\lambda}{\lambda|b^{(0)}|^{-1}+n\sum_{m=0}^{k-1}c^m}\right)^{-1}\left(\frac{c^k\lambda}{\lambda|b^{(0)}|^{-1}+n\sum_{m=0}^{k-1}c^m}\right)\bh\\
&=\frac{nc^k\bh}{\lambda|b^{(0)}|^{-1}+n\sum_{m=0}^{k-1}c^m+nc^k}\\
&=\frac{c^{k+1}(n\bh/c)}{\lambda|b^{(0)}|^{-1}+n\sum_{m=0}^{k}c^m}=\frac{c^{k+1}\lambda\sign\bh}{\lambda|b^{(0)}|^{-1}+n\sum_{m=0}^{k}c^m},
\end{align*}
which is precisely~(\ref{sequence1-simpler}) as applied to $b^{(k+1)}$.
\end{proof}

\begin{proof}[%
Proof of Theorem~\ref{thm:one-d-properties}%
]
Statement~\ref{one-d-all-nonzero} is clear from inspection of the statement of Lemma~\ref{lem:one-d-form}.  To prove statement~\ref{one-d-all-same-side}, begin by noting that from statement~\ref{one-d-all-nonzero}, statement~\ref{one-d-all-same-side} holds trivially if either $\bh=0$ or $b^{(0)}=0$.  Hence, we may assume that $\bh$ and $b^{(0)}$ are both nonzero, in which case we have by~\ref{one-d-all-nonzero} that $\sign b^{(k)}=\sign\bh\ne0$ for all $k\ge1$. In light of this result, statement~\ref{one-d-all-same-side} is trivial if $\bhl{}=0$, so we may assume that $\bhl{}\ne0$.  Then $\sign b^{(k)}=\sign\bh=\sign\bhl{}\ne0$, and $|\bhl{}|=|\bh|-\lambda/n$.  Now note that if $|b^{(k)}|>|\bhl{}|$, then $|b^{(k)}|>|\bh|-\lambda/n$, and we have from the recursion relation~(\ref{sequence1}) that
\begin{align*}
|b^{(k+1)}|=\frac{|b^{(k)}|\;|\bh|}{|b^{(k)}|+\lambda/n}>|\bh|-\lambda/n=|\bhl{}|.
\end{align*}
If instead $|b^{(k)}|<|\bhl{}|$, then $|b^{(k)}|<|\bh|-\lambda/n$, and
\begin{align*}
|b^{(k+1)}|=\frac{|b^{(k)}|\;|\bh|}{|b^{(k)}|+\lambda/n}<|\bh|-\lambda/n=|\bhl{}|.
\end{align*}
Finally, if $|b^{(k)}|=|\bhl{}|$, then $|b^{(k)}|=|\bh|-\lambda/n$, and
\begin{align*}
|b^{(k+1)}|=\frac{|b^{(k)}|\;|\bh|}{|b^{(k)}|+\lambda/n}=|\bh|-\lambda/n=|\bhl{}|,
\end{align*}
which establishes~\ref{one-d-all-same-side}.  To conclude, note that~\ref{one-d-fixed-points} follows immediately from~\ref{one-d-all-nonzero}~and~\ref{one-d-all-same-side}.
\end{proof}

\begin{proof}[%
Proof of Theorem~\ref{thm:one-d-convergence}%
]
Define $c=n|\bh|/\lambda$, and note that~(\ref{lasso-sth}) may be written as
\begin{align*}
\bhl{}=\left[\frac{(c-1)\bh}{c}\right]I(c>1)
=\left[\frac{(c-1)\lambda\sign\bh}{n}\right]I(c>1).
\end{align*}
Now consider three cases.

Case~1: Suppose that $n|\bh|<\lambda$.  Then $c<1$, and $\bhl{}=0$.  By Lemma~\ref{lem:one-d-form},
\begin{align*}
\left|b^{(k)}-\bhl{}\right|=\left|b^{(k)}\right|\le\frac{c^k\lambda|b^{(0)}|}{\lambda}=\left(\dfrac{n|\bh|}{\lambda}\right)^k\left|b^{(0)}-\bhl{}\right|.
\end{align*}

Case~2: Suppose that $n|\bh|=\lambda$.  Then $c=1$, and $\bhl{}=0$.  By Lemma~\ref{lem:one-d-form},
\begin{align*}
\left|b^{(k)}-\bhl{}\right|=\left|b^{(k)}\right|\le
\frac{\lambda}{\lambda|b^{(0)}|^{-1}+nk}\le\frac{\lambda}{nk}.
\end{align*}

Case~3: Suppose that $n|\bh|>\lambda$.  Then $c>1$, and $\bhl{}\ne0$.  Now write
\begin{align}
\left|b^{(k)}-\bhl{}\right|&=\left|\frac{c^k\lambda\sign\bh}{\lambda|b^{(0)}|^{-1}+n\sum_{m=0}^{k-1}c^m}-\frac{(c-1)\lambda\sign\bh}{n}\right|\notag\\
&=\left|\frac{c^k}{\lambda|b^{(0)}|^{-1}+n\sum_{m=0}^{k-1}c^m}-\frac{c-1}{n}\right|\lambda
\label{one-d-rate-step}
\end{align}
Now note that
\begin{align*}
\sum_{m=0}^{k-1}c^m=c^{k-1}\left(\sum_{m=0}^{\infty}c^{-m}-\sum_{m=k}^{\infty}c^{-m}\right)=c^{k-1}\left(\frac{1-c^{-k}}{1-c^{-1}}\right)=\frac{c^k-1}{c-1},
\end{align*}
and insert this into~(\ref{one-d-rate-step}) to obtain
\begin{align*}
\left|b^{(k)}-\bhl{}\right|
&=\left|\frac{c^k}{\lambda|b^{(0)}|^{-1}+n(c-1)^{-1}(c^k-1)}-\frac{c-1}{n}\right|\lambda\\
&=\left|\frac{nc^k-(c-1)\lambda|b^{(0)}|^{-1}-n(c^k-1)}{\lambda|b^{(0)}|^{-1}+n(c-1)^{-1}(c^k-1)}\right|\frac{\lambda}{n}\\
&=\left|\frac{|b^{(0)}|-(c-1)\lambda/n}{(c-1)\lambda/n+(c^k-1)|b^{(0)}|}\right|\left(\frac{(c-1)\lambda}{n}\right)\\
&=\left|\frac{|b^{(0)}|-|\bhl{}|}{|\bhl{}|+(c^k-1)|b^{(0)}|}\right|\left|\bhl{}\right|\\
&\le\frac{|\bhl{}|}{|\bhl{}|+(c^k-1)|b^{(0)}|}\left|b^{(0)}-\bhl{}\right|\\
&\le\frac{(c-1)|\bh|}{c(c^k-1)|b^{(0)}|}\left|b^{(0)}-\bhl{}\right|\\
&\le
\left(\frac{1}{c}\right)^k
\left|\dfrac{\bh}{b^{(0)}}\right|\left|b^{(0)}-\bhl{}\right|,
\end{align*}
and the result is immediately obtained by substituting $n|\bh|/\lambda$ for $c$.
(Also note that $1/c<1$ since $c>1$ in this case.)
\end{proof}

\clearpage
\newpage
\section{A Stochastic/Annealing Variant to SLOG}
\label{sec:AppE}

\textcolor{black}{
Rather than focusing on the recursion as $\sigma^2\to0$ given by (7), it is also of interest to consider the case of a strictly positive value of~$\sigma^2$. From a theoretical standpoint, a strictly positive value of~$\sigma^2$ alters the deterministic nature of our algorithm.  If $\sigma^2>0$, then the conditional distributions in each step of the Gibbs sampling cycle no longer collapse to degeneracy at a single point.  In essence, a strictly positive value of $\sigma^2$ corresponds to implementing the usual Gibbs sampling, which will give samples from the posterior.}

\textcolor{black}{
More specifically, with a strictly positive value of~$\sigma^2$, the SLOG recursion (7) becomes
\begin{equation}\label{eq:aSLOG}
\bm b^{(k+1)}=\left(\bm X^T\bm X+\lambda\,\bm U^{(k)}\right)^{-1}\bm X^T\bm y,
\end{equation}
where $\bm U^{(k)}=\text{diag}(U_1^{(k)},\ldots,U_p^{(k)})$ and $U_1^{(k)},\ldots,U_p^{(k)}$ are independent random variables with
\[
U_j^{(k)}\sim\text{InverseGaussian}\left(\frac1{\left|b_j^{(k)}\right|},\;\frac1{\sigma^2}\right)
\]
for each $j\in\{1,\ldots,p\}$.  (This assumes $b_j^{(k)}\ne0$ for all~$j$, but we already know this to be the case with probability~$1$.)  Note that if $\sigma^2$ is small, then
$U_j^{(k)}$ will be close to $|b_j^{(k)}|^{-1}$ (the corresponding value in the original SLOG recursion)
with high probability.}

\textcolor{black}{
Using recursion~\eqref{eq:aSLOG} an ``annealing" version of SLOG can be implemented by sampling the $\bm U^{(k)}$ with a value of $\sigma^2$ that is fixed for iteration $k$ but that decreases as iteration number $k$ increases. For simplicity we will refer to this annealing approach as aSLOG.}

\textcolor{black}{
To compare the performance of aSLOG and rSLOG we conducted a simulation study. In the simulations, aSLOG was implemented using~$\sigma^2$ sequences that started at either 1e-7 or 1e-10 and then decreased by $1\%$ at each iteration. For both aSLOG and rSLOG coefficients were thresholded to zero as soon as they were below 1e-13 in absolute value. Figure~\ref{aSLOG} contains the results of the simulations.}

 \begin{figure}
 \begin{center}
 \includegraphics[scale=0.8]{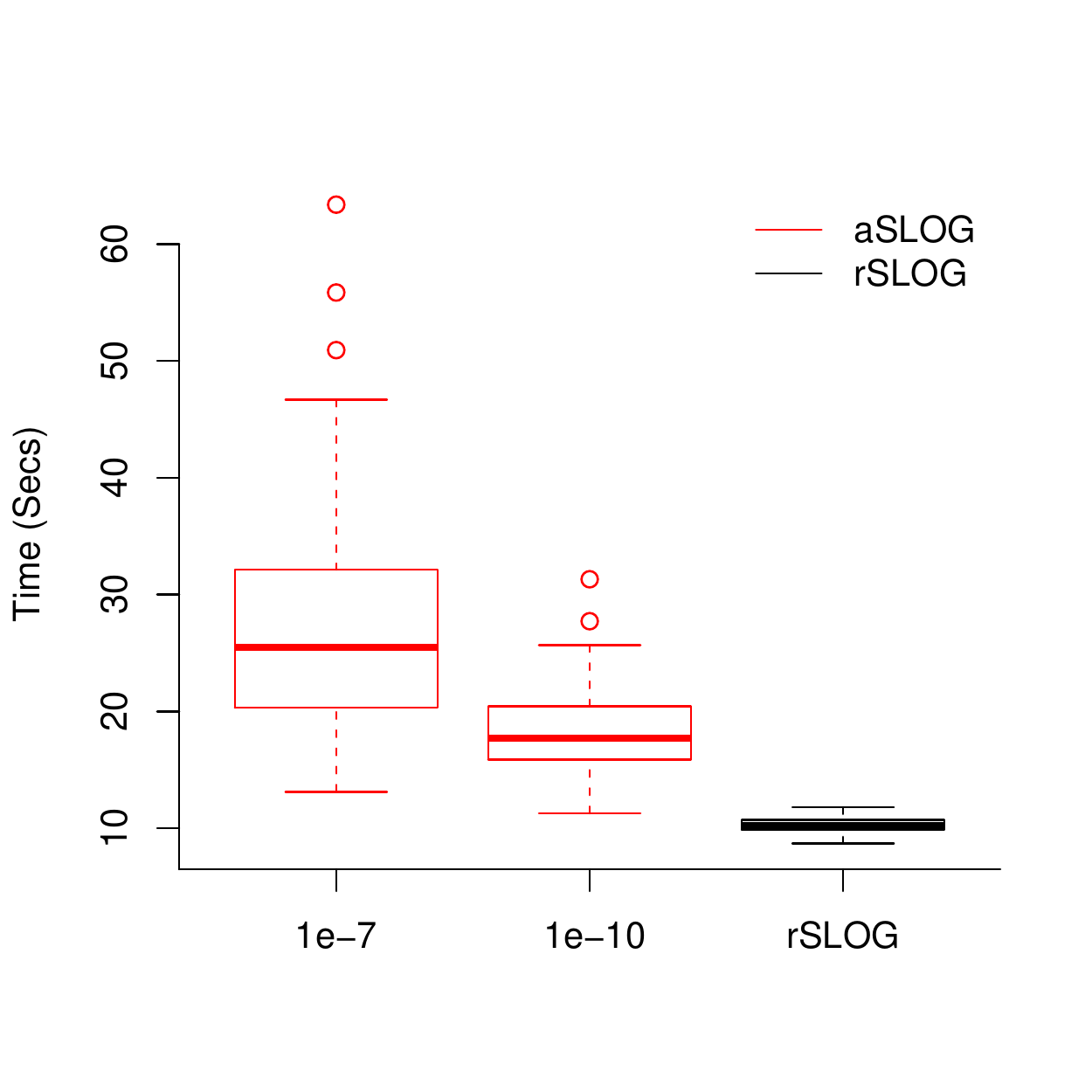}
 \caption{\label{F:aSLOGfig} \textcolor{black}{Timing comparisons between aSLOG and rSLOG for 100 datasets generated using relationship \eqref{Esim1} with $n=200$, $p=1000$, $\rho=0.95$. A sparsity ($s$) of 0.05 was used for $\bm{\beta}^{\lars}$. For aSLOG the starting values of the $\sigma^2$ sequences are as indicated in the x-axis of the plot and decrease by 1\% at each successive Gibbs iteration. \label{aSLOG}}}
 \end{center}
 \end{figure}

\textcolor{black}{
The simulations show that the compute time for rSLOG is less than the compute time for aSLOG. Additionally, the compute time for aSLOG increases as $\sigma^2$ increases. These findings can be explained by the fact that as~$\sigma^2$ increases the posterior is less tightly concentrated around its mode and, in a sense, further away from the lasso solution.
}

\clearpage
\newpage
\section{Comparison of
\color{black}
SLOG and the Algorithm of \citet{lai2013}%
\color{black}
}
\label{sec:AppB}

Here we provide a brief comparison of the SLOG algorithm and the closely related algorithm of \citet{lai2013} (denoted the Lai algorithm), described in Section~\ref{sec:Lai}. As a reminder to the reader, the Lai algorithm utilizes a sequence $\epsilon_k$ to avoid encountering infinite weights (due to division by zero) when performing IRLS. The introduction of the $\epsilon_k$ means that the Lai algorithm is solving an approximate problem at each IRLS iteration. A consequence of solving approximate problems is that if the solution is sparse the Lai algorithm will produce an approximate solution. The closeness of the approximation depends on the sequence $\epsilon_k$ which, in-turn, depends on $\alpha$ and the sparsity estimate $h$, via $\epsilon_k=\min\{\epsilon_{k-1},\,\alpha\, r(\bm\beta_k)_{h+1}\}$. In deriving the SLOG algorithm we show that a sequence $\epsilon_k$ is unnecessary. Importantly, this means that the SLOG algorithm can avoid the additional complexity and errors introduced by the $\epsilon$-approximations of the Lai algorithm.

The SLOG and Lai algorithms were implemented using code written in the statistical package R \citep{statR}. The R code is a slightly modified version of code provided by \citet{lai2013} for implementing the Lai algorithm. The SLOG algorithm was implemented (via the Lai algorithm) by setting $\epsilon_k=0,\;\forall k$ and the starting values for the coefficients to 1e-6. For the Lai algorithm the non-zero $\epsilon_k$ allow the starting values for the coefficients to be set to 0. The two algorithms were run until $d(\SLOG(k))<$1e-6 (or $d(\Lai(k))<$1e-6), where the scaled norm distance $d()$ is defined in Section~\ref{sec:numerical}. The Lai algorithm was implemented over a sequence of $\alpha$ values and with $h$ set to the number of non-zero elements in $\bm{\beta^{\lars}}$ plus ten.
By setting $h$ to this value we are explicitly using knowledge of the sparsity induced by the penalty parameter $\lambda$. It is important to note that in practice such information will not be available and $h$ may not be consistent with the provided $\lambda$. An advantage of the SLOG algorithm is that it does not require a measure of sparsity $h$. The results reported below for the Lai algorithm will change depending on the value of $h$.

\begin{figure}
\subfigure[$\rho=0.1, s=0.90$]{\includegraphics[scale=0.81]{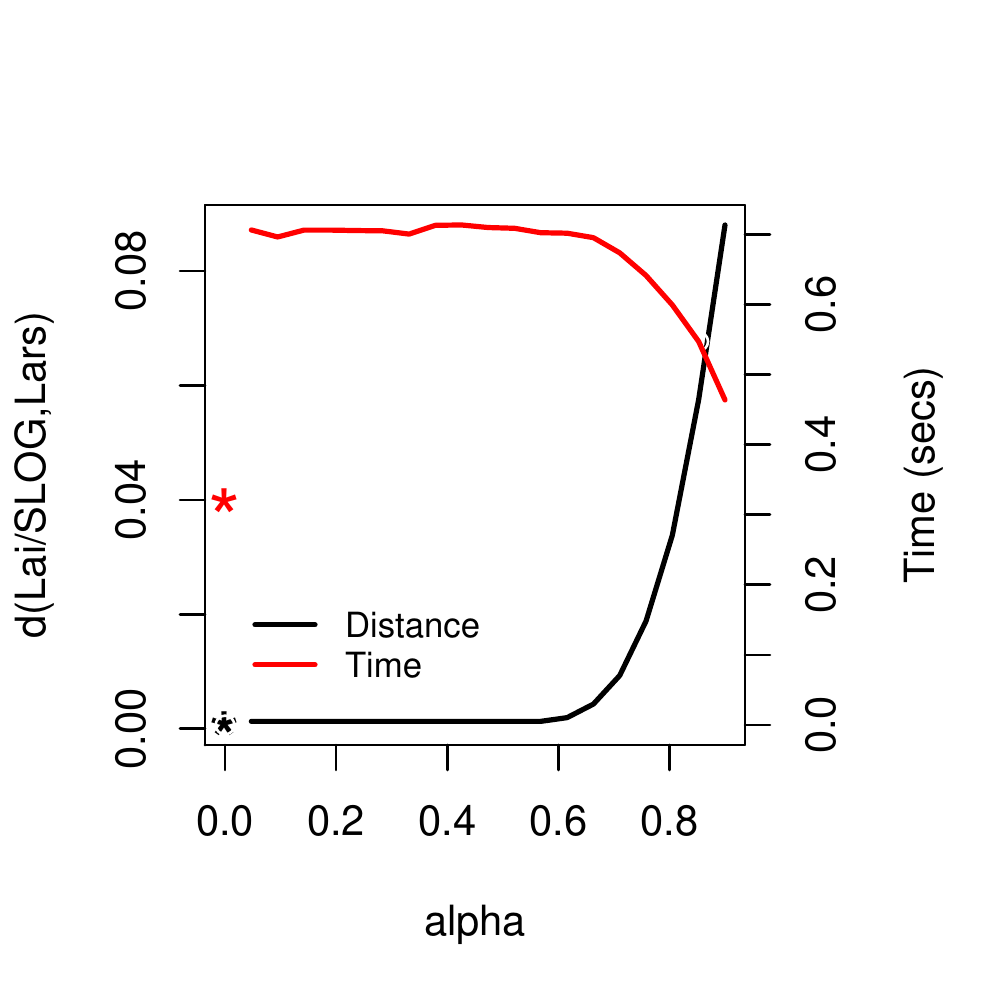} \label{F:lai1}}
\subfigure[$\rho=0.1, s=0.05$]{\includegraphics[scale=0.81]{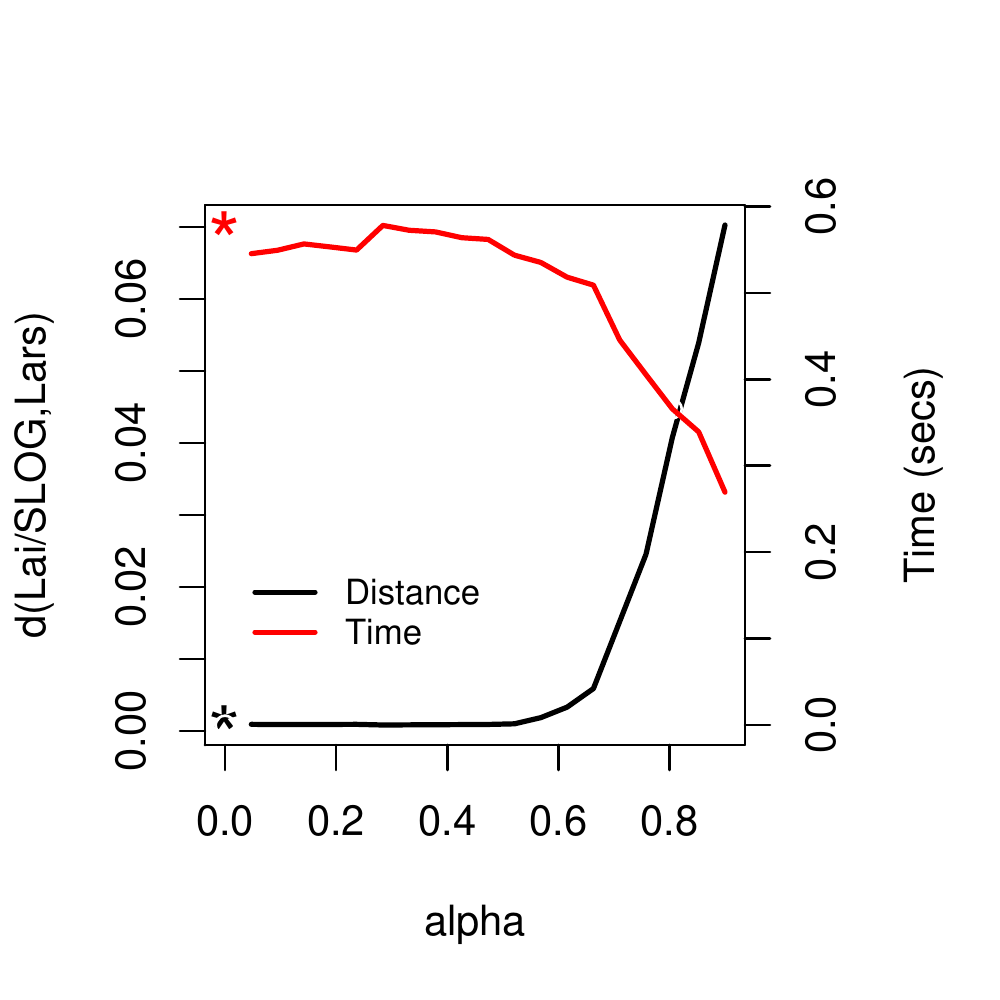} \label{F:lai2}}
\subfigure[$\rho=0.95, s=0.90$]{\includegraphics[scale=0.81]{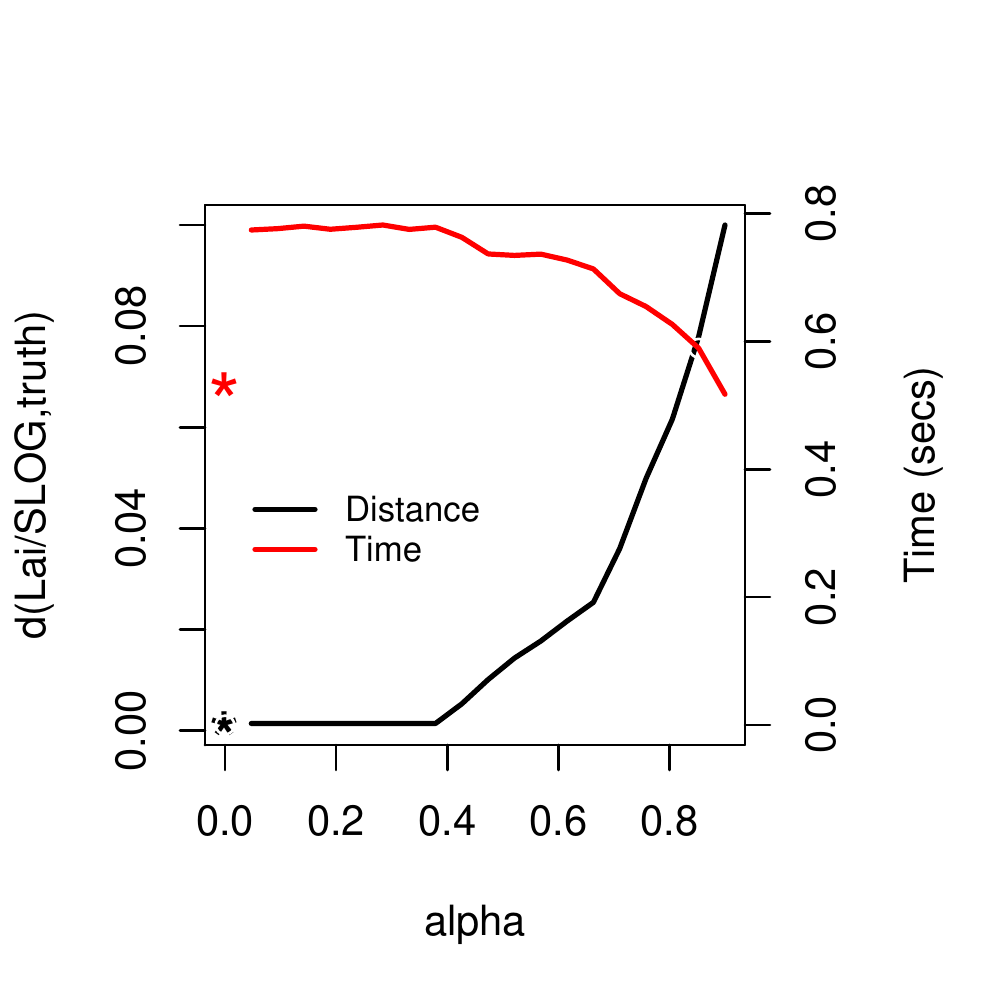} \label{F:lai3}}
\subfigure[$\rho=0.95, s=0.05$]{\includegraphics[scale=0.81]{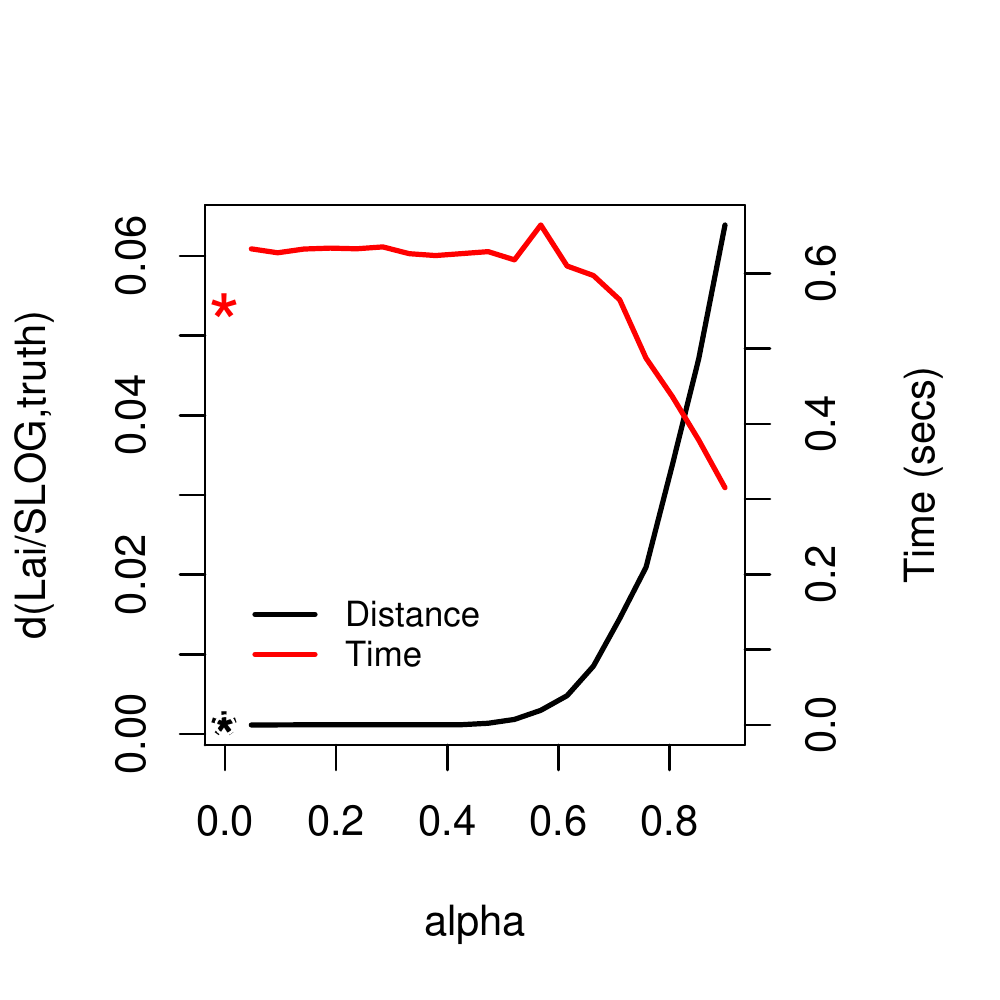} \label{F:lai4}}
\caption{\label{F:lai} Plots of the distance of $\bm{\hat{\beta}}^{\SLOG}$ ($d(\SLOG,\lars)$) and $\bm{\hat{\beta}}^{\Lai}$ ($d(\Lai,\lars)$) from $\bm{\beta^{\lars}}$ and the computation time (seconds) of $\bm{\hat{\beta}}^{\SLOG}$ and $\bm{\hat{\beta}}^{\Lai}$ versus $\alpha$, for data simulated using relationship \eqref{Esim1} with $n=50$, $p=100$ and varying levels of correlation ($\rho$). Varying values for the sparsity ($s$) of $\bm{\beta}^{\lars}$ are considered. In the plots black is used for distance and red for time. The $\ast$ give the distance and timings for the SLOG algorithm. The --- give the distance and timings for the Lai algorithm. The values in the figures are averages over 100 simulations.}
\end{figure}

Figure~$\ref{F:lai}$ contains the results of applying the SLOG and Lai algorithms to data generated using relationship \eqref{Esim1} with $n=50$ and $p=100$. The plots depict the speed of convergence of the two algorithms and the distance of the converged solutions from the exact solution. It is observed that the SLOG algorithm offers the potential for improved performance in settings of high sparsity ($s=0.90$) for both high as well as low multicollinearity. In situations of low sparsity ($s=0.05$) the SLOG and Lai algorithms offer similar performance. Further, similar to rSLOG, it is observed that increasing $\alpha$ in the Lai algorithm offers the potential for computational speed-ups, if the user is willing to sacrifice accuracy. The relative performance of the two algorithms does not appear to be impacted significantly by the multicollinearity of the covariates.

\clearpage
\newpage
\section{A Hybrid Coordinatewise-SLOG Algorithm}
\label{sec:AppD}

\textcolor{black}{
To investigate the hyrbrid approach discussed in Section~\ref{sec:rslog}, we consider an overall covariate space of dimension $n=200$ and $p=1000$. The overall covariate space is constructed of two known blocks of covariates of size $p=500$. We consider two different settings for the pairwise correlation between the covariates within blocks:}
\textcolor{black}{
\begin{enumerate}
\item The covariates in one block are highly multicollinear and in the other block uncorrelated (pairwise correlation of approximately $\rho=0.95$ in one block and $\rho=0$ in the other block).
\item Within each block the covariates are highly multicollinear (pairwise correlation of approximately $\rho=0.95$ in each block).
\end{enumerate}}

\textcolor{black}{
To implement the hybrid approach in setting (1) it would be appropriate to apply rSLOG to the correlated block and CD to the uncorrelated block. Similarly, in setting (2) it would be appropriate to apply rSLOG to both blocks. Additionally, we consider two designs for the correlation between the two blocks of covariates:}
\textcolor{black}{
\begin{enumerate}[label=(\alph*)]
\item An ``idealistic" orthogonal design where the covariates between the two blocks are uncorrelated.
\item A non-orthogonal design where the covariates between the two blocks are mildly correlated.
\end{enumerate}}

\textcolor{black}{
In the case of the orthogonal design, because the coefficient estimates between blocks are independent, fitting rSLOG or CD within blocks is all that is required to implement the hyrbid approach. However, in the non-orthogonal design, due to the dependence of the coefficient estimates between blocks, the estimates obtained from the individual blocks are used as starting values for rSLOG applied to the entire data.  In this last step, estimates that are set to zero within blocks need to be replaced by small non-zero values before being used as starting values. This is necessary to allow for the fact that in the non-orthogonal design coefficients can be incorrectly set to zero within individual blocks.}

\begin{figure}
\hfill
\subfigure[Setting (1)]{\includegraphics[scale=0.8]{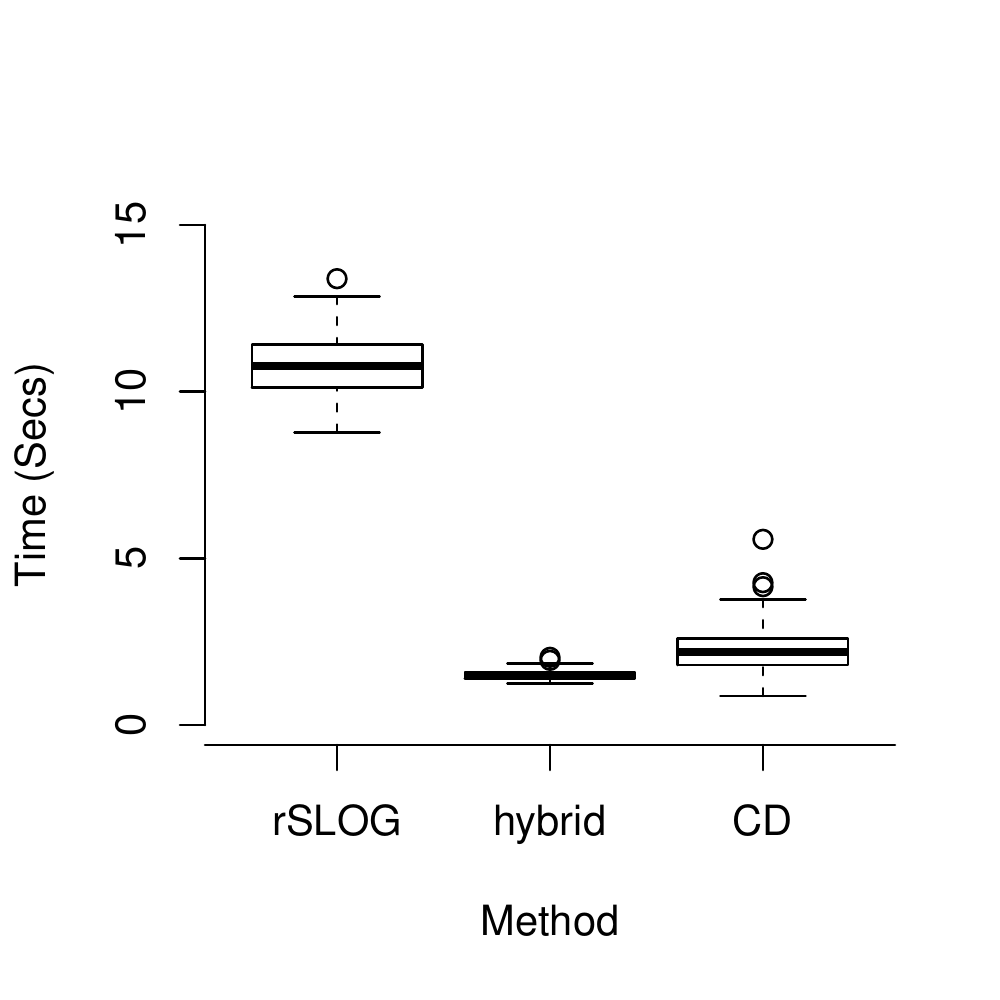} \label{F:start1nc1}}
\hfill
\subfigure[Setting (2)]{\includegraphics[scale=0.8]{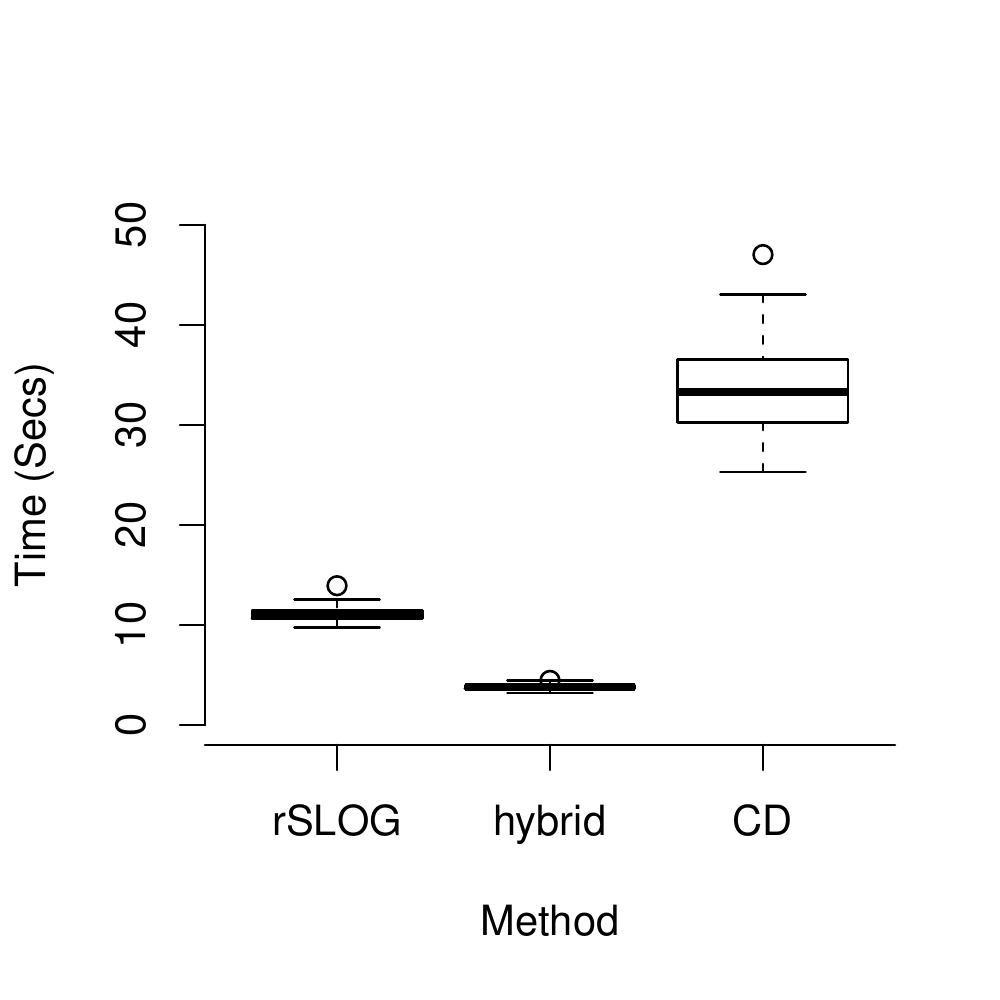} \label{F:start2nc2}}
\hfill
\caption{\label{F:start1} \textcolor{black}{
Comparison of timings for the hybrid approach, rSLOG and CD applied in an orthogonal design setting for simulated datasets of dimension $n=200$ and $p=1000$. The values in the plots are based on 100 datasets simulated using relationship \eqref{Esim1}. A sparsity ($s$) of 0.05 was used for $\bm{\beta}^{\lars}$.}}
\end{figure}

\begin{figure}
\hfill
\subfigure[Setting (1)]{\includegraphics[scale=0.8]{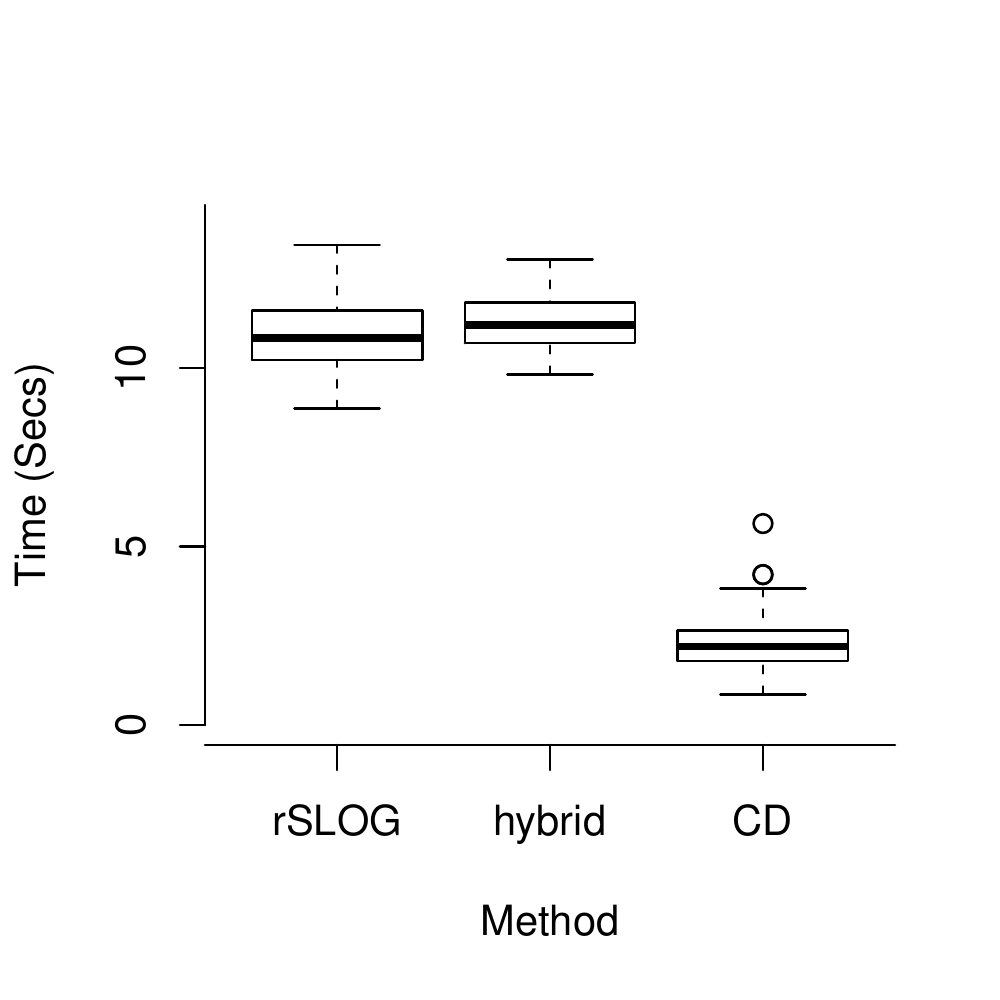} \label{F:start11}}
\hfill
\subfigure[Setting (2)]{\includegraphics[scale=0.8]{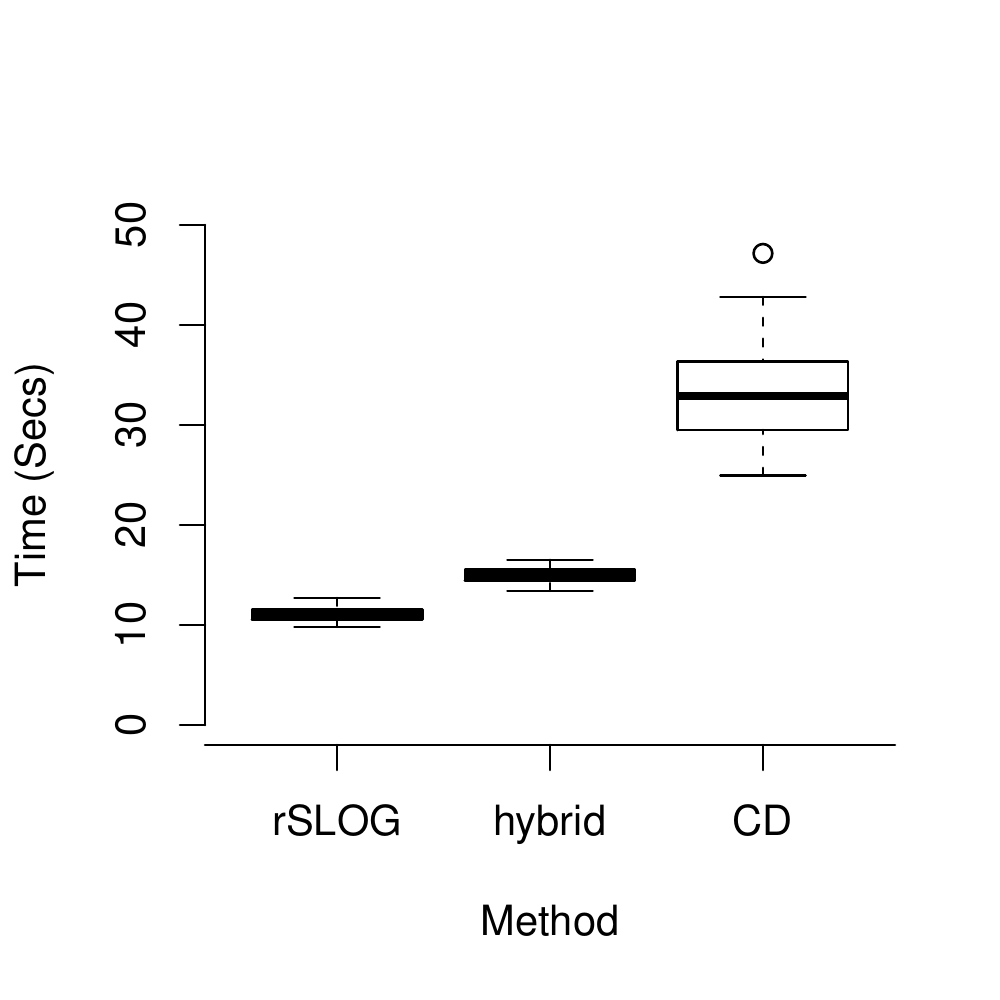} \label{F:start22}}
\hfill
\caption{\label{F:start2} \textcolor{black}{
Comparison of timings for the hybrid approach, rSLOG and CD in a non-orthogonal design setting for simulated datasets of dimension $n=200$ and $p=1000$. The values in the plots are based on 100 datasets simulated using relationship \eqref{Esim1}. A sparsity ($s$) of 0.05 was used for $\bm{\beta}^{\lars}$.}}
\end{figure}

\textcolor{black}{
Figure~\ref{F:start1} and Figure~\ref{F:start2} illustrate the results of applying the hybrid approach, rSLOG and CD in the context of both an orthogonal and non-orthogonal design, respectively.
From Figure~\ref{F:start1} it is evident that the hybrid approach is computationally faster than rSLOG and CD for both of the orthogonal design settings. In the non-orthogonal design settings of Figure~\ref{F:start2} the hyrbid approach is outperformed by CD in setting (1) and rSLOG in setting (2). The primary reason for the slower compute time for the hybrid approach in the non-orthogonal design setting is the need to combine the estimates from the individual blocks.}

\textcolor{black}{
The above simulations illustrate the relative merits of the hybrid approach, rSLOG and CD. In particular, our simulations illustrate: (a) the potential benefits of the hybrid approach in the orthogonal design setting; and (b) the fact that in the non-orthogonal design setting rSLOG is an effective means of overcoming the repeated inversion of a $p \times p$ matrix as required by SLOG. In practice, there are a number of factors that could potentially contribute to the slower compute time for the hybrid approach, compared to rSLOG or CD. These factors include:}
\textcolor{black}{
\begin{enumerate}
\item The need to combine the coefficient estimates obtained from individual blocks within the context of the entire data. The reason for this is that, outside of the orthogonal design setting, the final coefficient estimates obtained within individual blocks will be different to those in the context of the entire data.
\item The fact that, outside of the orthogonal design setting, a given value of the regularization parameter applied to the entire data will yield different coefficient estimates when applied to individual blocks. For example, this fact means that a coefficient estimate of zero within an individual block could be non-zero in the context of the entire data. Importantly, this means that when the coefficient estimates obtained from different blocks are combined, that all $p$ covariates need to be considered, \emph{not} just those with non-zero coefficient estimates.
\item Determining the blocks of covariates: both the number of blocks and how to assign the covariates to blocks.
\end{enumerate}}

\textcolor{black}{
In the non-orthogonal design, the first two of these factors contributed to the observed increase in the computational time of the hybrid approach compared to rSLOG and CD. In practice, the third factor would serve as a pre-screening step before the hybrid approach could be implemented, and hence would further increase the relative compute time of the hybrid approach.}

\clearpage
\newpage
\section{Timing Comparisons in High Sparsity or Low Multicollinearity Settings}
\label{sec:additional}
\textcolor{black}{
In this section we provide some further timing comparisons which illustrate: (1) the relationship between the sample size $n$, sparsity $s$, multicollinearity $\rho$ and the relative compute times of CD and rSLOG; and (2) the fact that rSLOG can offer improved computational speed compared to CD even in high sparsity or low multicollinearity settings.}

\begin{table}[h!]
\begin{small}
\begin{center}
\begin{threeparttable}[!h]
\caption{Ratio of the CD compute time to that of the rSLOG compute time (CD / rSLOG). The reported ratios are averages over 10 simulated datasets. Data are generated using relationship \eqref{Esim1} with $p=1000$ and various values of $s$, $\rho$, and $n$.}
\label{T:furthersims}
\begin{footnotesize}
\begin{tabular}{rrrrrrrrrr}
\hline
\multicolumn{1}{c}{$s$} & \multicolumn{1}{c}{$\rho$} & \multicolumn{1}{c}{$n=300$} & \multicolumn{1}{c}{$n=400$} & \multicolumn{1}{c}{$n=500$} & \multicolumn{1}{c}{$n=600$} & \multicolumn{1}{c}{$n=700$} & \multicolumn{1}{c}{$n=800$} & \multicolumn{1}{c}{$n=900$} & \multicolumn{1}{c}{$n=1000$}\\
\hline
0.70 & 0.30 & 0.00 & 0.00 & 0.01 & 0.01 & 0.02 & 0.02 & 0.03 & 0.05 \\
0.70 & 0.60 & 0.01 & 0.03 & 0.05 & 0.07 & 0.10 & 0.12 & 0.27 & 0.29 \\
0.70 & 0.90 & 0.12 & 0.27 & 0.51 & 0.83 & 0.97 & \textbf{\textcolor{red}{1.05}} & \textbf{\textcolor{red}{1.94}} & \textbf{\textcolor{red}{2.65}} \\
0.40 & 0.30 & 0.01 & 0.03 & 0.08 & 0.09 & 0.13 & 0.26 & 0.24 & 0.32 \\
0.40 & 0.60 & 0.09 & 0.24 & 0.33 & 0.50 & 0.77 & \textbf{\textcolor{red}{1.24}} & \textbf{\textcolor{red}{1.26}} & \textbf{\textcolor{red}{1.86}} \\
0.40 & 0.90 & 0.70 & \textbf{\textcolor{red}{1.24}} & \textbf{\textcolor{red}{2.44}} & \textbf{\textcolor{red}{3.22}} & \textbf{\textcolor{red}{4.52}} & \textbf{\textcolor{red}{7.58}} & \textbf{\textcolor{red}{7.84}} & \textbf{\textcolor{red}{11.48}} \\
0.10 & 0.30 & 0.09 & 0.17 & 0.31 & 0.52 & 0.80 & \textbf{\textcolor{red}{1.07}} & \textbf{\textcolor{red}{1.23}} & \textbf{\textcolor{red}{4.66}} \\
0.10 & 0.60 & 0.48 & 0.76 & \textbf{\textcolor{red}{1.28}} & \textbf{\textcolor{red}{1.64}} & \textbf{\textcolor{red}{2.53}} & \textbf{\textcolor{red}{5.97}} & \textbf{\textcolor{red}{7.97}} & \textbf{\textcolor{red}{12.29}} \\
0.10 & 0.90 & \textbf{\textcolor{red}{1.78}} & \textbf{\textcolor{red}{4.80}} & \textbf{\textcolor{red}{7.01}} & \textbf{\textcolor{red}{9.81}} & \textbf{\textcolor{red}{14.71}} & \textbf{\textcolor{red}{22.73}} & \textbf{\textcolor{red}{33.27}} & \textbf{\textcolor{red}{53.34}} \\
  \hline
\hline
\end{tabular}
\end{footnotesize}
\begin{scriptsize}
\begin{tablenotes}[flushleft]
\item Note: glmnet convergence criterion is 1e-13 and $\theta=1e-18$.
\end{tablenotes}
\end{scriptsize}
\end{threeparttable}
\end{center}
\end{small}
\end{table}

\textcolor{black}{
Table~\ref{T:furthersims} contains timing comparisons between CD and rSLOG for $p=1000$ and varying values of $n$, $s$ and $\rho$. We now proceed to understand further the role of $n$, $s$ and $\rho$ in the relative compute times. The results suggest the relative compute time of rSLOG to that of CD improves}

\begin{itemize}
\item as $\rho$ \emph{increases} with $n$ and $s$ fixed (rSLOG performs better in the presence of higher multicollinearity);
\item as $s$ \emph{decreases} with $n$ and $\rho$ fixed (rSLOG performs better when the model is less sparse);
\item as $n$ \emph{increases} with $\rho$ and $s$ fixed (rSLOG performs better for higher sample sizes).
\end{itemize}

\textcolor{black}{
The results of this section demonstrate that rSLOG can offer improvements in compute times compared to CD even in high sparsity or low multicollinearity settings. For example, in Table~\ref{T:furthersims} there are settings with $s=0.70$ (high sparsity) or $\rho=0.30$ (low multicollinearity) where rSLOG is observed to offer improved computational speeds compared to CD. This is an important observation that highlights the utility of rSLOG outside of low sparsity and high multicollinearity settings.}
\\ \\
\textcolor{black}{
\noindent {\it Remarks:}
\begin{enumerate}
\item The values reported in Table~\ref{T:furthersims} are based on only 10 simulations because CD takes an extremely long time to converge in some of the settings considered. For example, in the setting $s=0.10$ and $\rho=0.90$ CD takes approximately one hour to converge.
\item To be consistent with the other simulations reported in the paper, the values reported in Table~\ref{T:furthersims} are based on a glmnet convergence criterion of 1e-13. Different values of this criterion will result in different relative compute times for CD versus rSLOG. In particular, reducing the convergence criteria for both CD and rSLOG results in larger decreases in compute times for CD compared to rSLOG. Stricter convergence criteria favour rSlOG over CD. For sufficiently ``lax" convergence criteria, such as 1e-7, the relative compute times of CD versus rSLOG reported in Table~\ref{T:furthersims} are all less than one.
\end{enumerate}
}

\clearpage
\newpage
\section{Optimal Level of Sparsity for the Infrared Spectroscopy Data}
\label{sec:AppC}

\textcolor{black}{
Figure~$\ref{F:cookieSpars}$ contains the 10-fold cross-validated mean squared prediction error versus the level of sparsity ($s$) for the infrared spectroscopy data of Section~\ref{sec:ISdata}. The figure illustrates that the optimal $s$ is close to zero and that the prediction error decreases as $s$ decreases from 1 to 0.}

\begin{figure}[h]
 \begin{center}
 \includegraphics[scale=0.8]{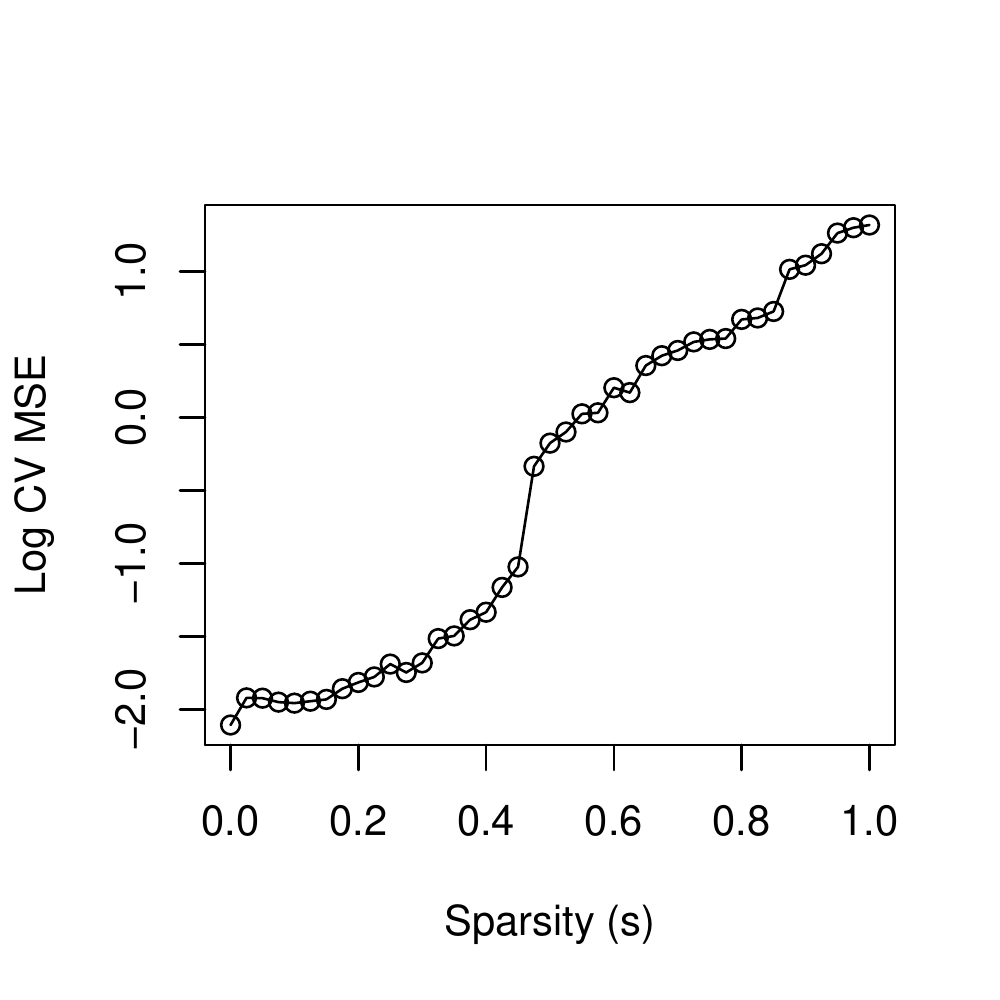}
 \caption{\label{F:cookieSpars} Cross-validated mean-squared-prediction-error versus level of sparsity ($s$) for the Infrared Spectroscopy Data of Section~\ref{sec:ISdata}.}
 \end{center}
 \end{figure}

\clearpage
\newpage
\section{R Code for Running SLOG/rSLOG}

\begin{verbatim}
SLOG <- function(x,y,l,times,thresh,start=NULL){
  #function for implementing SLOG/rSLOG
  #x: covariate data
  #y: response data
  #l: value of lasso regularizaton parameter
  #times: convergence criteria - difference between successive coefficient vectors
  #thresh: below this value estimates are set to 0 (runs rSLOG)
  #start: allows starting values other than sign(xty)*l/p to be specified.

  xtx<-crossprod(x)
  xty<-crossprod(x,y)
  p<-length(xty)
  n<-length(y)
  b.cur <- sign(xty)*l/p
  if(!is.null(start))   b.cur <- start
  b.old<-b.cur
  vin<-1:p
  temp<-rep(0,p)

  conv=FALSE
  k<-1

  while(conv==FALSE){

    b <- b.cur[vin]
    p <- length(b)
    B.inv  <- diag(l/abs(b),nrow=p)
    b.new<-tcrossprod(chol2inv(chol(B.inv+xtx[vin,vin])),t(xty[vin]))

    temp[vin]<-as.vector(b.new)
    temp[abs(temp)<=thresh]<-0
    b.new<-temp
    vin<-which(b.new!=0)

    b.cur <- as.vector(b.new)
    conv<-(sqrt(sum((b.cur-b.old)^2))/sqrt(sum(b.old^2)))<times
    b.old<-b.cur
    k<-k+1

  }

  return(c(k,b.cur))
}
\end{verbatim}

\end{document}